\documentclass{LMCS}

\def\doi{8 (1:22) 2012}
\lmcsheading%
{\doi}
{1--43}
{}
{}
{Nov.~\phantom05, 2010}
{Mar.~10, 2012}
{}

\usepackage{amssymb,amsmath,xspace,graphicx,enumerate,dsfont,hyperref}
\usepackage{amsthm}
\usepackage{mathrsfs}

\usepackage{mathrsfs}
\DeclareFontFamily{OT1}{pzc}{}
\DeclareFontShape{OT1}{pzc}{m}{it}%
             {<-> s * [1.130] pzcmi7t}{}
\DeclareMathAlphabet{\mathpzc}{OT1}{pzc}%
                                 {m}{it}

\newif\ifdraft\draftfalse
\newif\ifwords\wordstrue
\newif\ifappendix\appendixtrue

\theoremstyle{definition}
\newtheorem{definition}[thm]{Definition}
\newtheorem{example}[thm]{Example}

\theoremstyle{remark}
\newtheorem{remark}[thm]{Remark}

\theoremstyle{plain}
\newtheorem{theorem}[thm]{Theorem}
\newtheorem{lemma}[thm]{Lemma}
\newtheorem{claim}{Claim}
\newtheorem{corollary}[thm]{Corollary}
\newtheorem{proposition}[thm]{Proposition}
\newtheorem{open}[thm]{Open question}

\theoremstyle{plain}
\newtheorem*{dicksonsl}{Dickson's Lemma}
\newtheorem*{higmansl}{Higman's Lemma}

\newcommand{\msf}[1]{\ensuremath{\mathsf{#1}}\xspace}
\newcommand{\ara}{\msf{ARA}}
\newcommand{\opguess}{\msf{guess}}
\newcommand{\opspread}{\msf{spread}}

\newcommand{\atraconf}{\+T_{\textsc{atra}}}
\newcommand{\atranconf}{\+N_{\textsc{atra}}}
\newcommand{\araconf}{\+C_{\textsc{ara}}}
\newcommand{\araT}{\rightarrowtail}

\newcommand{\lqw}{\precsim}
\newcommand{\lqwS}{\preceq}

\newcommand{\gqwS}{\succeq}

\newcommand{\Threads}{\Delta}
\newcommand{\set}[1]{\{#1\}}

\newcommand{\atra}{\ensuremath{\mathsf{ATRA}}\xspace}

 \newcommand{\ttypea}{\tridown}
 \newcommand{\ttypena}{\bar{\tridown}}
 \newcommand{\ttypeb}{\triright}
 \newcommand{\ttypenb}{\bar{\triright}}

\newcommand{\tridown}{\triangledown}
\newcommand{\triright}{\rhd}
\newcommand{\rightarrowtridown}{\rightarrow_{\tridown}}
\newcommand{\rightarrowtriright}{\rightarrow_{\scriptscriptstyle\triright}}
\newcommand{\rdc}{\msf{rdc}} 
\newcommand{\wqo}{\msf{wqo}} 
\newcommand{\subsets}{\wp}
\newcommand{\subsetsf}{\subsets_{<\infty}}
\newcommand\+[1]{\mathcal{#1}}

\newcommand{\todo}[1]{}
\newcommand{\xpatheps}{\xpath^{\varepsilon}}
\newcommand{\rxpatheps}{\rxpath^{\varepsilon}}

\newcommand{\ie}{\textit{i.e.}}
\newcommand{\cf}{\textit{cf.}}

\newcommand{\drafttext}{DRAFT - ver.~21/6/2010}
\newcommand{\THETITLE}{Alternating register automata~ 
\\ on finite data words and trees}
\newcommand{\THETITLERUNNING}{Alternating register automata on finite data words and trees}

\begin{document}

\newcommand{\aA}{\mathbf{a}} 
\newcommand{\bB}{\mathbf{b}} 
\newcommand{\eE}{\mathbf{e}} 
\newcommand{\dD}{\mathbf{d}} 
\newcommand{\tT}{\mathbf{t}} 
\newcommand{\wW}{\mathbf{w}} 
\newcommand{\mM}{\mathbf{m}} 

\newcommand{\cl}[1]{\ensuremath{\mathpzc{#1}}\xspace}

\newcommand{\anAut}{\cl{A}}
\newcommand{\midd}{\mathrel{\hspace{2.5px}\mid\hspace{2.5px}}}
\newcommand{\conc}{{\cdot}} 
\newcommand{\pos}{\mathsf{POS}}
\newcommand{\tpos}{\mathsf{pos}}
\newcommand{\talph}{\mathsf{alph}}
\newcommand{\prk}{\cl{p}}
\newcommand{\REG}{\textit{REG}}
\newcommand{\Trees}{\textit{Trees}}
\newcommand{\fudgeceb}{\mathrel{\phantom{\coloneqq}}\mathopen{\phantom{\{}}}
\newcommand{\fudgece}{\mathrel{\phantom{\coloneqq}}}
\newcommand{\fudgeeb}{\mathrel{\phantom{=}}\mathopen{\phantom{\{}}}
\newcommand{\xml}{\textsc{xml}\xspace}
\newcommand{\tup}[1]{\langle #1 \rangle}
\newcommand{\dbracket}[1]{[\![ #1 ]\!]}
\newcommand{\abracket}[1]{[\![ #1 ]\!]}
\newcommand{\N}{\mathds{N}} 
\newcommand{\Np}{\N}
\newcommand{\Nz}{\N_0}
\newcommand{\A}{\mathbb{A}}
\newcommand{\B}{\mathbb{B}}
\newcommand{\E}{\mathbb{E}}
\newcommand{\Ealt}{\mathbb{F}}
\newcommand{\Z}{\mathbb{Z}}
\newcommand{\NP}{\textsc{NP}\xspace}
\newcommand{\ptime}{\textsc{PTime}\xspace}
\newcommand{\exptime}{\textsc{ExpTime}\xspace}
\newcommand{\nexptime}{\textsc{NExpTime}\xspace}
\newcommand{\logspace}{\textsc{LogSpace}\xspace}
\newcommand{\pspace}{\textsc{PSpace}\xspace}
\newcommand{\npspace}{\textsc{NPSpace}\xspace}
\newcommand{\apspace}{\textsc{APSpace}\xspace}
\newcommand{\D}{\ensuremath{\mathbb{D}}\xspace}
\newcommand{\eps}{\varepsilon}
\newcommand{\coloneqq}{\mathrel{\mathop:}=}
\newcommand{\Coloneqq}{\mathrel{{\mathop:}{\mathop:}}=}

\newcommand{\type}{\msf{type}}

\newcommand{\ltl}[1]{\ensuremath{{\sf LTL}^\downarrow(#1)}\xspace}
\newcommand{\ltlplain}{\ensuremath{{\sf LTL}^\downarrow}\xspace}
\newcommand{\F}{\ensuremath{{\sf F}}}
\newcommand{\Fs}{\ensuremath{{\sf F_{\!s}}}}
\newcommand{\G}{\ensuremath{{\sf G}}}
\newcommand{\X}{\ensuremath{{\sf X}}}
\newcommand{\Xd}{\ensuremath{{\sf \bar{X}}}}
\newcommand{\U}{\ensuremath{{\sf U}}}
\newcommand{\R}{\Ud}
\newcommand{\Ud}{\ensuremath{{\sf \bar{U}}}}

\newcommand{\subf}{\msf{sub}}
\newcommand{\nnf}{\msf{nnf}}
\newcommand{\nsubf}{\msf{nsub}} 
\newcommand{\psubf}{\msf{psub}} 

\newcommand{\down}{\downarrow}
\newcommand{\dow}{\downarrow}
\newcommand{\rig}{\rightarrow}
\newcommand{\tu}{\uparrow^{\!*}\!}
\newcommand{\tus}{\uparrow^{\!+}\!}
\newcommand{\td}{\downarrow^{\!*}\!}
\newcommand{\tdz}{\downarrow^{\!\gg}\!}
\newcommand{\tds}{\downarrow^{\!+}\!}
\newcommand{\tl}{{}^{*}\!\!\leftarrow}
\newcommand{\tls}{\!\ ^{+}\!\!\!\leftarrow}
\newcommand{\tr}{\rightarrow^{\!*}\!}
\newcommand{\trs}{\rightarrow^{\!+}\!}
\newcommand{\trz}{\twoheadrightarrow^{\!*}\!}
\newcommand{\xpath}{\ensuremath{\mathsf{XPath}}\xspace}
\newcommand{\rxpath}{\ensuremath{\mathsf{regXPath}}\xspace}
\newcommand{\xp}[1]{{\sf XPath(#1,=)\xspace}}
\newcommand{\xpd}{\xp{\td}}
\newcommand{\xpds}{\xp{\tds}}
\newcommand{\xprl}{\xp{\tl,\tr}}
\newcommand{\xprls}{\xp{\tls,\trs}}
\newcommand{\xpr}{\xp{\tr}}
\newcommand{\xprs}{\ensuremath{\xp{\trs}}\xspace}
\newcommand{\rxpdx}{\msf{regXPath}(\downarrow,=)}
\newcommand{\xpdx}{\xp{\td,\downarrow}}


\title[\THETITLERUNNING]{\THETITLE
\ifdraft\\(\drafttext)\fi\rsuper*}

\author[D.~Figueira]{Diego Figueira}
\address{INRIA \& ENS Cachan, LSV, France and
University of Edinburgh, UK}
\email{dfigueira@gmail.com}
\thanks{Work supported by the Future and Emerging Technologies (FET) programme within the Seventh Framework Programme for Research of the European Commission, under the FET-Open grant agreement FOX, number FP7-ICT-233599.}


\titlecomment{{\lsuper*}This article extends~\cite{Fig10}, and the results contained here also appeared in~\cite[Chapters~3 and 6]{FigPhD}.}

\begin{abstract}
We study alternating register automata on data words and data trees in relation to  logics. A data word (resp.\ data tree) is a word (resp.\ tree) whose every position carries a label from a finite alphabet and a data value from an infinite domain. We investigate  one-way automata  with alternating control over data words or trees, with one register for storing data and comparing them for equality. This is a continuation of the study started by Demri,  Lazi{\'c} and Jurdzi{\'n}ski.

From the standpoint of register automata models, this work aims at two objectives: (1) simplifying the existent decidability proofs for the emptiness problem for alternating register automata; and (2) exhibiting decidable extensions for these models. 

From the logical perspective, we show that (a) in the case of data words, satisfiability of LTL with one register and quantification over data values is decidable; and (b) the satisfiability problem for the so-called \emph{forward} fragment of \xpath on \xml documents is decidable, even in the presence of DTDs and even of key constraints. The decidability is obtained through a reduction to the automata model introduced. This fragment contains the child, descendant, next-sibling and following-sibling axes, as well as data equality and inequality tests.
\end{abstract}

\maketitle


\section{Introduction}
\label{sec:introduction}

In static analysis of databases as in software verification, we frequently find the need of reasoning with \emph{infinite alphabets}. In program verification one may need to decide statically whether a program satisfies some given specification; and the necessity of dealing with infinite alphabets can arise from different angles.
For example, in the presence of concurrency, suppose that an unbounded number of processes run, each one with its process identification, and we must deal with properties expressing the interplay between these processes. Further, procedures may have parameters and data from some infinite domain could be exchanged as arguments. 
On the other hand, in the databases context, static analysis on \xml and its query languages recurrently needs to take into account not only the labels of the nodes, but also the actual data contained in the attributes. It is hence important to study formalisms to reason with words or trees that can carry elements from some infinite domain.

This work is about decidable alternating register automata models on data words and data trees in relation to logics manipulating data. This is a continuation of the investigation carried out by Demri,  Lazi{\'c} and Jurdzi{\'n}ski \cite{DL-tocl08,JL08}. A non trivial consequence of our contribution on alternating automata is that the satisfiability problem for the forward fragment of \xpath on \xml documents is decidable.

We consider two kinds of data models: data words and data trees. A data word (data tree) is a finite word (unranked finite tree) whose every position carries a pair of elements: a symbol from a finite alphabet and and an element (a \emph{datum}) from an infinite set (the \emph{data domain}). We work on finite structures, and all the results we present are relative to finite words and trees.

Over these two models we consider two formalisms: alternating register automata on the one hand, and logics on the other. Each automata model is related to a logic, in the sense that the satisfiability of the logic can be reduced to the emptiness of the automata model. Both automata models we present (one for data words, the other for data trees) have in essence the same behavior. Let us give a more detailed description of these formalisms.

\subsection*{Automata}
The automata model we define is based on the $\ara$ model (for Alternating Register Automata) of \cite{DL-tocl08} in the case of data words, or the \atra model (for Alternating Tree Register Automata) of \cite{JL08} in the case of data trees. 
\ara are one-way automata with alternating control and one register to store  data values for later comparison. \atra correspond to a natural extension of \ara over data trees. The \atra model can move in two directions: to the leftmost child, and/or to the next sibling to the right. Both models were shown to have a decidable emptiness problem. The proofs of decidability are based on non trivial reductions to a class of decidable counter automata with faulty increments. 

In the present work, decidability of these models is  shown by interpreting the semantics of the automaton in the theory of well-quasi-orderings in terms of a well-structured transition system (see~\cite{FS01}). The object of this alternative proof is twofold. On the one hand, we propose a direct, unified and self-contained proof of the main decidability results of~\cite{DL-tocl08,JL08}. 
Whereas in \cite{DL-tocl08,JL08} decidability results are shown by reduction to a class of faulty counter automata, here we avoid such translation, and show decidability directly interpreting the configurations of the automata in the theory of well-structured transition systems. We stress, however, that the underlying techniques used here are similar to those of \cite{DL-tocl08,JL08}.
 On the other hand, we further generalize these results. 
Our proof can be easily extended to show the decidability of the nonemptiness problem for two powerful extensions. These extensions consist in the following abilities: (a) the automaton can nondeterministically \emph{guess} any data value of the domain and store it in the register; and  (b) it can make a certain kind of universal quantification over the data values seen along the run of the automaton, in particular over the  data values seen so far. We name these extensions \opguess and \opspread respectively. These extensions can be both added to the \ara model or to the \atra model, since our proofs for \ara and \atra emptiness problems share the same core. We call the class of alternating register automata with these extensions as $\ara(\opguess,\opspread)$ in the case of data words, or $\atra(\opguess,\opspread)$ in the case of data trees. We demonstrate that these extensions are also decidable if the data domain is equipped with a  linear order and the automata model is extended accordingly. Further, these models are powerful enough to decide a large fragment of \xpath, and of a temporal logic with registers.

\subsection*{{XP}ath}
This work is originally motivated by the increasing importance of reasoning tasks about \xml documents. An \xml document can be seen as an unranked ordered tree where each node carries a \emph{label} from a finite alphabet and a set of \emph{attributes}, each with an associated datum from some infinite domain. A data tree is a simplification of an \xml document that happens to be equivalent for the problems treated here.


\xpath is arguably the most widely used \xml node selecting language, part of XQuery and XSLT; it is an open standard and a W3C Recommendation \cite{xpath:w3c}. Static analysis on \xml languages is crucial for query optimization tasks, consistency checking of \xml specifications, type checking transformations, or many applications on security. Among the most important problems are those of query equivalence and query containment. 
By answering these questions we can decide at compile time whether the query contains a contradiction, and thus whether the computation of the query on the document can be avoided, or if one query can be safely replaced by a  simpler one. For logics closed under boolean combination, these problems reduce to satisfiability checking, and hence we focus on this problem.
Unfortunately, the satisfiability problem for \xpath with \emph{data tests} is undecidable \cite{GF05} even when the data domain has no structure  (\ie, where the only data relation available is the test for equality or inequality). It is then natural to identify and study decidable expressive fragments. 

In this work we prove that the \emph{forward} fragment of \xpath has a decidable satisfiability problem by a reduction to the nonemptiness problem of $\atra(\opguess,\opspread)$. Let us describe this logic.
Core-XPath \cite{GKP05} is the fragment of \xpath that captures all the navigational behavior of \xpath. It has been well studied and its satisfiability problem is known to be decidable in \exptime in the presence of DTDs~\cite{M04}. We consider an extension of this language with the possibility to make \emph{equality} and \emph{inequality} tests between attributes of \xml elements. This logic is named  Core-Data-XPath in \cite{BMSS09:xml:jacm}, and its satisfiability problem is undecidable \cite{GF05}. Here we address a large fragment of Core-Data-XPath named `forward \xpath', that contains the `child', `descendant', `self-or-descendant', `next-sibling', `following-sibling', and `self-or-following-sibling' axes. For economy of space we refer to these axes as $\down$, $\tds$, $\td$, $\rightarrow$, $\trs$, $\tr$ respectively. Note that $\trs$ and $\tr$ are interdefinable in the presence of $\rightarrow$, and similarly with $\tds$ and $\td$. We then refer to this fragment as $\xpath(\down,\td,\rightarrow,\tr,=)$, where `$=$' is to indicate that the logic can express equality or inequality tests between data values.

Although our automata model does not capture forward-\xpath in terms of expressiveness, we show that there is a  reduction to the nonemptiness problem of $\atra(\opguess, \opspread)$. These automata can recognize any regular language, in particular a DTD, a Relax NG document type, or the core of XML Schema (stripped of functional dependencies). Since we show that $\xpath(\down,\td,\rightarrow,\tr,=)$ can express unary key constraints, it then follows that satisfiability of forward-\xpath in the presence of DTDs and key constraints is decidable. This settles a natural question left open in \cite{JL08}, where decidability for a restriction of forward-\xpath was shown. The fragment treated in the cited work is restricted to data tests of the form $\tup{\varepsilon = \alpha}$ (or $\tup{\varepsilon \neq \alpha}$), that is, formul{\ae} that test whether there exists an element accessible via the $\alpha$-relation with the same (resp.\ different) data value as the \emph{current} node of evaluation. However, the forward fragment allows unrestricted tests $\tup{\alpha=\beta}$ and makes the coding into a register automata model highly non trivial.
 As a consequence, we also answer positively the open question raised in~\cite{BFG08} on whether the downward fragment of \xpath in the presence of DTDs is decidable.\footnote{The satisfiability problem on downward \xpath but in the \emph{absence} of DTDs is shown to be \exptime-complete in~\cite{F09}.}

\subsection*{Temporal logics}
$\ara(\opguess,\opspread)$ over data words also yield new decidability results on the satisfiability for some extensions of the temporal logic with one register denoted by $\ltl{\U,\X}$ in \cite{DL-tocl08}. This logic contains a `freeze' operator to store the current datum and a `test' operator to test the current datum against the stored one. Our automata model captures an extension of this logic with quantification over data values, where we can express ``\textsl{for all data values in the past, $\varphi$ holds}'', or ``\textsl{there exists a data value in the future where $\varphi$ holds}''. Indeed, none of these two types of properties can be expressed in the previous formalisms of \cite{DL-tocl08} and \cite{JL08}. These quantifiers may be added to $\ltl{\U,\X}$ over data words without losing decidability. What is more, decidability is preserved if the data domain is equipped with a  linear order that is accessible by the logic. Also, by a translation into $\atra(\opguess,\opspread)$, these operators can be added to a \msf{CTL} version of this logic over data trees, or to the $\mu$-calculus treated in \cite{JL07}\footnote{This is the conference version of \cite{JL08}.}.
However, adding the \emph{dual} of either of these operators results in an undecidable logic.

\subsection*{Contribution}
From the standpoint of register automata models, our contributions can be summarized as follows.
\begin{iteMize}{$\bullet$}
\item We exhibit a unified framework to show decidability for the emptiness problem for alternating register automata. This proof has the advantage of working both for data words and data trees practically unchanged. It is also a simplification of the existing decidability proofs.
\item We exhibit decidable extensions for these models of automata. These extensions work for automata either running over words or trees.  For each of these models there are consequences on the satisfiability of some expressive logics.
\end{iteMize}
From the standpoint of logics, we show the following results.
\begin{iteMize}{$\bullet$}
\item In the case of data trees, we show that the satisfiability problem for the `forward' fragment of \xpath with data test equalities and inequalities is decidable, even in the presence of DTDs (or any regular language) and unary key constraints.    Decidability is shown through a reduction to the emptiness problem for $\atra(\opguess,\opspread)$ on data trees.
\item We show that the temporal logic $\ltl{\U,\X}$  for data words  extended with some quantification over data values is decidable. This result is established thanks to a translation from formul{\ae} to alternating register automata of $\ara(\opguess,\opspread)$.
\end{iteMize}
\newpage

\subsection*{Related work}
The main results presented here first appeared in the conference paper \cite{Fig10}.
Here we include the full proofs and the analysis for alternating register automata on \emph{data words} (something that was out of the scope of \cite{Fig10}) as well as its relation to temporal logics. Also, we show how to proceed in the presence of a  linear order, maintaining decidability. The results contained here also appear in the thesis~\cite[Ch.~3 and 6]{FigPhD}.


By the lower bounds given in \cite{DL-tocl08}, the complexity of the problems we treat in the work is very high: non-primitive-recursive. This lower bound is also known to hold for the logics we treat here $\ltl{\U,\X}$ and forward-$\xpath$. In fact, even very simple fragments of ${\sf LTL}^\downarrow$ and $\xpath$ are known to have non-primitive-recursive lower bounds, including the fragment of \cite{JL08} or even much simpler ones without the one-step `$\rightarrow$' axis,  as shown in \cite{FS09}. This is the reason why in this work we limit ourselves to  decidability / undecidability results.

The work in \cite{BFG08} investigates the satisfiability problem for many \xpath logics, mostly fragments without negation or without data equality tests in the absence of sibling axes. Also, in~\cite{F09} there is a  study of the satisfiability problem for \emph{downward} \xpath fragments  with and without data equality tests. All these are sub-fragments of forward-\xpath, and notably, none of these works considers \emph{horizontal} axes to navigate between siblings. Hence, by exploiting the bisimulation invariance property enjoyed by these logics, the complexity of the satisfiability problem is kept relatively low (at most \exptime) in the presence of data values. However, as already mentioned, when horizontal axes are present most of the problems have a non-primitive-recursive complexity.
In~\cite{GF05}, several  fragments with horizontal axes are treated. The only fragment with data tests and negation studied there is incomparable with the forward fragment, and it is shown to be undecidable. In \cite{FS10} it is shown that the \emph{vertical} fragment of \xpath is decidable in non-primitive-recursive time. This is the fragment with both \emph{downward} and \emph{upward} axes, but notably no horizontal axes (in fact, adding horizontal axes to the vertical fragment results in undecidability). 

The satisfiability of first-order logic with two variables and data equality tests is explored in \cite{BMSS09:xml:jacm}. It is shown that ${\sf FO}^2$ with local one-step relations to move around the data tree and a data equality test relation is decidable. \cite{BMSS09:xml:jacm} also shows the decidability of a fragment of $\xpath(\uparrow,\downarrow,\leftarrow,\rightarrow,=)$ with sibling and upward axes. However, the logic is restricted to one-step axes and to data formul{\ae} of the kind $\tup{\varepsilon = \alpha}$ (or $\not=$), while the fragment we treat here cannot move upwards but has transitive axes and unrestricted data tests.

\section{Preliminaries}

\subsection*{Notation}
We first fix some basic notation. Let $\subsets(C)$ denote the set of subsets of $C$, and $\subsetsf(C)$ be the set of \emph{finite} subsets of $C$. Let $\Np=\set{1,2, \dotsc}$ be the set of positive integers, and let $[n] \coloneqq \set{i \mid 1 \leq i \leq  n}$ for any $n \in \Np$. We call $\Nz = \Np \cup \set 0$.
We fix once and for all $\D$ to be any infinite domain of data values; for simplicity in our examples we will consider $\D = \Nz$. In general we use letters $\A$, $\B$ for finite alphabets, the letter $\D$ for an infinite alphabet and the letters $\E$ and $\Ealt$ for any kind of alphabet. By $\E^*$ we denote the set of finite sequences over $\E$ and by $\E^\omega$ the set of infinite sequences over $\E$. We use `$\conc$' as the concatenation operator between sequences. We write $| S |$ to denote the length of $S$ (if $S$ is a sequence), or the cardinality of $S$ (if $S$ is a set).

\subsection{Finite words}

\newcommand{\Words}{\mathit{Words}}
We consider a finite word over  $\E$ as a function $\wW : [n] \to \E$ for some $n \in \Np$, and we define the set of words as $\Words(\E) \coloneqq \set{[n] \to \E \mid n \in \Np}$. 
We write $\tpos(\wW) = \set{1, \dotsc, n}$ to denote the set of \textbf{positions} (that is, the domain of $\wW$), and we define the size of $\wW$ as $|\wW| = n$. Given $\wW \in \Words(\E)$ and $\wW' \in \Words(\Ealt)$ with $\tpos(\wW) = \tpos(\wW')=P$, we write $\wW \otimes \wW' \in \Words(\E \times \Ealt)$ for the word such that $\tpos(\wW \otimes \wW') = P$ and $(\wW \otimes \wW') (x) = (\wW(x),\wW'(x))$. A \textbf{data word} is a word of $\Words(\A \times \D)$, where $\A$ is a finite alphabet of letters and $\D$ is an infinite domain. Note that we define data words as having at least one position. This is done for simplicity of the definition of the formalisms we work with, and the results contained here also extend to possibly-empty data words and trees.
 We define the \textbf{word type} as a function $\type_{\wW} : \tpos(\wW) \to \set{ \ttypeb, \ttypenb}$ that
specifies whether a position has a next element or not. That is, $\type_{\wW}(i) = \ttypeb$ if{f} $(i+1) \in \tpos(\wW)$. 

\subsection{Unranked ordered finite trees}
\newcommand{\data}{\mathit{data}}
\newcommand{\ancestor}{\preceq}
\newcommand{\ancestorstr}{\prec}
\newcommand{\descendant}{\succeq}
\newcommand{\descendantstr}{\succ}
We define $\Trees(\E)$, the set of finite, ordered and unranked trees over an alphabet $\E$.
 A \textbf{position} in the context of a tree is an element of $\Np^*$. The root's position is the empty string and we note it `$\epsilon$'. The position of any other node in the tree is the concatenation of the position of its parent and the node's index in the ordered list of siblings. Along this work we use $x, y, z, w, v$ as variables for positions, while  $i,j,k,l,m,n$ as variables for numbers. Thus, for example $x \conc i$ is a position which is not the root, and that has $x$ as parent position, and there are $i-1$ siblings to the left of $x \conc i$.

\newcommand{\pto}{\rightharpoonup}
Formally, we define $\pos \subseteq \subsetsf(\Np^*)$ the set of sets of finite tree positions, such that:
$X \in \pos$ if{f} (a) $X \subseteq \Np^*, |X| <
\infty$; (b) it is prefix-closed; and (c) if $n\conc (i+1) \in X$ for $i\in\Np$, then $n \conc i \in X$. 
A tree is a mapping from a set of positions to letters of the alphabet 
\[\Trees(\E) \coloneqq \set{ \tT : P \to \E \mid P \in \pos}\ .\]
Given a tree $\tT \in \Trees(\E)$, $\tpos(\tT)$ denotes the domain of $\tT$, which consists of the set of positions of the tree, and $\talph(\tT)=\E$ denotes the alphabet of the tree. From now on, we informally write `node' to denote a position $x$ together with the value $\tT(x)$.
We define the \emph{ancestor} partial order $\ancestor$ as the prefix relation $x \ancestor x \conc y$ for every $x\conc y$, and the strict version $\ancestorstr$ as the strict prefix relation $x \ancestorstr x \conc y$ for $|y|>0$.
Given a tree $\tT$ and $x\in \tpos(\tT)$, `$\tT|_x$' denotes the subtree of $\tT$ at position $x$. That is, $\tT|_x : \set{ y \mid x \conc y \in \tpos(\tT) } \to \talph(\tT)$ where $\tT|_x(y) = \tT(x \conc y)$. In the context of a tree $\tT$, a \textbf{siblinghood} is a maximal sequence of siblings. That is, a sequence of positions $x \conc 1 , \dotsc, x \conc l \in \tpos(\tT)$ such that $x \conc (l+1) \not\in\tpos(\tT)$.

Given two trees $\tT_1 \in \Trees(\E)$, $\tT_2 \in \Trees(\Ealt)$ such that $\tpos(\tT_1)=\tpos(\tT_2)=P$, we define $\tT_1 \otimes \tT_2 : P \to (\E {\times} \Ealt)$ as $(\tT_1 \otimes \tT_2) (x) = (\tT_1(x),\tT_2(x))$.

The set of \textbf{data trees} over a finite alphabet $\A$ and an infinite domain $\D$  is defined as $\Trees(\A {\times} \D)$. Note that every tree $\tT \in \Trees(\A {\times} \D)$ can be decomposed into two trees $\aA \in \Trees(\A)$ and $\dD \in \Trees(\D)$ such that $\tT = \aA \otimes \dD$. Figure~\ref{fig:data-tree} shows an example of a data tree.
We define the \textbf{tree type} as a function $\type_{\tT} : \tpos(\tT) \to \set{\ttypea, \ttypena} \times \set{ \ttypeb, \ttypenb}$ that
specifies whether a node has children and/or siblings to the right. That is, $\type_{\tT}(x) \coloneqq (a,b)$ where  $a = \tridown$ if{f} $x\conc 1 \in \tpos(\tT)$, and where $b = \triright$ if{f} $x = x'\conc i$ and $x' \conc (i+1) \in \tpos(\tT)$. 
\begin{figure}
  \centering
    \includegraphics[scale=0.6]{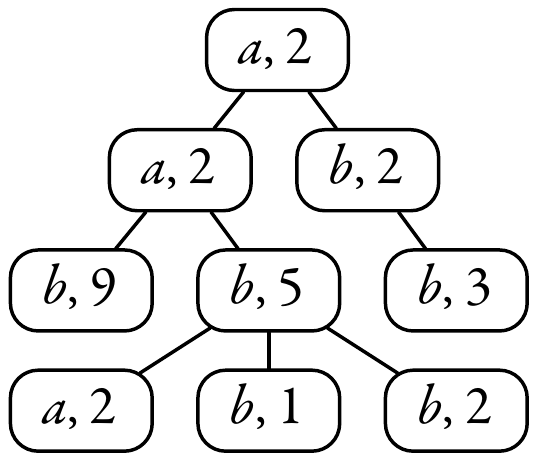}
  \caption{A data tree.}
\label{fig:data-tree}
\end{figure}
The notation for the set of data values used in a data tree is
\[\data(\aA \otimes\dD) \coloneqq \set{\dD(x) \mid x \in \tpos(\dD)} \ .\]
We abuse notation and write $\data(X)$ to refer to all the elements of $\D$ contained in $X$, for whatever object $X$ may be.

\subsection*{First-child and next-sibling coding}
We will make use of the first-child and next-sibling underlying binary tree of an unranked tree $\tT$. This is the binary tree whose nodes are the same as in $\tT$, a node $x$ is the left child of $y$ if $x$ is the leftmost child of $y$ in $\tT$, and a node $x$ is the right child of $y$ if $x$ is the next sibling to the right of $y$ in $\tT$. We will sometimes refer to this tree by ``the fcns coding of $\tT$''. Let us define the relation between positions $\ancestor_{fcns}$ such that $x \ancestor_{fcns} y$ if $x$ is the ancestor of $y$  in the first-child and next-sibling coding of the tree, that is, if $y$ is reachable from $x$ by traversing the tree through the operations `go to the leftmost child' and `go to the next sibling to the right'.

\subsection{Properties of languages}
\newcommand{\Lang}{\+{L}}
We use the standard definition of \textbf{language}. Given an automaton $\anAut$ over data words (resp.\ over data trees),  let $\Lang(\anAut)$ denote the set of data words (resp.\ data trees) that have an accepting run on $\anAut$. We say that $\Lang(\anAut)$ is the language of words (resp.\ trees) recognized by $\anAut$. We extend this definition to a class of automata $\mathscr{A}$: $\Lang(\mathscr{A}) = \set{\Lang(\anAut) \mid \anAut \in \mathscr{A}}$, obtaining a class of languages.

Equivalently, given a formula $\varphi$ of a logic $\mathscr{L}$ over data words (resp.\ data trees) we denote by $\Lang(\varphi)$ the set of words (resp.\ trees) verified by $\varphi$. This is also extended to a logic $\Lang(\mathscr{L}) = \set{ \Lang(\varphi) \mid \varphi \in \mathscr{L} }$.

We say that a class of automata (resp.\ a logic) $\mathscr{A}$  is \textbf{at least as expressive as}  another class (resp.\ logic) $\mathscr{B}$ if{f} $\Lang(\mathscr{B}) \subseteq \Lang(\mathscr{A})$. If additionally $\Lang(\mathscr{B}) \neq \Lang(\mathscr{A})$ we say that $\mathscr{A}$ is \textbf{more expressive} than $\mathscr{B}$.

\smallskip

We say that a class of automata  $\mathscr{A}$ \textbf{captures} a logic $\mathscr{L}$ if{f} there exists a translation $t : \mathscr{L} \to \mathscr{A}$ such that for every $\varphi \in \mathscr{L}$ and model (\ie, a data tree or a data word) $\mM$, we have that $\mM \models \varphi$ if and only if $\mM \in \Lang(t(\varphi))$.

\subsection{Well-structured transition systems}
\newcommand{\aSet}{S}
\newcommand{\aSetk}{K}
\newcommand{\aSetb}{T}
\newcommand{\aSetc}{U}
\newcommand{\aSetd}{V}
\newcommand{\lqdom}{\leq_{\subsets}}
\newcommand{\lqemb}{\sqsubseteq}

The main argument for the decidability of the emptiness of the automata models studied here is based on the theory of well-quasi-orderings. We interpret the automaton's run as a transition system with some good properties, and this allows us to obtain an effective procedure for the emptiness problem. This is known in the literature as a well-structured transition system (WSTS)~(see~\cite{FS01,ACJT-lics96}). Next, we reproduce some standard definitions and known results that we will make use of.

\begin{definition}
For a set $\aSet$, we define $(\aSet, \leq)$  to be a \textbf{well-quasi-order} (\wqo) if{f} `$\leq$' $\subseteq \aSet \times \aSet$ is a relation that is reflexive, transitive and for every infinite sequence $w_1,w_2,\dotsc \in \aSet^\omega$ there are two indices $i<j$ such that $w_i \leq w_j$.
\end{definition}

\index{Dickson's Lemma}
\newcommand{\dicksonslem}{\hyperlink{dicksonsl}{Dickson's Lemma}\xspace}
\begin{dicksonsl}[\hypertarget{dicksonsl}{}\textup{\cite{dicksonslem}}]
Let ${\leq_k} \subseteq \Nz^k \times \Nz^k$ such that $(x_1, \dotsc, x_k) \leq_k (y_1, \dotsc, y_k)$ if{f} $x_i \leq y_i$ for all $i\in[k]$. For all $k\in\Nz$, $(\Nz^k,\leq_k)$ is a well-quasi-order.
\end{dicksonsl}

\begin{definition}
Given a quasi-order $(\aSet,\leq)$ we define the \textbf{embedding order} as the relation ${\lqemb} \subseteq \aSet^* \times \aSet^*$ such that $x_1 \dotsb x_n \lqemb y_1 \dotsb y_m$ if{f} there exist $1 \leq i_1 < \dotsb < i_n \leq m$ with $x_j \leq y_{i_j}$ for all $j\in[n]$.
\end{definition}

\index{Higman's Lemma}
\newcommand{\higmanslem}{\hyperlink{higmansl}{Higman's Lemma}\xspace}
\begin{higmansl}[\hypertarget{higmansl}{}\textup{\cite{higmanslem}}]
Let $(\aSet,\leq)$ be a well-quasi-order. Let ${\lqemb} \subseteq \aSet^* \times \aSet^*$ be the embedding order over $(\aSet,\leq)$. Then, $(\aSet^*,\lqemb)$ is a well-quasi-order.
\end{higmansl}
\begin{corollary}[of \higmanslem]\label{cor:higman-lemma}
  Let $\aSet$ be a \emph{finite} alphabet. Let ${\lqemb}$ be the subword relation over $\aSet^* \times \aSet^*$ (\ie, $x \lqemb y$ if $x$ is the result of removing some (possibly none) positions from $y$). Then, $(\aSet^*,\lqemb)$ is a well-quasi-order.
\end{corollary}
\begin{proof}
It suffices to realize that $\lqemb$ is indeed the embedding order over $(\aSet,=)$, which is trivially a \wqo since $\aSet$ is finite.
\end{proof}

\newcommand{\sDom}{\cl{A}}
\newcommand{\sDomb}{\cl{B}}
\newcommand{\sDomc}{\cl{C}}
\newcommand{\sDomd}{\cl{D}}
\newcommand{\upclose}{\uparrow\!}
\newcommand{\downclose}{\downarrow\!}
\begin{definition}
Given a transition system $(\aSet,\rightarrow)$, and $\aSetb \subseteq \aSet$ we define $Succ(\aSetb) \coloneqq \set{a'\in\aSet \mid \exists\, a \in \aSetb  \text{ with } a \rightarrow a'}$, and $Succ^*$ to its reflexive-transitive closure. Given a \wqo $(\aSet,\leq)$ and $\aSetb \subseteq \aSet$, we define the \emph{upward closure} of $\aSetb$ as $\upclose \aSetb \coloneqq \set{ a \mid \exists\, a' \in \aSetb, a' \leq a}$, and the \emph{downward closure} as $\downclose \aSetb \coloneqq \set{ a \mid \exists\, a' \in \aSetb, a \leq a'}$. We say that $\aSetb$ is \textbf{downward-closed} (resp.\ \textbf{upward-closed})  if{f} $\downclose \aSetb = \aSetb$ (resp.\ $\upclose \aSetb = \aSetb$).
\end{definition}
\begin{definition}
We say that a transition system $(\aSet,\rightarrow)$ is \textbf{finitely branching} if{f} $Succ(\set a )$ is finite for all $a \in \aSet$. If $Succ(\set a)$ is also computable for all $a$, we say that $(\aSet,\rightarrow)$ is \textbf{effective}.
\end{definition}
\newcommand{\anEquivRel}{\equiv}
\begin{definition}[\rdc]
A transition system  $(\aSet,\rightarrow)$  is \textbf{reflexive downward compatible} (or  \rdc for short) with respect to a \wqo  $(\aSet,\leq)$ if{f} for every $a_1,a_2,a'_1 \in \aSet$ such that $a'_1 \leq a_1$ and $a_1 \rightarrow a_2$, there exists $a'_2 \in \aSet$ such that $a'_2 \leq a_2$ and either $a'_1 \rightarrow a'_2$ or $a'_1 = a'_2$.
\end{definition}

Our forthcoming decidability results on alternating register automata of Sections~\ref{sec:emptiness-problem-ara-guess-spread} and \ref{sec:atrags-emptiness} will be shown as a consequence of the following propositions.
\begin{proposition}[\textup{\cite[Proposition~5.4]{FS01}}]
\label{prop:rdc-computable}\,
  If $(\aSet,\leq)$ is a \emph{\wqo} and $(\aSet,\rightarrow)$ a transition system such that (1) it is \emph{\rdc}, (2) it is effective, and (3) $\leq$ is decidable; then for any finite $\aSetb \subseteq \aSet$ it is possible to compute a finite set $\aSetc \subseteq \aSet$ such that ${\upclose \aSetc} = {\uparrow\! Succ^*(\aSetb)}$.
\end{proposition}

\begin{lemma}\label{lem:downward-closed-decidable}
  Given $(\aSet,\leq,\rightarrow)$ as in Proposition~\ref{prop:rdc-computable}, a recursive downward-closed set $\aSetd \subseteq \aSet$, and a finite set $\aSetb \subseteq \aSet$. The problem of whether there exists $a \in \aSetb$ and $b \in \aSetd$ such that $a \rightarrow^* b$ is decidable.
\end{lemma}
\begin{proof}
Applying Proposition~\ref{prop:rdc-computable}, let $\aSetc \subseteq \aSet$ finite, with ${\upclose \aSetc} = {\upclose Succ^*(\aSetb)}$. Since $\aSetb$ is finite and $\aSetd$ is recursive, we can test for every element of $b \in \aSetc$ if $b \in \aSetd$. On the one hand, if there is one such $b$, then by definition $b \in {\uparrow\! Succ^*(\aSetb)}$, or in other words $a \rightarrow^* a' \leq b$ for some $a \in \aSetb$, $a' \in Succ^*(a)$. But since $\aSetd$ is downward-closed, $a' \in \aSetd$ and hence the property is true. On the other hand, if there is no such $b$ in $\aSetc$, it means that there is no such $b$ in $\upclose\aSetc$ either, as $\aSetd$ is downward-closed. This means that there is no such $b$ in $Succ^*(\aSetb)$ and hence that the property is false.
\end{proof}


\begin{definition}[$\lqdom$]
  Given an ordering $(\aSet,\leq)$, we define the \textbf{majoring ordering} over $\leq$ as $(\subsetsf(\aSet),\lqdom)$, where $\+S \lqdom \+S'$ if{f} for every $a\in\+S$ there is $b \in \+S'$ such that $a \leq b$.
\end{definition}

\begin{proposition}\label{prop:lqdom-wqo}
  If $(\aSet,\leq)$ is a \wqo, then the majoring order over $(\aSet,\leq)$ is a \wqo.
\end{proposition}\newpage

\begin{proof}
The fact that this order is reflexive and transitive is immediate from the fact that $\leq$ is a quasi-order. The fact of being a \emph{well}-quasi-order is a simple consequence of \higmanslem. Each finite set $\set{a_1, \dotsc, a_n}$ can be seen as a sequence of elements $a_1, \dotsc, a_n$, in any order. In this context, the embedding order is stricter than the majoring order. In other words, if $a_1, \dotsc, a_n \lqemb a'_1, \dotsc, a'_m$, then $\set{a_1, \dotsc, a_n} \lqdom \set{a'_1, \dotsc, a'_m}$. By \higmanslem the embedding order over $(\aSet,\leq)$ is a \wqo, implying that the majoring order is as well.
\end{proof}

The following technical Proposition will become useful in Section~\ref{sec:atrags-emptiness} for transferring the decidability results obtained for class of automata $\ara(\opguess,\opspread)$ on data words to the class $\atra(\opguess,\opspread)$ on data trees.
  \begin{proposition}\label{prop:wsts-higman}
    Let $\leq, \rightarrow_1\; \subseteq \aSet \times \aSet$, \,
    $\lqdom,\rightarrow_2\; \subseteq \subsetsf(\aSet) \times \subsetsf(\aSet)$ where $\lqdom$ is the \emph{majoring order} over $(\aSet,\leq)$ and $\rightarrow_2$ is
 such that if    $\+S \rightarrow_2 \+S'$ then:  $\+S= \set{a} \cup \hat{\+S}$, $\+S' =  \set{b_1, \dotsc, b_m}\cup\hat{\+S}$  with $a \rightarrow_1 b_i$ for every $i \in [m]$. Suppose that $(\aSet,\leq)$ is a \wqo which is \rdc with respect to      $\rightarrow_1$. Then, $(\subsetsf(\aSet),\lqdom)$ is a \wqo which is \rdc with respect to $\rightarrow_2$.
  \end{proposition}
  \begin{proof}
    The fact that $(\subsetsf(\aSet),\lqdom)$ is a \wqo is given by Proposition~\ref{prop:lqdom-wqo}. The \rdc property with respect to $\rightarrow_2$ is simple. 
Let
\[\set{a'_1, \dotsc, a'_\ell} \cup \+S' \quad \lqdom \quad  \set{a} \cup \+S \quad \rightarrow_2 \quad \set{b_1, \dotsc, b_m} \cup \+S  \]
with $\ell \geq 0$ and $\+S' \lqdom \+S$, $a'_i \leq a$ for all $i \in [\ell]$, and $a \rightarrow_1 b_i$ for all $i \in [m]$. We show by induction on $\ell$ that there exists $\+S''$ with $\set{a'_1, \dotsc, a'_\ell} \cup \+S' \rightarrow_2^* \+S''$ and $S'' \lqdom \set{b_1, \dotsc, b_m} \cup \+S$.
If $\ell =0$ and $a$ has no pre-image, then $\+S' \lqdom \+S$, and the relation is reflexive compatible since this means that $\+S' \lqdom \set{b_1, \dotsc, b_m} \cup \+S$.

Suppose now that $\ell >0$. Note that $a'_\ell \leq a$ and $\leq$ is \rdc with $\rightarrow_1$.  One possibility is that for each $a \rightarrow_1 b_i$ there is some $b'_i$ such that $a'_\ell \rightarrow_1 b'_i$ and $b'_i \leq b_i$. In this case we obtain 
\begin{align*}
  \set{a'_1, \dotsc, a'_\ell} \cup \+S' &\;\;\rightarrow_2\;\;  \set{a'_1, \dotsc, a'_{\ell-1}}  \;\cup\; \set{b'_1, \dotsc, b'_m} \cup \+S' 
\end{align*}
where
\begin{align*}
\set{a'_1, \dotsc, a'_{\ell-1}}  \;\cup\; (\set{b'_1, \dotsc, b'_m} \cup \+S')
  &\;\;\lqdom\;\;
   \set{a} \;\cup\; (\set{b_1, \dotsc, b_m} \cup  \+S)
\;\;\rightarrow_2\;\; \set{b_1, \dotsc, b_m} \cup \+S
\end{align*}
and $\set{b'_1, \dotsc, b'_m} \cup \+S' \lqdom \set{b_1, \dotsc, b_m} \cup \+S$. We can then apply the inductive hypothesis on $\ell-1$ and obtain $\+S''$ such that $\set{a'_1, \dotsc, a'_{\ell-1}}  \cup \set{b'_1, \dotsc, b'_m} \cup \+S' \rightarrow_2^* \+S''$ and $\+S'' \lqdom \set{b_1, \dotsc, b_m} \cup \+S$. Hence, $\set{a'_1, \dotsc, a'_\ell} \cup \+S' \rightarrow_2^* \+S''$, obtaining the downward compatibility.

The only case left to analyze is when, for some $a'_\ell \leq a \rightarrow_1 b_i$ the compatibility is reflexive, that is, $a'_\ell \leq b_i$. In this case we take a reflexive compatibility as well, since $\set{a'_\ell}\cup\+S' \lqdom \set{b_1, \dotsc, b_m} \cup \+S$. We then apply the inductive hypothesis on $\set{a'_1, \dotsc, a'_{\ell-1}} \;\cup\; \set{a'_\ell}\cup\+S'$ in the same way as before.
\end{proof}

\part{Data words}\label{part:words}

In this first part, we start our study on data words. In Section~\ref{sec:ara-model} we introduce our automata model $\ara(\opguess,\opspread)$, and we show that the emptiness problem for these automata is decidable. This extends the results of \cite{DL-tocl08}. To prove decidability, we adopt a different approach than the one of~\cite{DL-tocl08} that enables us to show the decidability of some extensions, and to simplify the decidability proofs of Part~\ref{part:trees}, which can be seen as a corollary of this proof. In Section~\ref{sec:ltl-with-registers} we introduce a temporal logic with registers and quantification over data values. This logic is shown to have a decidable satisfiability problem by a reduction to the emptiness problem of $\ara(\opguess,\opspread)$ automata.

\medskip

\section{{ARA} model}
\label{sec:ara-model}
\newcommand{\atraT}{\mathrel{\twoheadrightarrow}}
\newcommand{\atraTn}{\rightarrow}
\newcommand{\dd}{{:}}
\newcommand{\opset}{\msf{store}}
\newcommand{\opeq}{\msf{eq}}
\newcommand{\opneq}{\ensuremath{\overline{\msf{eq}}}}

An \textbf{Alternating Register Automaton} (\ara) consists in a one-way automaton on data words with alternation and \emph{one} register to store and test data. In~\cite{DL-tocl08} it
was shown that the emptiness problem is decidable and non-primitive-recursive. Here, we
consider an extension of \ara with two operators: \opspread and \opguess.
 We call this model $\ara(\opspread,\opguess)$.

\begin{definition}\label{def:ara-model}
  An alternating register automaton of $\ara(\opspread,\opguess)$
  is a tuple $\anAut=\tup{\A,Q,q_I,\delta}$ such that
  \begin{iteMize}{$\bullet$}
  \item $\A$ is a finite alphabet;
  \item $Q$ is a finite set of states;
  \item $q_I \in Q$ is the initial state; and
  \item $\delta : Q \to \Phi$ is the transition function,
    where $\Phi$ is defined by the grammar
\begin{multline*}
    a \midd \bar{a} \midd \odot ? \midd \opset(q) \midd \opeq \midd \opneq \midd q
    \land q' \midd q \lor q' \midd \triright q \midd
    \opguess(q) \midd \opspread(q,q')
\end{multline*}
    where $a\in \A, q,q' \in Q,$ $ \odot \in \set{\ttypeb,\ttypenb}$.
  \end{iteMize}
\end{definition}
\noindent This formalism without the \opguess and \opspread transitions is equivalent to the automata model of~\cite{DL-tocl08} on finite data words, where $\triright$ is to move to the next position to the right on the data word, $\opset(q)$ stores the current datum in the register and \opeq{} (resp.\ \opneq) tests that the current node's value is (resp.\ is not) equal to the stored. We call a state $q\in Q$ \textbf{moving} if $\delta(q) = \triright q'$ for some $q' \in Q$. 

As this automaton is one-way, we define its semantics as a set of `threads' for each node that progress synchronously. That is, all threads at a node move one step forward simultaneously and then perform some non-moving transitions independently. This is done for the sake of simplicity of the formalism, which simplifies both the presentation and the decidability proof.

Next we define a \emph{configuration} and then we give a notion of a \emph{run} over a data word $\wW$.  A \textbf{configuration} is a tuple
$\tup{i,\alpha,\gamma,\Threads}$ that describes the partial state of the execution at position $i$. The number $i \in \tpos(\wW)$ is the \emph{position} in the data word $\wW$, $\gamma = \wW(i) \in \A \times \D$ is the current position's letter and datum, and $\alpha = \type_{\wW}(i)$ is the \emph{word type} of the position $i$. Finally, $\Threads \in \subsetsf (Q \times \D)$ is a finite set of active \textbf{threads}, each thread $(q,d)$ consisting in a state $q$ and the value $d$ stored in the register. We will always note the set of threads of a configuration with the symbol $\Threads$, and we write $\Threads(d) = \set{q \in Q \mid (q,d) \in \Threads}$ for $d\in\D$, $\Threads(q) = \set{d \in \D \mid (q,d) \in \Threads}$ for $q \in Q$.  By $\araconf$ we denote the set of all configurations. 
Given a set of threads $\Threads$ we write 
$\data(\Threads) \coloneqq \set{d\mid(q,d) \in \Threads}$, and $\data(\tup{i,\alpha,(a,d),\Threads})\coloneqq\set{d}\cup\data(\Threads)$. We say that a configuration is \textbf{moving} if for every $(q,d) \in \Threads$, $q$ is moving.

\label{moving-nonmoving-rels}
To define a run we first introduce three transition relations over \emph{node configurations}: the \emph{non-moving} relation $\rightarrow_\eps$ and the \emph{moving} relation $\rightarrowtriright$. We start with $\rightarrow_\eps$. If the transition corresponding to a thread is a $\opset(q)$, the automaton \emph{sets} the register with current data value and continues the execution of the thread with state $q$; if it is $\opeq$, the thread \emph{accepts} (and in this case disappears from the configuration) if the current datum is \emph{equal} to that of the register, otherwise the computation for that thread cannot continue. The reader can check that the rest of the cases defined in Figure~\ref{fig:transition-ara-epsilon} follow the intuition of an alternating automaton.
\begin{figure}
  \centering
  \begin{align}
    \rho &\rightarrow_\eps \tup{i,\alpha,(a,d),\set{(q_j,d')}
      \cup \Threads}
    && \text{ if } \delta(q) = q_1 \lor q_2, j \in
    \set{1,2}\label{def:righteps:1}
    \\
    \rho & \rightarrow_\eps
    \tup{i,\alpha,(a,d),\set{(q_1,d'),(q_2,d')} \cup \Threads}
    && \text{ if } \delta(q) = q_1 \land q_2
    \\
    \rho &\rightarrow_\eps \tup{i,\alpha,(a,d),\set{(q',d)}
      \cup \Threads}
    && \text{ if } \delta(q) = \opset(q')\label{def:righteps:6}
    \\
    \rho &\rightarrow_\eps \tup{i,\alpha,(a,d),\Threads}
    && \text{ if } \delta(q) = \opeq \text{ and } d =
    d' \label{def:righteps:3}
    \\
    \rho &\rightarrow_\eps \tup{i,\alpha,(a,d),\Threads}
    && \text{ if } \delta(q) = \opneq \text{ and } d\not=d'
    \\
    \rho &\rightarrow_\eps \tup{i,\alpha,(a,d),\Threads}
    && \text{ if } \delta(q) =\beta ? \text{ and } \beta\in\alpha
    \\
    \rho &\rightarrow_\eps \tup{i,\alpha,(a,d),\Threads}
    && \text{ if } \delta(q) = b \text{ and } b=a
    \\
    \rho &\rightarrow_\eps \tup{i,\alpha,(a,d),\Threads}
    && \text{ if } \delta(q) = \bar{b} \text{ and } b\not= a
  \end{align} 
  \caption{Definition of the transition relation ${\rightarrow_\eps} \subseteq \araconf \times \araconf$, given a configuration $\rho = \tup{i,\alpha,(a,d),\set{(q,d')} \cup \Threads}$.}
  \label{fig:transition-ara-epsilon}
\end{figure}

The cases that follow correspond to our extensions to the model of~\cite{DL-tocl08}. The \opguess instruction extends the model with the ability of storing \emph{any} datum from the domain $\D$. Whenever $\delta(q)=\opguess(q')$ is executed, a data value
(nondeterministically  chosen) is saved in the register.
\begin{align}
    \rho &\rightarrow_\eps \tup{i,\alpha,(a,d),\set{(q',e)}
      \cup \Threads} \quad 
  \text{ if } \delta(q) = \opguess(q'), e \in \D \label{def:righteps:guess}
\end{align}
Note that the $\opset$ instruction may be simulated with the $\opguess$,  $\opeq$ and $\land$ instructions, while $\opguess$ cannot be expressed by the \ara model.

The `\opspread' instruction is an unconventional operator
in the sense that it depends on the data of \emph{all} 
threads in the current configuration with a certain state. Whenever
$\delta(q)=\opspread(q_2,q_1)$ is executed, a new thread with state $q_1$
and datum $d$ is created for each thread $\tup{q_2,d}$ present in the
configuration. With this operator we can code a universal quantification over all the data values that appeared so far (\ie, that appeared in smaller positions). We demand that this transition may only be applied if all other possible $\rightarrow_\eps$ kind of transitions were already executed. In other words, only $\opspread$ transitions or \emph{moving} transitions are present in the configuration.
\begin{gather}
  \rho
        \rightarrow_\eps \tup{i,\alpha,(a,d),\set{\tup{q_1,d}
          \mid {\tup{q_2,d} \in \Threads}} \cup \Threads} \label{def:righteps:spread}
\end{gather}
if{f} $\delta(q) = \opspread(q_2,q_1)$ and for all $\tup{\tilde{q},\tilde{d}} \in \Threads$ either $\delta(\tilde q)$ is a \opspread or a moving instruction.
We also use a weaker one-argument version of $\opspread$. 
\begin{gather}
  \rho
        \rightarrow_\eps \tup{i,\alpha,(a,d),\set{\tup{q_1,d}
          \mid { \exists q_2 . \tup{q_2,d} \in \Threads}} \cup \Threads} \label{def:righteps:spread-simple}
\end{gather}
if{f} $\delta(q) = \opspread(q_1)$ and for all $\tup{\tilde{q},\tilde{d}} \in \Threads$ either $\delta(\tilde q)$ is a \opspread or a moving instruction. Notice that the one-argument version of $\opspread$ can be simulated with the two-argument version, and hence we do not include it in the definition of the automaton.
Also, note that we enforce the $\opspread$ operation  to be executed once all other non-moving transitions have been applied, in order to take into account all the data values that may have been introduced in these transitions, as a result of $\opguess$ operations. This behavior simplifies the reduction from the satisfiability of forward-$\xpath$ we will show later. This reduction will  only need to use the weak one-argument version of $\opspread$.

  The $\rightarrowtriright$ transition advances all threads of the node simultaneously, and is defined, for any type $\alpha' \in \set{\triright,\bar{\triright}}$ and symbol and with data value $\gamma' \in \A \times \D$,
  \begin{gather} 
    \tup{i,\ttypeb,\gamma,\Threads} \rightarrowtriright
    \tup{i+1,\alpha',\gamma',\Threads_\triright}\label{eq:def:ara:ns}
  \end{gather}
  if{f} 
  \begin{enumerate}[(i)]
\item \label{ara:moving:i} for all $(q,d) \in \Threads$, $\delta(q)$ is moving; and
\item \label{ara:moving:ii} $\Threads_\triright = \set{\tup{q',d}
    \mid {(q,d)\in \Threads}, \delta(q)=\triright\,q'}$.
  \end{enumerate}

Finally, we define the transition between configurations as ${\araT} \coloneqq {\rightarrowtriright} \cup {\rightarrow_\eps}$.

A \emph{run}  over a data word $\wW = \aA \otimes \dD$ is a nonempty  sequence  $\cl{C}_1 \araT \dotsb \araT \cl{C}_n$  with $\cl{C}_1= \tup{1,\alpha_0,\gamma_0,\Threads_0}$ and $\Threads_0 = \set{\tup{q_I,\dD(1)}}$ (\ie, the thread consisting in the initial state with the first datum), such that for every $j \in [n]$ with $\cl{C}_j = \tup{i,\alpha,\gamma,\Threads} $: (1) $i \in \tpos(\wW)$; (2) $\gamma=\wW(i)$; and (3) $\alpha = \type_{\wW}(i)$.
  We say that the run is \emph{accepting} if{f} $\cl{C}_n = \tup{i,\alpha,\gamma,\emptyset}$ contains an empty set of threads. If for an automaton $\anAut$ we have that $\Lang(\anAut) \neq \emptyset$ we say that $\anAut$ is \emph{nonempty}.

\subsection{Properties}
\label{sec:ara-properties}


We show the following two statements
\begin{iteMize}{$\bullet$}
  \item $\Lang(\ara(\opguess,\opspread))$ is not closed under complementation,
  \item the $\ara(\opguess,\opspread)$ class is  more expressive  than $\ara$.
\end{iteMize}
In fact in the proof below we show the first one, that implies the second one, given the fact that the $\ara$ model is closed under complementation.

\begin{proposition}[Expressive power]\label{prop:guess-spread-more-expressive} \ 
  \begin{enumerate}[\em(a)]
  \item the $\ara(\opguess)$ class is  more expressive than $\ara$;
  \item the $\ara(\opspread)$ class is  more expressive  than $\ara$.
  \end{enumerate}
\end{proposition}
\newcommand{\inc}{\msf{inc}}
\newcommand{\dec}{\msf{dec}}
\newcommand{\iz}{\msf{ifzero}}
\newcommand{\minski}{(i)}
\newcommand{\minskii}{(ii)}
\newcommand{\minskiii}{(iii)}
\begin{proof}
Let $\wW = \aA \otimes \dD$ be a data word.
To prove (a), consider the property ``\textsl{There exists a datum $d$ and a position $i$ with $\dD(i)=d$, $\aA(i)=a$, and there is no position $j \leq i$ with $\dD(j)=d$, $\aA(j)=b$.}''. This property can be easily expressed by $\ara(\opguess)$. It suffices to guess the data value $d$ and checks that we can reach an element $(a,d)$, and that for every previous element $(b,d')$ we have that $d\not=d'$. We argue that this property cannot be expressed by the $\ara$ model. Suppose ad absurdum that it is expressible. This means that its negation would also be expressible by $\ara$ (since they are closed under complementation). The negation of this property states 
\begin{align}
  \begin{split}
    &\text{\textbf{``}\textsl{For every data value $d$, if there is an element $(a,d)$}}\\[-.3em]
     &\text{\textsl{    in the word, then there is a previous element $(b,d)$.}\textbf{''}}
  \end{split}
\tag{P1} \label{eq:property-ara-guess-undec}
\end{align}
With this kind of property one can code an accepting run of a Minsky machine, whose emptiness problem is undecidable. This would prove that $\ara(\opguess)$ have an undecidable emptiness problem, which is in contradiction with the decidability proof that we will give in Section~\ref{sec:emptiness-problem-ara-guess-spread}. Let us see how the reduction works.

The emptiness problem for Minsky machine is known to be undecidable even with an alphabet consisting of one symbol, so we disregard the letters read by the automaton in the following description. Consider then a 2-counter alphabet-blind Minsky machine  whose instructions are of the form $(q,\ell,q')$ with $\ell \in \set{\inc,\dec,\iz}\times \set{1,2}$ being the operation over the counters, and $q,q'$  states from the automaton's set of states $Q$.
A run on this automaton is a sequence of applications of transition rules, for example

\medskip
\begin{tabular}{c@{\,}c@{\,}c@{\,}c@{\,}c@{\,}c}
  $(q_1,\inc_1,q_2)$&$(q_2,\inc_2,q_3)$& $(q_3,\inc_1,q_2)$&$(q_2,\dec_1,q_1)$& $(q_1,\dec_1,q_2)$&$(q_2,\iz_1,q_3)$
\end{tabular}
\medskip

\noindent
This run has an associated data word over the alphabet 
\[Q\times\set{\inc,\dec,\iz}\times \set{1,2}\times Q,\] 
where the corresponding data value of each instruction is used to match  increments  with decrements, for example,

\medskip
\begin{tabular}{c@{\,}c@{\,}c@{\,}c@{\,}c@{\,}c}
  $(q_1,\inc_1,q_2)$&$(q_2,\inc_2,q_3)$& $(q_3,\inc_1,q_2)$&$(q_2,\dec_1,q_1)$& $(q_1,\dec_1,q_2)$&$(q_2,\iz_1,q_3)$.\\
$1$ & $2$ & $3$ & $2$ & $1$ & $4$
\end{tabular}
\medskip

Using the \ara model we can make sure that 
(i)
all increments have different data values and all decrements have different data values; and 
(ii) 
for every $(~,\inc_i,~)$ element with data value $d$ that occurs to the left of a $(~,\iz_i,~)$, there must be a  $(~,\dec_i,~)$ element with data value $d$ that occurs in between. \cite{DL-tocl08} shows how to express these properties using \ara. However, properties {\minski} and {\minskii} are not enough to make sure that every prefix of the run ending in a $\iz_i$ instruction must have as many increments as decrements of counter $i$. Indeed, there could be more decrements than increments --- but not the opposite, thanks to \minskii.

The missing condition to verify that the run is correct is:
(iii) 
for every decrement there exists a  previous increment with the same data value. In fact, we can see that property~\eqref{eq:property-ara-guess-undec} can express condition \minskiii: we only need to change $a$ by a decrement transition in the coding, and $b$ by an increment transition of the same counter. But then, assuming that property~\eqref{eq:property-ara-guess-undec} can be expressed by $\ara(\opguess)$, the emptiness problem for $\ara(\opguess)$ is undecidable. 
This is absurd, as the emptiness problem is decidable, as we will show later on in Theorem~\ref{thm:arags-decidable}. 

\medskip

Using a similar reasoning as before, we show (b):  the $\ara(\opspread)$ class is  more expressive than $\ara$. Consider the property: ``\textsl{There exists a position $i$ labeled $b$ such that $\dD(i)\not=\dD(j)$ for all $j<i$ with $\aA(j)=a$.}'' as depicted in Figure~\ref{fig:prop-spread}. Let us see how this can be coded into $\ara(\opspread)$. Assuming $q_0$ is the initial state, the transitions should reflect that every datum with label $a$ seen along the run is  saved with a state $q_a$, and that this state is in charge of propagating this datum. Then, we guess a position labeled with $b$ and check that all these stored values under $q_a$ are different from the current one. For succinctness we write the transitions as positive boolean combinations of the basic operations.
  \begin{align*}
    \begin{split}
      \delta(q_0) &= (b \land \opspread(q_a,q_1)) \; \lor \;
      \big( ( \bar{a} \lor \opset(q_a) ) \land 
      \triright q_0 \big) ,
    \end{split}
\\ 
    \delta(q_1) &= \opneq, \qquad
    \delta(q_a) =  (\bar{\triright}? \lor \triright q_a) .
  \end{align*}
\begin{figure}
  \centering
    \includegraphics[scale=.7]{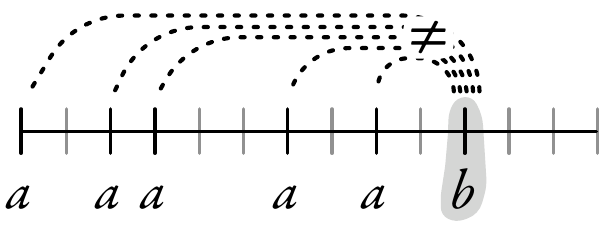}
  \vspace{-1mm}  
  \caption{A property not expressible in \ara.}
  \label{fig:prop-spread}
  \vspace{-3mm}  
\end{figure}
This property cannot be expressed by the \ara model. Were it expressible, then its negation 
\begin{gather}
  \begin{split}
    &\text{\textbf{``}\textsl{for every element $b$ there exists a previous one }}\\[-.3em]
    &\text{\textsl{ labeled $a$ with the same data value}\textbf{''}}  \end{split}\tag{P2}\label{eq:property-ara-spreadundec}
\end{gather}
would also be. Just as before we can use  property \eqref{eq:property-ara-spreadundec} to express condition \minskiii, and force that for every decrement in a coding of a Minsky machine there exists a corresponding previous increment. This leads to a contradiction by proving that the emptiness problem for $\ara(\opspread)$ is undecidable.
\end{proof}

\begin{corollary}\label{cor:ara-not-closed-compement}
  $\ara(\opguess)$, $\ara(\opspread)$ and $\ara(\opguess,\opspread)$ are not closed under complementation.
\end{corollary}
\begin{proof}
  If they were closed under complementation, then we could express some of the properties described in the proof of Proposition~\ref{prop:guess-spread-more-expressive}, resulting in an undecidable model, which is in contradiction with Theorem~\ref{thm:arags-decidable}.
\end{proof}

We then have the following properties of the automata model.
\begin{proposition}[Boolean operations]\label{prop:ara-closed}
 The class $\Lang(\ara(\opspread,\opguess))$ has the following properties:
  \begin{enumerate}[\em(i)]
\item \label{prop:ara-closed:i}  it is closed under union, 
\item \label{prop:ara-closed:ii} it is closed under
    intersection,
\item \label{prop:ara-closed:iii} it is not closed under complementation.
  \end{enumerate}
\end{proposition}
\begin{proof}[Proof sketch]
Items \eqref{prop:ara-closed:i} and \eqref{prop:ara-closed:ii} are straightforward if we notice that the first argument of \opspread ensures that this transition is always relative to the states of one of the automata being under intersection or union. Item \eqref{prop:ara-closed:iii} follows from Corollary~\ref{cor:ara-not-closed-compement}.
\end{proof}

\subsection{Emptiness problem}
\label{sec:emptiness-problem-ara-guess-spread}

This section is dedicated to show the following theorem.

\begin{theorem} \label{thm:arags-decidable}
The  emptiness problem for $\ara(\opguess,\opspread)$ is decidable.
\end{theorem}

  As already mentioned, decidability for $\ara$ was proved  in~\cite{DL-tocl08}. Here we propose an alternative approach that simplifies the proof of decidability of the two extensions \opspread and \opguess.

  The proof goes as follows. We will define a \wqo $\lqw$ over $\araconf$ and show that $(\araconf, \lqw)$ is \emph{\rdc} with respect to $\araT$ (Lemma~\ref{lem:nodeconf-rdc}). Note that strictly speaking $\araT$ is an infinite-branching transition system as  $\rightarrowtriright$ may take \emph{any} value from  the infinite set $\D$, and $\rightarrow_\eps$ can also guess any  value. However, it can trivially be restricted to an effective \emph{finitely} branching one. Then, by Lemma~\ref{lem:downward-closed-decidable}, $(\araconf,\araT)$ has a computable upward-closed reachability set, and this implies that the emptiness problem of $\ara(\opguess,\opspread)$ is decidable.

\newcommand{\eqconf}{\sim}
Since our model of automata only cares about equality or inequality of data values, it is convenient to work modulo renaming of data values.
  \begin{definition}[$\eqconf$]
    We say that two configurations $\rho, \rho'\in\araconf$ are \textbf{equivalent} (notation $\rho \eqconf \rho'$) if there is a bijection $f : \data(\rho) \to \data(\rho')$ such that $f(\rho)=\rho'$, where $f(\rho)$ stands for the replacement of every data value $d$ by $f(d)$ in $\rho$.
  \end{definition}

  \begin{definition}[$\lqwS$]
    We first define the relation $(\araconf,\lqwS)$ such that
    \begin{gather*}
      \tup{i,\alpha,\gamma,\Threads} \;\lqwS\;
      \tup{i',\alpha',\gamma',\Threads'}\label{eq:1}
    \end{gather*}
 if{f} $\alpha = \alpha'$, $\gamma = \gamma'$, and $\Threads \subseteq \Threads'$. 
  \end{definition}
Notice that by the definition above we are `abstracting away' the information concerning the position $i$. We finally define $\lqw$ to be $\lqwS$ modulo $\eqconf$.
  \begin{definition}[$\lqw$]
We define $\rho \lqw \rho'$ if{f} there is $\rho'' \eqconf \rho'$ with $\rho \lqwS \rho'$.
  \end{definition}

The following lemma follows from the definitions.
  \begin{lemma} \label{lem:nodeconf-wqo}
    $(\araconf,\lqw)$ is a well-quasi-order.
  \end{lemma}
 \begin{proof}
    The fact that $\lqw$ is a \emph{quasi-order} (\ie,
    reflexive and transitive) is immediate from its definition. To
    show that it is a \emph{well}-quasi-order, suppose we have an
    infinite sequence of configurations $\rho_1\rho_2\rho_3\dotsb$.
    It is easy to see that it contains an infinite subsequence
    $\tau_1\tau_2\tau_3 \dotsb$ such that all its elements are of the
    form $\tup{i,\alpha_0,(a_0,d),\Threads}$ with
    \begin{iteMize}{$\bullet$}
    \item $\alpha_0$ and $a_0$ fixed, and
    \item $\Threads(d)=\+{C}_0$ fixed,
    \end{iteMize}
This is because we can  see each of these elements as a finite \emph{coloring}, and apply the pigeonhole principle on the infinite set $\set{\rho_i}_i$.

\newcommand{\size}[1]{|#1|}
    Consider then the function $g_\Threads:\subsets(Q)\to\Nz$, such that $g_\Threads(\+{S})=|\set{ d \mid \+{S}=\Threads(d)}|$ (we can think of $g_\Threads$ as a tuple of $(\Nz)^{\subsets(Q)}$). Assume the relation $\lqw^\dag$ defined as $\Threads \lqw^\dag \Threads'$ if{f} $g_\Threads(\+{S})\leq g_{\Threads'}(\+{S})$ for all    $\+{S}$. By Dickson's Lemma $\lqw^\dag$ is a \wqo, and then there are two $\tau_i=\tup{i',\alpha_0,(a_0,d_i),\Threads_i}$,
    $\tau_j=\tup{j',\alpha_0,(a_0,d_j),\Threads_j}$, $i < j$ such that $\Threads_i \lqw^\dag \Threads_j$. For each $\+{S}\subseteq Q$, there exists an
    injective mapping $f_{\+{S}} : \set{d \mid
    \Threads_i(d)=\+{S}} \to \set{d \mid \Threads_j(d)=\+{S}}$ such that $f_{\+S}(d_i)=d_j$,
    as the latter set is bigger than the former by $\lqw^\dag$. We define the injection $f: \data(\tau_i) \to \data (\tau_j)$ as the (disjoint) union of all $f_{\+{S}}$'s. The union is disjoint since for every data value $d$ and set of threads $\Threads$, there is a unique set $\+S$ such that $d \in g_{\Threads}(\+S)$. We then have that $\tau_i \eqconf f(\tau_i) \lqwS \tau_j$. Hence, $\tau_i \lqw \tau_j$.
  \end{proof}

    The core of this proof is centered in the following lemma.
    \begin{lemma}\label{lem:nodeconf-rdc}
$(\araconf,\araT)$       is \rdc with respect to $(\araconf,\lqw)$.
    \end{lemma}

    \begin{proof}
      We shall show that for all $\rho, \tau, \rho' \in \araconf$ such that $\rho \araT \tau$ and $\rho' \lqw \rho$, there is $\tau'$ such that $\tau' \lqw \tau$ and either $\rho' \araT \tau'$ or $\tau' = \rho'$. Since by definition of $\lqw$ we work modulo $\eqconf$, we can further assume that $\rho' \lqwS \rho$ without any loss of generality.
The proof is a simple      case analysis of the definitions for $\araT$. All      cases are treated alike, here we present the most      representative. Suppose first that $\rho \rightarrow_\eps \tau$, then one of the definition conditions      of $\rightarrow_\eps$ must apply. 

If Eq.~\eqref{def:righteps:3} of the definition of $\rightarrow_\eps$ (Fig.~\ref{fig:transition-ara-epsilon}) applies, let
      \[\rho = \tup{i,\alpha,(a,d),\set{(q,d)} \cup \Threads}
      \,\rightarrow_\eps\, \tau=\tup{i,\alpha,(a,d),\Threads}\]
      with $\delta(q) = \opeq$. Let $\rho'=\tup{i',\alpha,(a,d),\Threads'} \lqwS \rho$. If $(q,d)\in \Threads'$, we can then apply the same $\rightarrow_\eps$-transition obtaining $\rho \gqwS \rho' \rightarrow_\eps \tau' \lqwS \tau$. If there is no such $(q,d)$, we can safely take $\rho'=\tau'$ and check that $\tau' \lqwS \tau$.


      If Eq.~\eqref{def:righteps:6} applies, let
      \begin{gather*}
        \rho=\tup{i,\alpha,(a,d),\set{(q,d')} \cup \Threads}
        \,\rightarrow_\eps\, \tau=\tup{i,\alpha,(a,d),\set{(q',d)}
          \cup \Threads}
      \end{gather*}
      with $\rho \rightarrow_\eps \tau$ and $\delta(q) =\opset(q')$. Again let $\rho' \lqwS \rho$ containing $(q,d')\in \Threads'$. In this case we can apply the same      $\rightarrow_\eps$-transition arriving to $\tau'$ where $\tau' \lqwS \tau$. Otherwise, if $(q,d') \not\in \Threads'$, we take $\rho'=\tau'$ and then $\tau' \lqwS \tau$.

      If a \opguess is performed (Eq.~\eqref{def:righteps:guess}), let
      \begin{gather*}
        \rho = \tup{i,\alpha,(a,d),\set{(q,d')} \cup \Threads} \,
        \rightarrow_\eps  \tau = \tup{i,\alpha,(a,d),\set{(q',e)}
          \cup \Threads}
      \end{gather*}
      with $\delta(q) = \opguess(q')$. Let $\rho' =
      \tup{i',\alpha,(a,d),\Threads'} \lqwS \rho$. Suppose there is $(q,d') \in \Threads'$, then we then take a \opguess transition from $\rho'$ obtaining some $\tau'$ by guessing $e$ and hence $\tau' \lqwS \tau$. Otherwise, if $(q,d') \not\in \Threads'$, we take $\tau' = \rho'$ and check that $\tau' \lqwS \tau$.

      Finally, if a \opspread is performed (Eq.~\eqref{def:righteps:spread}),
      let
      \begin{gather*}
        \rho = \tup{i,\alpha,\gamma,\set{(q,d')} \cup \Threads} \,
        \rightarrow_\eps  \tau = \tup{i,\alpha,\gamma,\set{(q_1,d)
          \mid (q_2,d) \in \Threads} \cup \Threads}
      \end{gather*}
      with $\delta(q) =
      \opspread(q_2,q_1)$. Let $\rho' = \tup{i',\alpha,\gamma,\Threads'}
      \lqwS \rho$ and suppose there is $(q,d') \in \Threads'$ (otherwise $\tau' = \rho'$ works). We then take a \opspread instruction $\rho' \rightarrow_\eps \tau'$ and see that      $\tau' \lqwS \tau$, because any $(q_1,e)$ in $\tau'$      generated by the \opspread must come from $(q_2,e)$ of      $\rho'$, and hence there is some $(q_2,e)$ in $\rho$; now by      the \opspread applied on $\rho$, $(q_1,d')$ is in $\tau$. 

The      remaining cases of $\rightarrow_\eps$ are only easier.



      Finally, there can be a `moving' application of $\araT$. Suppose that we have
      \begin{gather*}
        \rho=\tup{i,\ttypeb,(a,d),\Threads} \rightarrowtriright
        \tau=\tup{i+1,\alpha_1,(a_1,d_1),\Threads_1} .
      \end{gather*}
%
     Let
      $\rho'=\tup{i',\ttypeb,(a,d),\Threads'} \lqwS \rho$.
      If $\rho'$ is such that $\rho' \lqwS \tau$, the relation is trivially compatible. Otherwise, we shall prove that there is $\tau'$ such that $\rho' \araT \tau'$ and $\tau' \lqw \tau$. Condition~\ref{ara:moving:i} of $\rightarrowtriright$ (\ie, that all states are moving) holds for $\rho'$, because all the states present in $\rho'$ are also in $\rho$ (by definition of    $\lqwS$) where the condition must hold. Then, we can apply the      $\rightarrowtriright$ transition to $\rho'$ and obtain      $\tau'$ of the form $\tup{i'+1,\alpha_1,(a_1,d_1),\Threads'_1}$. Notice      that we are taking $\alpha_1$, $a_1$ and $d_1$ exactly as in $\tau$,      and that $\Threads'_1$ is completely determined by the      $\rightarrowtriright$ transition from $\Threads'$. We only need to check that $\tau' \lqwS \tau$. Take any $(q,d')\in\Threads'_1$. There must be some $(q',d')\in\Threads'$ with $\delta(q')=\triright q$. Since $\Threads' \subseteq \Threads$, we also have $(q,d) \in \Threads_1$. Hence, $\Threads'_1 \subseteq \Threads_1$ and then $\tau' \lqwS \tau$.
    \end{proof}

We just showed that $(\araconf,\araT)$ is \emph{\rdc} with respect to $(\araconf,\lqw)$. The only missing ingredient to have decidability is the following, which is trivial.
\begin{lemma}\label{lem:ara-downward-closed}
The set of accepting configurations of $\araconf$ is downward closed with respect to $\lqw$.
\end{lemma}

We write $\araconf / {\eqconf}$ to the set of configurations modulo $\eqconf$, by keeping one representative  for every equivalence class. Note that the transition system $(\araconf / {\eqconf}, \araT)$ is effective. This is just a consequence of the fact that the $\araT$-image of any configuration has only a finite number of configurations modulo $\eqconf$, and representatives for every class are computable. Hence, we have that $(\araconf / {\eqconf},\lqw,\araT)$ verify conditions (1) and (2) from Proposition~\ref{prop:rdc-computable}. Finally, condition (3) holds since $(\araconf,\lqw)$  is computable.  We can then apply Lemma~\ref{lem:downward-closed-decidable}, obtaining that for a given $\anAut \in \ara(\opguess, \opspread)$, testing wether there exists a final configuration $\tau$ and an element $\rho$ in
\begin{gather*} \set{\tup{1,\alpha,(a,d_0),\set{(q_I,d_0)}}\mid \alpha \in
    \set{\ttypeb, \ttypenb }, a \in \A}
\end{gather*}
---for any fixed $d_0$--- such that $\rho \eqconf \rho' \araT^* \tau$ (for some $\rho'$) is decidable. Thus, we can decide the  emptiness problem and Theorem~\ref{thm:arags-decidable} follows.

\subsection{Ordered data}
\label{subsect:ara-guess-spread-order}
\newcommand{\lqwo}{<\!\!<}
\newcommand{\testineq}{\ensuremath{\mathsf{test}}}
\index{ARA(guess,spread,<)@$\ara(\opguess,\opspread,<)$}
We show here that the previous decidability result holds even if we add \emph{order} to the data domain.
Let $(\D,<)$ be a linear order, like for example the reals or the natural numbers with the standard ordering. Let us replace the instructions $\opeq$ and $\opneq$ with
\begin{align*}
  \delta(q) \coloneqq \; \dotsc \midd \testineq(>) \midd \testineq(<) \midd \testineq(=) \midd \testineq(\neq)
\end{align*} 
and let us call this model of automata $\ara(\opguess,\opspread,<)$. The semantics is as expected. $\testineq(<)$ verifies that the data value of the current position is less than the data value in the register, $\testineq(>)$ that is greater, and $\testineq(=)$ (resp.\ $\testineq(\neq)$) that both are (resp.\ are not) equal.
We modify accordingly $\rightarrow_\eps$, for $\rho = \tup{i,\alpha,(a,d),\set{(q,d')} \cup \Threads}$.
\begin{align}
    \rho &\rightarrow_\eps  \tup{i,\alpha,(a,d),\Threads} \quad  \text{ if } \delta(q) = \testineq(<) \text{ and }  d < d' \label{def:righteps:testineq<}\\
    \rho &\rightarrow_\eps  \tup{i,\alpha,(a,d),\Threads} \quad  \text{ if } \delta(q) = \testineq(>) \text{ and }  d > d' \label{def:righteps:testineq>}\\
    \rho &\rightarrow_\eps  \tup{i,\alpha,(a,d),\Threads} \quad  \text{ if } \delta(q) = \testineq(=) \text{ and }  d = d' \label{def:righteps:testineq=}\\
    \rho &\rightarrow_\eps  \tup{i,\alpha,(a,d),\Threads} \quad  \text{ if } \delta(q) = \testineq(\neq) \text{ and }  d \neq d' \label{def:righteps:testineq-neq}
\end{align}
All the remaining definitions are preserved. We can show that the emptiness problem for this extended model of automata is still decidable.

\begin{theorem}\label{thm:ara-guess-spread-order:decidable}
  The emptiness problem for $\ara(\opguess,\opspread,<)$ is decidable.
\end{theorem}
As in the proof in the previous Section~\ref{sec:emptiness-problem-ara-guess-spread}, we show that there is a \wqo ${\lqwo} \subseteq \araconf \times \araconf$ that is $\rdc$ with respect to $\araT$, such that the set of final states is $\lqwo$-downward closed. However, we need to be more careful when showing that we can always work modulo an equivalence relation.

\begin{definition}  A function $f$ is an \textbf{order-preserving bijection} on $D \subseteq \D$ if{f} it is a bijection on $\D$, and furthermore for every $\set{d,d'} \subseteq D$, if $d< d' $ then $f(d) < f(d')$.
\end{definition}

The following Lemma is straightforward from the definition just seen. 
\begin{lemma}\label{lem:make-space-in-ordered-bijection}
  Let $D \subseteq \D$, $|D| < \infty$. There exists an order-preserving bijection $f$ on $D$ such that 
  \begin{iteMize}{$\bullet$}
  \item for every $\set{d,d'} \subseteq D$ such that $d < d'$ there exists $\tilde d$ such that $f(d) < \tilde d < f(d')$,
  \item for every $d \in D$ there exists $\tilde d$ such that $f(d) < \tilde d$, and there exists $\tilde d$ such that $\tilde d < f(d)$.
  \end{iteMize}
\end{lemma}
\newcommand{\eqconfo}{\eqconf_{\mathit{ord}}}
\begin{definition}[$\eqconfo$]
Let $\rho, \rho'$ be two configurations.  We define  $\rho \eqconfo \rho'$ if{f} $f(\rho)=\rho'$ for some order-preserving bijection $f$ on $\data(\rho)$.
\end{definition}
\begin{remark}\label{rem:only-one-datum-added-in-one-step}
  If $\rho \araT \rho'$ then there exists $\hat d \in \D$ such that $\set{\hat d} \cup \data(\rho) \subseteq \data(\rho')$. This is a simple consequence of the definition of $\araT$.
\end{remark}

\newcommand{\araTord}{\araT_\textit{ord}}
Let us define a version of $\araT$ that works modulo $\eqconfo$, and let us call it $\araTord$.
\begin{definition}
Let $\rho_1, \rho_2$ be two configurations. We define $\rho_1 \araTord \rho_2$ if{f} $\rho'_1 \araT \rho'_2$ for some $\rho'_1 \eqconfo \rho_1$ and $\rho'_2 \eqconfo \rho_2$.
\end{definition}

In the previous section, when we had that $\eqconf$ was simply a bijection and we could not test any linear order $<$, it was clear that we could work modulo $\eqconf$. However, here we are working modulo a more complex relation $\eqconfo$. In the next lemma we show that working with $\araT$ or working with $\araTord$ is equivalent.

\begin{lemma}
  If $\rho_1 \araTord \dotsb \araTord \rho_n$, then $\rho'_1 \araT \dotsb \araT \rho'_n$, with $\rho'_i \eqconfo \rho_i$ for every $i$.
\end{lemma}
\begin{proof}
  The case $n=1$ is trivial. Otherwise, if $n>1$, we have $\rho_1 \araTord \dotsb \araTord \rho_{n-1} \araTord \rho_n$. Then, by inductive hypothesis, we obtain $\rho'_1 \araT \dotsb \araT \rho'_{n-1}$ and $\rho''_{n-1} \araT \rho''_n$ with $\rho'_i \eqconfo \rho_i$ for every $i \in \set{1, \dotsc,  n-1}$,  and $\rho''_{j} \eqconfo \rho_{j}$ for every $j \in \set{n-1, n}$. Let $g$ be the witnessing bijection such that $g(\rho''_{n-1})=\rho'_{n-1}$, and let us assume that $\set{\hat d} \cup \data(\rho''_{n-1}) \subseteq \data(\rho''_n)$ by Remark~\ref{rem:only-one-datum-added-in-one-step}.

Let $f$ be an order-preserving bijection on $\bigcup_{i \leq n-1}\data(\rho'_i)$ as in Lemma~\ref{lem:make-space-in-ordered-bijection}. We can then pick a data value $\tilde d$ such that
\begin{iteMize}{$\bullet$}
\item for every $d > \hat d$ with $d \in \data(\rho''_{n-1})$, $f(g(d)) > \tilde d$, and
\item for every $d < \hat d$ with $d \in \data(\rho''_{n-1})$, $f(g(d)) < \tilde d$.
\end{iteMize}
Let $h \coloneqq (g \circ f) [\hat d \mapsto \tilde d]$. We then have
\begin{iteMize}{$\bullet$}
\item $h(\rho''_{n-1}) \araT h(\rho''_n)$,
\item $f(\rho'_1) \araT \dotsb \araT f(\rho'_{n-1})$,
\item $f(\rho'_{n-1}) = h(\rho''_{n-1})$.
\end{iteMize}
In other words, $f(\rho'_1) \araT \dotsb \araT f(\rho'_{n-1}) \araT h(\rho''_n)$, with $h(\rho''_n) \eqconfo \rho_n$ and  $f(\rho'_i) \eqconfo \rho_i$ for every $i \leq n-1$.
\end{proof}

Now that we proved that we can work modulo $\eqconfo$, we show that we can decide if we can reach an accepting configuration by means of $\araTord$, by introducing some suitable ordering $\lqwo$ and showing the following lemmas.

\begin{lemma}\label{lem:ordered-wqo}
  $(\araconf,\lqwo)$ is a well-quasi-order.
\end{lemma}
\begin{lemma}\label{lem:ordered-rdc}
  $(\araconf,\lqwo)$ is \rdc with respect to $\araTord$.
\end{lemma}
\begin{lemma}\label{lem:ordered-accepting-downward}
The set of accepting configurations of $\araconf$ is downward closed with respect to $\lqwo$.
\end{lemma}

We next define the order $\lqwo$ and show that the aforementioned lemmas are valid. In the same spirit as before, $\lqwo$ is defined as $\lqwS$  modulo  $\eqconfo$.
\begin{definition}[$\lqwo$] $\rho_1 \lqwo \rho_2$ if{f} $\rho'_1 \lqwS \rho'_2$ for some $\rho'_1 \eqconfo \rho_1$ and $\rho'_2 \eqconfo \rho_2$.
\end{definition}


\newcommand{\abs}{\mathit{abs}}
To prove Lemma~\ref{lem:ordered-wqo}, given a configuration $\rho = \tup{i_0,\alpha,(a,d),\Threads}$, with $\data(\rho) = \set{d_1 < \dotsb < d_n}$ we define
\begin{align*}
\abs(d_i) &=\Threads(d_j) \cup \set{\star \mid d_j = d} \subseteq Q \cup \set{\star}\\
\abs(\rho) &= \abs(d_1), \dotsc, \abs(d_n) \in (\subsets(Q \cup \set{\star}))^*
\end{align*}
where $\star \not\in Q$ is to denote that the data value is the one of the current configuration.

  \begin{proof}[Proof of Lemma~\ref{lem:ordered-wqo}]
    This is a consequence of \higmanslem stated as in Corollary~\ref{cor:higman-lemma}. As stated above, we can see each configuration $\rho = (i,\alpha,(a,d),\Threads)$ as a word over $(\subsets(Q\cup\set{\star}))^*$. As shown in Lemma~\ref{lem:nodeconf-wqo} if there is an infinite sequence, there is an infinite subsequence $\rho_1, \rho_2, \dotsc$, with the same type $\alpha$ and letter $a$. Then for the infinite sequence $\abs(\rho_1), \abs(\rho_2), \dotsc$,  Corollary~\ref{cor:higman-lemma} tells us that there are $i < j$ such that $\abs(\rho_i)$ is a substring of $\abs(\rho_j)$. This implies that they are in the $\lqwo$ relation. 
  \end{proof}
  \begin{proof}[Proof of Lemma~\ref{lem:ordered-rdc}]
Note that although $\lqwo$ is a more restricted \wqo, for all the non-moving cases in which the register is not modified (that is, all except $\opguess$, $\opspread$, and $\opset$), the $\rightarrow_\eps$ transition continues to be \rdc. This is because for any $\tau \lqwo \rho \rightarrow_\epsilon \rho'$, $\rho=\tup{i,\alpha,\gamma,\Threads}$ and $\rho'=\tup{i',\alpha',\gamma',\Threads'}$ are  similar in the following sense. Firstly $\data(\rho)=\data(\rho')$, and moreover the only difference between $\rho$ and $\rho'$ is that $\Threads'$ is the result of removing some thread $(q,d)$  from $\Threads$ and inserting another one $(q',d)$ with the same data value $d$. This kind of operation is compatible, since $\tau$ can perform the same operation $\tau \rightarrow_\eps \tau'$ on the data value $d'$, supposing that $d'$ is the preimage of $d$ given by the $\lqwo$ ordering. In this case, $\tau' \lqwo \rho'$.  Otherwise, if there is no preimage of $d$, then $\tau \lqwo \rho'$. The compatibility of $\opspread$ is shown equivalently.

Regarding the $\opset$ instruction, we see that the operation consists in removing some $(q,d)$ from $\Threads$ and inserting  some $(q',d_0)$ with $d_0$ the datum of the current configuration. This is downwards compatible since $\tau$ can perform the same operation on the configuration's data value, which is necessarily the preimage of $d_0$.

For the remaining two cases (\opguess and $\triright$) we rely on the premise that we work modulo $\eqconfo$. The idea is that we can always assume that we have enough data values to choose from in between the existing ones. That is, for every pair of data values $d<d'$ in a configuration, there is always one in between. We can always assume this since otherwise we can apply a bijection as the one described by Lemma~\ref{lem:make-space-in-ordered-bijection} to obtain this property. Thus, at each point where we need to guess a data value (as a consequence of a $\opguess(q)$ or a $\triright q$ instruction) we will have no problem in performing a symmetric action, preserving the embedding relation. 

More concretely, suppose the execution of a transition  $\delta(q)=\opguess(q')$ on a thread $(q,d_j)$ of configuration $\rho$ with $\data(\rho)=\set{d_1 < \dotsb < d_n}$ guesses a data value $d$ with $d_i<d<d_{i+1}$. Then, for any configuration $\tau \lqwo \rho$ with $\data(\tau)=\set{e_1< \dotsb < e_m}$ and the property just described, there must be an order-preserving injection $f : \data(\tau) \to \data(\rho)$ with $f(\tau) \lqwS \rho$. If $\tau$ contains a thread $(q,e_{j'})$ with $f(e_{j'})=d_j$ the operation is simulated by guessing a data value $e$ such that $e>e_\ell$ for all $e_\ell$ such that $f(e_\ell) \leq d_i$ and $e<e_k$ for all $e_k$ such that $f(e_k) > d_i$. Such data value $e$ exists as explained before. The \rdc compatibility of a $\delta(q)=\triright q'$ instruction is shown in an analogous fashion.
  \end{proof}
  \begin{proof}[Proof of Lemma~\ref{lem:ordered-accepting-downward}]
Given that $\lqwo$ is a subset of $\lqw$, and that by Lemma~\ref{lem:ara-downward-closed} the set of accepting configurations is $\lqw$-downward closed, it follows that this set is also $\lqwo$-downward closed.
  \end{proof}

Finally, we should note that $(\araconf / {\eqconfo}, \araTord)$ is also finitely branching and effective. As in the proof of Section~\ref{sec:emptiness-problem-ara-guess-spread}, by Lemmas~\ref{lem:ordered-wqo}, \ref{lem:ordered-rdc} and \ref{lem:ordered-accepting-downward} we have that all the conditions of Proposition~\ref{prop:rdc-computable} are met and by Lemma~\ref{lem:downward-closed-decidable} we conclude that the emptiness problem for $\ara(\opguess,\opspread,<)$ is decidable.


\begin{remark}
Notice that this proof works independently of the particular ordering of $(\D,<)$. It could be dense or discrete, contain accumulation points, be open or closed, etc. In some sense, this automata model is blind to these kind of properties. If there is an accepting run on $(\D,<)$ then there is an accepting run on $(\D,<')$ for any linear order $<'$.
\end{remark}

\begin{open}
It is perhaps possible that these results can be extended  to prove decidability when $(\D,<)$ is a tree-like \emph{partial order}, this time making use of Kruskal's tree theorem \cite{Kr60} instead of \higmanslem. We leave this issue as an open question.
\end{open}

\begin{remark}[constants]
One can also extend this model with a finite number of constants $\set{c_1, \dotsc, c_n} \subseteq \D$. In this case, we extend the transitions with the possibility of testing that the data value  stored in the register is (or is not) equal to $c_i$, for every $i$. In the proof, it suffices to modify $\eqconfo$ to take into account every constant $c_i$. In this case we define that $\rho \eqconfo \tau$ if{f} $f(\rho) = \rho'$ for some order-preserving bijection $f$ on $\data(\rho) \cup \set{c_1, \dotsc, c_n}$ such that $f(c_i)=c_i$ for every $i$.  In this case Lemma~\ref{lem:make-space-in-ordered-bijection} does not hold anymore, as there could be finitely many elements in between two constants. This is however an easily surmountable obstacle, by adapting Lemma~\ref{lem:make-space-in-ordered-bijection} to work separately on the $n+1$ intervals defined by $c_1, \dotsc, c_n$.  Suppose $c_1 < \dotsb < c_n$. Without any loss of generality we can assume that between $c_i$ and $c_{i+1}$ there are infinitely many data values, or none (we can always add some constants to ensure this). Then,
for every infinite interval $[c_i,c_{i+1}]$, we will have a lemma like Lemma~\ref{lem:make-space-in-ordered-bijection} that we can apply separately.
\end{remark}

\subsection{Timed automata}
\label{sec:ordered-reg-aut-VS-timed-automata}
\index{Timed automata}
Our investigation on register automata also yields new results on the class of timed automata. An \emph{alternating 1-clock timed automaton} is an automaton that runs over \emph{timed words}. A timed word is a finite sequence of \emph{events}. Each event carries a symbol from a finite alphabet and a \emph{timestamp} indicating the quantity of time elapsed from the first event of the word. A timed word can hence be seen as a data word over the rational numbers, whose data values are  strictly increasing. The automaton has alternating control and contains one \emph{clock} to measure the lapse of time between two events (that is, the difference between the data of two positions of the data word). It can reset the clock, or test whether the clock contains a number equal, less or greater than a constant, from some finite set of constants. For more details on this automaton we refer the reader to \cite{AD94}.

Register automata over ordered domains  have a strong connection with timed automata. The work in \cite{FHL10} shows that the problems of nonemptiness, language inclusion, language equivalence and universality are equi\-va\-lent---modulo an \exptime reduction---for timed automata and register automata over a linear order. That is, any of these problems for register automata can be reduced  to the same problem on timed automata, preserving the number of registers equal to the number of clocks, and the mode of computation (nondeterministic, alternating). And in turn, any of these problems for timed automata can also be reduced to a similar problem on register automata over a linear order. We argue that this is also true when the automata are equipped with $\opguess$ and $\opspread$.

Consider an extension of 1-clock alternating timed automata, with  $\opspread$ and $\opguess$, where
\begin{iteMize}{$\bullet$}
\item the operator $\opspread(q,q')$ works in the same way as for register automata, duplicating all threads with state $q$ as threads with state $q'$, and
\item the $\opguess(q)$ operator resets the clock to any value, non deterministically chosen, and continues the execution with state $q$.
\end{iteMize}

\noindent The coding technique of \cite{FHL10} can be adapted to deal with the guessing of a clock (the $\opspread$ operator being trivially compatible), and one can show the following statement.

\begin{remark}\label{lem:alt1clock-guess-spread-reduces-ara}
  The emptiness problem for alternating 1-clock timed automata extended with $\opguess$ and $\opspread$
reduces to the emptiness problem for the class \mbox{$\ara(\opguess,\opspread,<)$}.
\end{remark}

Hence, by Remark~\ref{lem:alt1clock-guess-spread-reduces-ara} cum Theorem~\ref{thm:ara-guess-spread-order:decidable} we obtain the following.
\begin{remark}
  The emptiness problem for alternating 1-clock timed automata extended with $\opguess$ and $\opspread$ is decidable.
\end{remark}

\subsection{A note on complexity}
Although the $\ara(\opguess, \opspread)$ and $\ara(\opguess, \opspread, <)$ classes have both non-primitive-recursive complexity, we must remark that the decision procedure for the latter has much higher complexity. While the former can be roughly bounded by the  Ackermann function applied to the number of states, the complexity of the decision procedure we give for $\ara(\opguess,\opspread,<)$ majorizes every multiply-recursive function (in particular, Ackermann's). In some sense this is a consequence of relying on \higmanslem instead of \dicksonslem for the termination arguments of our algorithm.

More precisely, it can be seen that the emptiness problem for
$\ara(\opguess,$ $\opspread,<)$ sits in the class $\mathfrak
F_{\omega^\omega}$ in the Fast Growing Hierarchy \cite{fast}---an
extension of the Grzegorczyk Hierarchy for non-primitive-recursive
func\-tions---by a reduction to the emptiness problem for timed one
clock automata (see \S\ref{sec:ordered-reg-aut-VS-timed-automata}),
which are known to be in this class.\footnote{The emptiness problem
  for timed one clock automata can be at the same time reduced to that
  of Lossy Channel Machines \cite{ADOW05}, which are known to be
  `complete' for this class, \ie, in $\mathfrak F_{\omega^\omega}
  \setminus \mathfrak F_{<\omega^\omega}$ (see~\cite{CS-lics08}).}
  However, the emptiness problem for $\ara(\opguess,\opspread)$
  belongs to $\mathfrak F_\omega$ in the hierarchy. The lower bound
  follows by a reduction from Incrementing Counter Automata
  \cite{DL-tocl08}, which are hard for $\mathfrak F_\omega$
  \cite{phs-mfcs2010,FFSS10}. The upper bound is a consequence of
  using a saturation algorithm with a \wqo that is the component-wise
  order of the coordinates of a vector of natural numbers in a
  controlled way. The proof that it belongs to $\mathfrak F_\omega$
  goes similarly as for Incrementing Counter Automata (see~\cite[\S
    7.2]{FFSS10}).  We do not know whether the emptiness problem for
  $\ara(\opguess, \opspread, <)$ is also in $\mathfrak F_\omega$.

\section{LTL with registers}
\label{sec:ltl-with-registers}
\newcommand{\ltlt}[1]{{\sf LTL}^\downarrow_2(#1)\xspace}
\newcommand{\ltlnnf}[1]{{\sf LTL}^\downarrow_\nnf(#1)}
\newcommand{\ltlF}[1]{\ensuremath{\ltlnnf{\mathfrak F ,#1}}\xspace}
\newcommand{\ltlFut}{\ensuremath{\ltlnnf{\mathfrak F}}\xspace}
\newcommand{\ltlf}{\ensuremath{\ltl{\F}}\xspace}
\newcommand{\ltlfp}{\ensuremath{\ltl{\F,\F^{-1}}}\xspace}
\newcommand{\ltlux}{\ensuremath{\ltl{\U,\X}}\xspace}
\newcommand{\ltlurx}{\ltlFut}
\newcommand{\ltlfx}{\ensuremath{\ltl{\X,\F}}\xspace}
\newcommand{\ltlfs}{\ensuremath{\ltl{\Fs}}\xspace}
\newcommand{\ltlfps}{\ensuremath{\ltl{\Fs,\Fs^{-1}}}\xspace}
\newcommand{\ltluxae}{\ltlF{\Edown,\Adown}}
\newcommand{\ltlurxae}{\ltlF{\Edown,\Adown}}
\newcommand{\ltlurxa}{\ltlF{\Adown}\xspace}
\newcommand{\ltlurxe}{\ltlF{\Edown}\xspace}
\newcommand{\Edown}{\Edownf}
\newcommand{\Adown}{\Adownp}
\newcommand{\Edownp}{\exists_\leq^\downarrow}
\newcommand{\Edownps}{\exists_<^\downarrow}
\newcommand{\Adownp}{\forall_\leq^\downarrow}
\newcommand{\Adownps}{\forall_<^\downarrow}
\newcommand{\Edownf}{\exists_\geq^\downarrow}
\newcommand{\Edownfs}{\exists_>^\downarrow}
\newcommand{\Adownf}{\forall_\geq^\downarrow}
\newcommand{\Adownfs}{\forall_>^\downarrow}

The logic $\ltl{\U,\X}$ is a logic for data words that corresponds to the extension of the Linear-time Temporal Logic ${\sf LTL}(\U,\X)$ on data words, with the ability to
use one \emph{register} for storing a data value for later comparisons, studied in \cite{DL-tocl08,DLN05}. It contains the next (\X) and until (\U) temporal operators to navigate the data word, and two extra operators. The \emph{freeze} operator $\downarrow \varphi$ permits to \emph{store} the current datum in the register and continue the evaluation of the formula $\varphi$. The operator $\uparrow$  \emph{tests} whether the current data value is equal to the one stored in the register.

As it was shown in~\cite{DL-tocl08}, if we allow more than one register to store data values, satisfiability of ${\sf LTL}^\downarrow(\U,\X)$ becomes undecidable. We will
then focus on the language that uses only \emph{one} register. We study an extension of this language with a restricted form of \emph{quantification} over data values. We will actually add \emph{two} sorts of quantification. On the one hand,  the operators $\Adownp$ and $\Edownp$ quantify universally or existentially over all the data
values occurring \emph{before} the current point of
evaluation. Similarly, $\Adownf$ and $\Edownf$ quantify over the
\emph{future} elements on the data word. For our convenience and without any loss of generality,
we will work in Negation Normal Form (\nnf), and we use $\Ud$ to denote
the dual operator of $\U$, and similarly for $\Xd$. Sentences of
$\msf{LTL}^\downarrow_\nnf(\U,\R,\X,\Xd,\+{O})$, where
$\+{O}\subseteq \set{\Adownp,\Edownp,\Adownf,\Edownf}$ are
defined as follows,
\begin{multline*}
  \varphi \Coloneqq a \midd {\lnot a} \midd {\uparrow} \midd {\lnot\uparrow}
  \midd {\downarrow \varphi} \midd {\X \varphi} \midd {\Xd \varphi} \midd 
  {\U(\varphi,
  \varphi)} \midd {\Ud(\varphi, \varphi)} \midd 
  {\text{op } \varphi} \midd {\varphi \land \varphi} \midd {\varphi \lor  \varphi}
\end{multline*}
where $a$ is a symbol from a finite alphabet $\A$, and $\text{op} \in \+{O}$. For economy of space we  write $\ltlF{\+{O}}$ to denote this logic. In this notation, $\mathfrak{F}$ is to mark that we have all the \emph{forward} modalities: $\U,\Ud,\X,\Xd$. Notice that the \emph{future} modality can be
defined $\F \varphi \coloneqq  \U(\varphi,\top)$ and its dual
$\G\varphi$ as the \nnf of $\lnot \F \lnot \varphi$.

Figure~\ref{fig:semantics-ltl} shows the  definition of the satisfaction relation $\models$.
For example, in this logic we can express properties like ``\textsl{for every $a$ element there is a future $b$ element with the same data value}'' as $\G ( {\lnot a} \lor {\downarrow (\F ({b} \land {\uparrow}))})$.
We say that $\varphi$ satisfies $\wW=\aA\otimes\dD$, written $\wW \models \varphi$, if $\wW, 1 \models^{\dD(1)} \varphi$. 

\begin{figure}
  \centering
  \begin{align*}
    (\wW, i) \models^d a  &\text{\quad if{f} \quad}  \aA(i) = a\\
    (\wW, i) \models^d \, \uparrow &\text{\quad if{f} \quad} d = \dD(i)\\ (\wW, i) \models^d \,
    \downarrow \varphi &\text{\quad if{f} \quad} (\wW, i)
    \models^{\dD(i)} \varphi\\
    (\wW, i) \models^d \U(\varphi,\psi) &\text{\quad if{f} \quad}
    \textrm{for some } i \leq j \in \tpos(\wW) 
    \text{ and for all } i \leq k < j\\
    &\hspace{0.88cm} \textrm{ we have } (\wW, j) \models^d \varphi \text{ and  } (\wW,k) \models^d \psi\\
    (\wW, i) \models^d \X \varphi  &\text{\quad if{f} \quad} i+1 \in \tpos(\wW) \textrm{ and }(\wW, i+1) \models^d \varphi\\
    (\wW, i) \models^d \Edownf \varphi &\text{\quad if{f} \quad}
    \text{there
      exists } i \leq j \in \tpos(\wW) \text{ such that }  (\wW, i) \models^{\dD(j)} \varphi\\
    (\wW, i) \models^d \Adown \varphi &\text{\quad if{f} \quad} \text{for
      all } 1 \leq j \leq i \text{ we have } (\wW, i)
    \models^{\dD(j)} \varphi
  \end{align*}
  \caption{Semantics of $\msf{LTL}^\downarrow(\U,\R,\X,\Xd,\Edownf,\Adownp)$ for a data word $\wW = \aA \otimes \dD$ and   $i \in\tpos(\wW)$.}
\label{fig:semantics-ltl}
\end{figure}

\subsection{Satisfiability problem}
This section is dedicated to the satisfiability problem for ${\sf LTL}^\downarrow_\nnf(\mathfrak F, \Adownp, \Edownf)$. But first let us show that $\Edownp$ and $\Adownf$ result in an undecidable logic.
\begin{theorem}\label{thm:temporal-undec-dec}
  Let $\Edownps$ be the operator $\Edownp$
  restricted only to the data values occurring strictly before the
  current point of evaluation. Then, on finite data words: 
  \begin{enumerate}[\em(1)]
  \item satisfiability of $\ltlnnf{\F,\G,\Edownps}$ is undecidable; and \label{thm:temporal-undec-dec:1}
  \item satisfiability of $\ltlnnf{\F,\G,\Adownf}$ is undecidable.\label{thm:temporal-undec-dec:2}
  \end{enumerate}
\end{theorem}
\begin{proof}
We prove \eqref{thm:temporal-undec-dec:1} and \eqref{thm:temporal-undec-dec:2} by reduction of the halting problem for Minsky machines. We show that these logics can code an accepting run of a 2-counter Minsky machine as in Proposition~\ref{prop:guess-spread-more-expressive}. Indeed, we show that the same kind of properties are expressible in this logic. To prove this, we build upon some previous results~\cite{FS09} showing that $\ltlnnf{\F,\G}$ can code conditions {\minski} and {\minskii} of the proof of Proposition~\ref{prop:guess-spread-more-expressive}. Here we show that both $\ltlnnf{\F,\G,\Edownps}$ and $\ltlnnf{\F,\G,\Adownf}$ can express condition \minskiii,  ensuring that for every decrement ($\dec_i$) there is a previous increment ($\inc_i$) with the same data value. Let us see how to code this.
  \begin{enumerate}[(1)]
  \item[\eqref{thm:temporal-undec-dec:1}] The ${\sf LTL}^\downarrow_\nnf(\F,\G,\Edownps)$ formula 
\[\G (\msf{dec}_i \rightarrow \Edownps \uparrow) \] states that the data value of every decrement must \emph{not} be
    new, and in the context of this coding this means that it must
    have been introduced by an increment instruction.
  \item[\eqref{thm:temporal-undec-dec:2}] The ${\sf LTL}^\downarrow_\nnf(\F,\G,\Adownf)$ formula
\[\Adownf (\F(\msf{dec}_i \land \uparrow) \rightarrow
    \F(\msf{inc}_i \land \uparrow))\]
evaluated at the first element of the data word expresses that for every data value: if there is a \emph{decrement} with value $d$, then there is an \emph{increment} with value $d$. It is then easy to ensure that they appear in the correct order (first the increment, then the decrement).
  \end{enumerate}
  The addition of any of these conditions to the coding of \cite{FS09}
  results in a coding of an $n$-counter Minsky machine, whose
  emptiness problem is undecidable.
\end{proof}
\begin{corollary}
  The  satisfiability problem for both ${\sf LTL}^\downarrow_\nnf(\mathfrak F, \Edownp)$ and
  ${\sf LTL}^\downarrow_\nnf(\mathfrak F, \Adownf)$ are undecidable. 
\end{corollary}
\begin{proof}
The property of item \eqref{thm:temporal-undec-dec:1} in the proof of Theorem~\ref{thm:temporal-undec-dec} can be equally coded in ${\sf LTL}^\downarrow_\nnf(\mathfrak F, \Edownp)$ as
  $\G (\X(\msf{dec}_i) \rightarrow \Edownp (\X\uparrow))$. The undecidability of ${\sf LTL}^\downarrow_\nnf(\mathfrak F, \Adownf)$ follows directly from Theorem~\ref{thm:temporal-undec-dec}, item \eqref{thm:temporal-undec-dec:2}.
\end{proof}

We now turn to our decidability result. We show that ${\sf LTL}^\downarrow_\nnf(\mathfrak F, \Adownp, \Edownf)$ has a decidable satisfiability problem by a translation to $\ara(\opguess,\opspread)$.

The translation to code ${\sf LTL}^\downarrow_\nnf(\mathfrak F, \Adownp, \Edownf)$ into $\ara(\opguess,\opspread)$ is standard, and follows same lines as \cite{DL-tocl08}\footnote{Note that this logic already contains $\ltl{\U,\X}$.} (which at the same time follows the translation from $\ltl{\U,\X}$ to alternating finite automata).  We then obtain the following result.

\begin{proposition}\label{prop:ltl2ara}
$\ara(\opguess,\opspread)$ captures $\ltlnnf{\mathfrak F,\Edown,\Adown}$.
\end{proposition}
From Proposition~\ref{prop:ltl2ara} and Theorem~\ref{thm:arags-decidable} it will follow the main result, stated next.
\begin{theorem}\label{thm:ltl-A-E-decidable}
The satisfiability problem for $\ltlnnf{\mathfrak F,\Edown,\Adown}$ is decidable.
\end{theorem}

We now show how to make the translation to $\ara(\opguess,\opspread)$ in order to obtain our decidability result.

\begin{proof}[Proof of Proposition~\ref{prop:ltl2ara}]
Let $\eta \in \ltlnnf{\mathfrak F, \Adownp,\Edownf}$. 
We show that for every formula $\eta \in \ltlnnf{\mathfrak F, \Adownp,\Edownf}$ there exists a computable $\ara(\opspread,\opguess)$ $\anAut_\eta$ such
  that for every data word $\wW$,
  \begin{gather*}
    \text{$\wW$ satisfies $\eta$ ~~~if{f}~~~ $\anAut_\eta$ accepts $\wW$.}
  \end{gather*}

In the construction of $\anAut_\varphi$,  we first make sure to maintain all  the data values seen so far as \emph{threads} of the configuration. We can do this by having 3  special states $q_1, q_2, q_{save}$ in $\anAut_\eta$, defining $q_1$ as the initial state, and  $\delta$  as follows.
  \begin{align*}
    \delta(q_1) &= \opset(q_2) \land q_\eta &
    \delta(q_2) &= (\bar{\triright}? \lor \triright q_1) \land q_{save}&
    \delta(q_{save}) &= \bar{\triright}? \lor \triright q_{save}
  \end{align*}
  Now we can assume that at any point of the run, we maintain the data  values of all the previous elements of the data word as threads  $(q_{\text{save}},d)$. Note that these threads are maintained until the last element of the data word, at which point the test $\ttypenb?$ is satisfied and they are accepted.



Now we show how to define $\anAut_\eta$. We proceed by induction on $|\eta|$. If $\eta = a$ or $\lnot a$, we simply define the set of states as $Q=\set{q_1,q_2,q_{\text{save}},q_\eta}$ and $\delta(q_\eta) = a$ or $\delta(q_\eta) = \bar a$ ($\delta(q_1)$, $\delta(q_2)$ and $\delta(q_{\text{save}})$ are defined as above). If $\eta = {\downarrow} \psi$ (or $\eta = {\uparrow}$, $\eta = \lnot{\uparrow}$)   we define it as follows. We add one new state $q_\eta$ to the set of states of $\anAut_\psi$, we extend the definition of $\delta$ with $\delta(q_\eta) = \opset(q_\psi)$ (or $\delta(q_\eta) = \opeq(q_\psi)$, $\delta(q_\eta) = \opneq(q_\psi)$), and we redefine $\delta(q_1)$ with $\delta(q_1) = \opset(q_2) \land q_\eta$. If $\eta = \F \psi$, $\eta = \G \psi$, or $\eta = \U(\psi,\psi')$, it is easy to define $\delta(q_\eta)$ from the definition of $\anAut_\psi$ and $\anAut_{\psi'}$, perhaps adding some new states. On the other hand, if $\varphi = \psi \land \psi'$ or $\varphi = \psi \lor \psi'$, it is also straightforward as it corresponds to the alternation and nondeterminism of the automaton. 

Suppose now that $\eta = \Adownp \psi$. We define $\anAut_\eta$ as $\anAut_\psi$ with the extra state $q_\eta$ and we define $\delta(q_\eta) = \opspread(q_{\text{save}},q_\psi) \land {\downarrow} q_{\psi}$. This ensures that at the moment of execution of this instruction, all the previous data values in the data word will be taken into account, including the current one.

Finally, if $\eta = \Edownf
  \psi$, we build $\anAut_\eta$ from $\anAut_\psi$ in the same fashion as before, adding new states $q_\eta, q_\eta', q''_\eta$, defining $\delta(q_\eta)=\opguess(q'_\eta)$,
  $\delta(q'_\eta)=q_\psi \land q''_\eta$, 
$\delta(q''_\eta) = \opeq \lor \triright q''_\eta$. Note that $q''_\eta$ checks that the guessed data value appears somewhere in a future position of the word.
\end{proof}

Moreover, we argue that these extensions add expressive power.

\begin{proposition} \label{prop:ltl-expressive}
  On finite data words:
  \begin{enumerate}[\em(i)]
  \item The logic $\ltlurxa$ is  more expressive than \ltlurx;
  \item The logic $\ltlurxe$ is  more expressive than \ltlurx.
  \end{enumerate}
\end{proposition}
\begin{proof}
This is a consequence of \ltlurx being closed under negation and Theorem~\ref{thm:temporal-undec-dec}. Ad absurdum, if one of these logics were as expressive as $\ltlurx$, then it would be closed under negation, and then we could express conditions \eqref{thm:temporal-undec-dec:1} or \eqref{thm:temporal-undec-dec:2} of the proof of Theorem~\ref{thm:temporal-undec-dec} and hence obtain that $\ltlurx$ is undecidable. But this leads to a contradiction since by Theorem~\ref{thm:ltl-A-E-decidable} $\ltlurx$ is decidable.
\end{proof}

\begin{remark}
  The translation of Proposition~\ref{prop:ltl2ara}  is far from using all the expressive
  power of $\opspread$. In fact, we can consider a 
  binary operator $\Adownp(\varphi,\psi)$ defined
  \begin{gather*}
    \wW, i \models^d \Adownp(\varphi,\psi) \text{~~~if{f}~~~for all $j
      \leq i$ such that $\wW,j \models^{\dD(j)} \psi$, we have $\wW,i
      \models^{\dD(j)} \varphi$.}
  \end{gather*}
  with $\psi \in \ltlurx$. This operator can be coded into
  $\ara(\opguess,\opspread)$, using the same technique as in
  Proposition~\ref{prop:ltl2ara}. The only difference is that instead of
  `saving' every data value in $q_{\text{save}}$, we use several
  states $q_{\text{save} (\psi)}$. Intuitively, only the data values
  that verify the test $\downarrow \psi$ are stored in
  $q_{\text{save} (\psi)}$. Then, a formula $\Adownp(\varphi,\psi)$ is
  translated as $\opspread(q_{\text{save}(\psi)},q_\varphi)$.
\end{remark}

\subsection{Ordered data}
If we consider a  linear order over $\D$ as done in Section~\ref{subsect:ara-guess-spread-order}, we can consider  $\ltluxae$ with richer tests  
\begin{align*}
  \varphi &\Coloneqq  \quad {\uparrow_>} \midd {\uparrow_<} \midd \dotsc
\end{align*}
that access to the linear order and compare the data values for $=,<,>$. The semantics are  extended accordingly, as in Figure~\ref{fig:semantics-ltl-order}.
\begin{figure}
  \centering
  \begin{align*}
    (\wW, i) \models^d \, \uparrow_> &\text{\quad if{f} \quad} d > \dD(i)\\ 
    (\wW, i) \models^d \, \uparrow_< &\text{\quad if{f} \quad} d < \dD(i)
  \end{align*}
  \caption{Semantics for the operators $\uparrow_>$, $\uparrow_<$ for a data word $\wW = \aA \otimes \dD$ and   $i \in\tpos(\wW)$.}
\label{fig:semantics-ltl-order}
\end{figure}
\newcommand{\ltluxaeo}{\mathit{ord}\text{-}\ltluxae}
Let us call this logic $\ltluxaeo$. The translation from $\ltluxaeo$ to $\ara(\opguess,\opspread,<)$ as defined in Section~\ref{subsect:ara-guess-spread-order} is straightforward. Thus we  obtain the following result.

\begin{proposition}
 The satisfiability problem for  $\ltluxaeo$ is decidable.
\end{proposition}

\part{Data trees}\label{part:trees}

The second part of this work deals with logics and automata for \emph{data trees}. In Section \ref{sec:atra-model}, we extend the model $\ara(\opguess, \opspread)$ to the new model $\atra(\opguess,\opspread)$ that runs over data trees (instead of data words). The decidability of the emptiness follows easily from the decidability result shown for $\ara(\opguess,\opspread)$. As in the case of data trees, this model allows to show the decidability of a logic.

In Section~\ref{sec:forward-xpath} we introduce `forward \xpath', a logic for \xml documents. The satisfiability problem for this logic will follow by a reduction to the emptiness problem of $\atra(\opguess,\opspread)$. This reduction is not as easy as that of the first part, since \xpath is closed under negation and our automata model is not closed under complementation. Indeed, $\atra(\opguess,\opspread)$ and forward \xpath have incomparable expressive power.

\medskip

\section{ATRA model}
\label{sec:atra-model}

Herein, we introduce the class of Alternating Tree Register Automata by slightly adapting the definition for alternating (word) register automata. This model is essentially the same automaton presented in Part~\ref{part:words}, that works on a (unranked, ordered) \emph{data tree} instead of a data word. The only difference is that instead of having one instruction $\triright$ that means `move to the next position', we have two instructions $\triright$ and $\tridown$ meaning `move to the next sibling to the right' and `move to the leftmost child'. This class of  automata is known as $\atra(\opguess,\opspread)$.
This model of computation will enable us to show decidability of a large fragment of \xpath. 

An \textbf{Alternating Tree Register Automaton} (\atra) consists in a top-down tree walking
automaton with alternating control and \emph{one} register to store and test data. \cite{JL08} shows that its emptiness problem is decidable and non-primitive-recursive. Here, as in the Part~\ref{part:words}, we consider an extension with the operators \opspread and \opguess.  We call this model $\atra(\opspread,\opguess)$.

\begin{definition}\label{def:atra-model}
  An alternating tree register automaton of $\atra(\opspread,\opguess)$
  is a tuple $\anAut=\tup{\A,Q,q_I,\delta}$ such that
  $\A$ is a finite alphabet;  $Q$ is a finite set of states; $q_I \in Q$ is the initial state;  and
 $\delta : Q \to \Phi$ is the transition function,
    where $\Phi$ is defined by the grammar
\begin{align*}
    a \midd \bar{a} \midd \odot ? \midd \opset(q) \midd \opeq \midd \opneq \midd q
    \land q' \midd q \lor q' \midd  \tridown q \midd \triright q \midd
    \opguess(q) \midd \opspread(q,q')
\end{align*}
    where $a\in \A, q,q' \in Q,$ $ \odot \in \set{\ttypea,$
    $\ttypena,$ $\ttypeb,$ $\ttypenb}$.
\end{definition}

We only focus on the differences with respect to the  \ara class. $\tridown$ and $\triright$ are to move to the leftmost child or to the next sibling to the right of the current position, and as before  `$\odot ?$' tests the current type of the position of the tree. For example, using $\ttypena?$ we test that we are in a leaf node, and by $\ttypeb?$ that the node has a sibling to its right. $\opset(q)$, $\opeq$ and $\opneq$ work in the same way as in the \ara model. We say that a state $q \in Q$ is \textbf{moving} if $\delta(q) = \triright q'$ or $\delta(q) = \tridown q'$ for some $q'\in Q$.

We define two sorts of configurations: \emph{node} configurations and \emph{tree} configurations.
In this context a \textbf{node configuration} is a tuple $\tup{x,\alpha,\gamma,\Threads}$  that describes the partial state of the execution at a position $x$ of the tree. $x \in \tpos(\tT)$ is the current position
in the tree $\tT$, $\gamma = \tT(x) \in \A \times \D$ is the current node's symbol and datum,
and $\alpha = \type_{\tT}(x)$ is the tree type of $x$. As before, $\Threads \in \subsetsf (Q \times \D)$ is a finite collection of execution threads. $\atranconf$ is the set of all node configurations. A \textbf{tree configuration} is just a finite set of node configurations, like
$\set{\tup{\epsilon,\alpha,\gamma,\Threads}, \tup{1211,\alpha',\gamma',\Threads'}, \dotsc }$. 
We call $\atraconf = \subsetsf(\atranconf)$ the set of all tree configurations.

We define the \emph{non-moving} relation $\rightarrow_\eps$ over node configurations just as in page \pageref{moving-nonmoving-rels}. As a difference with the $\ara$ mode, we have two types of moving relations. The \emph{first-child} relation $\rightarrowtridown$, to move to the leftmost child, and the \emph{next-sibling} relation $\rightarrowtriright$ to move to the next sibling to the right.

The $\rightarrowtridown$ and $\rightarrowtriright$ are defined,  for any  $\alpha_1 \in \set{\tridown,\bar{\tridown},\triright,\bar{\triright}}$,  $\gamma, \gamma_1 \in \A \times \D$, $h \in\set{\triright,\bar\triright}$, $v \in \set{\tridown,\bar\tridown}$, as follows
  \begin{align} \label{eq:def:atra:fs}
    \tup{x,(\ttypea,h),\gamma,\Threads} &\rightarrowtridown
    \tup{x\conc 1,\alpha_1,\gamma_1,\Threads_\tridown}, 
\\
    \tup{x\conc i,(v,\ttypeb),\gamma,\Threads} &\rightarrowtriright
    \tup{x \conc (i+1),\alpha_1,\gamma_1,\Threads_\triright}\label{eq:def:atra:ns}
  \end{align}
  if{f} (i) the configuration is `moving' (\ie,
    all the threads $(q,d)$ contained in $\Threads$ are of the form
    $\delta(q)=\tridown q'$ or $\delta(q)=\triright q'$); and (ii) for $\odot \in \set{\tridown,\triright}$,  $\Threads_\odot = \set{(q',d) \mid {(q,d)\in \Threads}, \delta(q)=\odot\,q'}$.

    Let ${\atraTn}  \;\coloneqq\;  {\rightarrow_{\eps} \cup \rightarrowtridown \cup \rightarrowtriright} \;\subseteq {\atranconf \times \atranconf}$. Note that through $\atraTn$ we obtain a run over a \emph{branch} of the tree (if we think about the underlying binary tree according to the first-child and next-sibling relations). In order to maintain all information about the run over all branches we need to lift this relation to tree configurations. We define the transition between tree configurations that we write $\atraT$. This corresponds to applying a `non-moving' $\rightarrow_\eps$ to a node configuration, or to apply a `moving' $\rightarrowtridown$, $\rightarrowtriright$, or both to a node configuration according to its type.  That is, we define $\cl{S}_1 \atraT \cl{S}_2$ if{f} one of  the following conditions holds:
  \begin{enumerate}[(1)]
  \item $\cl{S}_1=\set{\rho} \cup \cl{S}'$, \;
    $\cl{S}_2=\set{\tau} \cup \cl{S}'$, \;
    $\rho \rightarrow_\eps \tau$;
  \item $\cl{S}_1=\set{\rho} \cup \cl{S}'$, \;
    $\cl{S}_2=\set{\tau} \cup \cl{S}'$, \; 
    $\rho=\tup{x,(\ttypea,\ttypenb),\gamma,\Threads}$, \;
    $\rho \rightarrowtridown \tau$;
  \item $\cl{S}_1=\set{\rho} \cup \cl{S}'$, \;
    $\cl{S}_2=\set{\tau} \cup \cl{S}'$,\; 
    $\rho=\tup{x,(\ttypena,\ttypeb),\gamma,\Threads}$, \;
    $\rho \rightarrowtriright \tau$;
  \item $\cl{S}_1=\set{\rho} \cup \cl{S}'$, \;
    $\cl{S}_2=\set{\tau_1,\tau_2} \cup \cl{S}'$, \; 
    $\rho=\tup{x,(\ttypea,\ttypeb),\gamma,\Threads}$, \;
    $\rho \rightarrowtridown \tau_1$, \;
    $\rho \rightarrowtriright \tau_2$.
  \end{enumerate}

  A \emph{run}
  over a data tree $\tT=\aA \otimes\dD$ is a nonempty
  sequence  $\cl{S}_1 \atraT \dotsb \atraT \cl{S}_n$
  with $\cl{S}_1=\set{\tup{\epsilon,\alpha_0,\gamma_0,\Threads_0}}$ and $\Threads_0 = \set{(q_I,\dD(\epsilon))}$ (\ie, the thread consisting in the initial state with the root's datum), such that for every
  $i \in [n]$ and $\tup{x,\alpha,\gamma,\Threads} \in \cl{S}_i$: (1) $x\in\tpos(\tT)$; (2) $\gamma=\tT(x)$; and (3) $\alpha = \type_{\tT}(x)$.
As before, we say that the run is \emph{accepting} if 
  \[\cl{S}_n \subseteq \set{\tup{x,\alpha,\gamma,\emptyset} \mid  \tup{x,\alpha,\gamma,\emptyset} \in \atranconf }.\]



Note that the transition relation that defines the run replaces a node configuration in a given position by one or two node configurations in children positions of the fcns coding. Therefor the following holds.
  \begin{remark}\label{rem:no2nodeconfdescendant}
Every run is such that  any of its tree configurations never contains two node configurations in a descendant/ancestor relation of the fcns coding.
  \end{remark}

The \atra model is closed under all boolean operations~\cite{JL08}. However, the extensions introduced \opguess and \opspread, while adding expressive power, are not closed under complementation as a trade-off for decidability.
It is not surprising that the same properties as for the case of data words apply here.
\begin{proposition}\label{prop:atra-closed}
  $\atra(\opspread,\opguess)$ models have the following properties:
  \begin{enumerate}[\em(i)]
  \item they are closed under union,
  \item they are closed under
    intersection,
  \item they are not closed under complementation.
  \end{enumerate}
\end{proposition}

\begin{example}
  We show an example of the expressiveness that \opguess adds to \atra. Although as a corollary of Proposition~\ref{prop:guess-spread-more-expressive} we have that the $\atra(\opguess,\opspread)$ class is   more expressive than $\atra$, we give an example that inherently uses the tree structure of the model.
 We force  that the node at position $2$ and the node at position $1\conc 1$ of a data tree to have the same data value  without any further data constraints. Note that this datum does not
  necessarily has to appear at some common ancestor of these nodes.
  Consider the $\atra(\opguess)$ defined over $\A=\set{a}$ with
  \begin{align*}
    \delta(q_0) &= \opguess(q_1), & 
    \delta(q_1) &= \tridown q_2, &  
    \delta(q_2) &= q_3 \land q_4, \\
    \delta(q_3) &= \tridown q_5, & 
    \delta(q_4) &= \triright q_5, &
    \delta(q_5) &= \opeq .
  \end{align*}
For the data trees of Figure~\ref{fig:atra-vs-guess},  any \atra either accepts both, or rejects both. This is because when a thread is at position $1$, and performs a moving operation splitting into two threads, one at position $1\conc 1$, the other at position $2$, none of these configurations contain the data value $2$. Otherwise, the automaton should have read the data value $2$, which is not the case. But then, we see that the continuation of the run from the node configuration at position $2$ and at position $1$ are isomorphic independently of which data values we choose (either the data $2$ and $3$, or the data  $2$ and $2$). Hence, both trees are accepted or rejected. However, the $\atra(\opguess)$ we just built distinguishes them, as it can introduce the data value $2$ in the configuration, without the need of reading it from the tree. 
\begin{figure}
  \centering
  \includegraphics[scale=.55]{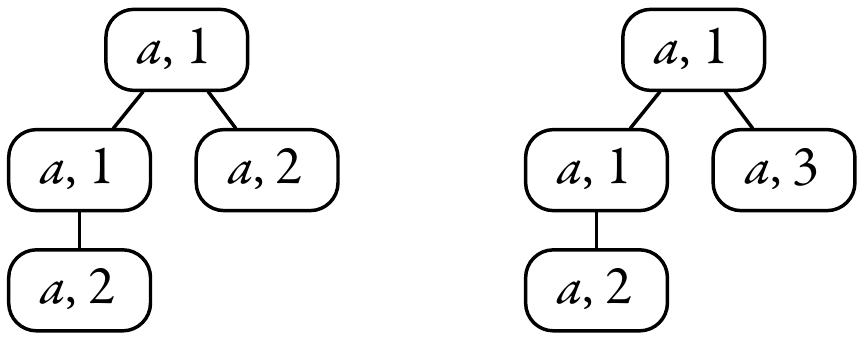}
  \caption{Two indistinguishable data trees for \atra.}
  \label{fig:atra-vs-guess}
\end{figure}
Equivalently, the property of  ``there are two leaves with the same data values'' is expressible in $\ara(\opguess)$ and not in $\ara$.
\end{example}

\subsection{Emptiness problem}
\label{sec:atrags-emptiness}

We show that the emptiness problem for this model is decidable, reusing the results of Part~\ref{part:words}. We remind the reader that the decidability of the emptiness of $\atra$ was proved in~\cite{JL08}. Here we extend the approach used for \ara and show the decidability of the two extensions \opspread and \opguess.

\begin{theorem} \label{thm:atrags-decidable}
  The emptiness problem of $\atra(\opguess,\opspread)$ is decidable.
\end{theorem}

\begin{proof}
The proof goes as follows. We will reuse the \wqo $\lqw$ used in Section~\ref{sec:emptiness-problem-ara-guess-spread}, which here is defined over the \emph{node configurations}. The only difference being that to use $\lqw$ over $\atranconf$ we work with \emph{tree} types instead of word types. Since $\rightarrow_\triright$ and $\rightarrow_\tridown$ are analogous, by the same proof as in Lemma~\ref{lem:nodeconf-rdc} we obtain the following.
\begin{lemma}\label{lem:atra-nodeconf-rdc}
$(\atranconf,\atraTn)$ is \rdc with respect to $(\atranconf,\lqw)$.
\end{lemma}

We now lift this result to \emph{tree configurations}.
%
We instantiate Proposition~\ref{prop:wsts-higman} by taking $\rightarrow_1$ as $\atraTn$, $\leq$ as $\lqw$, and taking $\lqdom$ the majoring order over $(\atranconf,\lqw)$.
We take $\rightarrow_2$ to be $\atraT{}$ as it verifies the hypothesis demanded in the Lemma. As a result we obtain the following.
\begin{lemma}
  $(\atraconf,\atraT)$ is \rdc with respect to $(\atraconf,\lqdom)$.
\end{lemma}
Hence, condition (1) of Proposition~\ref{prop:rdc-computable} is met. 
Let us write $\anEquivRel$ for the equivalence relation over $\atraconf$ such that $\+S \anEquivRel \+S'$ if{f} $\+S \lqdom \+S'$ and $\+S' \lqdom \+S$.
Similarly as for Part~\ref{part:words}, we have that $(\atraconf/{\anEquivRel},\atraT)$ is finitely branching and effective. That is, the $\atraT$-image of any configuration has only a finite number of configurations up to isomorphism of the data values contained (remember that only equality between data values matters), and representatives for every class are computable. Hence, we have that $(\atraconf/{\anEquivRel},\atraT,\lqdom)$ verifies condition (2) of Proposition~\ref{prop:rdc-computable}. Finally, condition (3) holds as $(\atraconf/{\anEquivRel},\lqdom)$ is a \wqo (by Proposition~\ref{prop:lqdom-wqo}) that is a computable relation.
We conclude the proof by the following obvious statement.
\begin{lemma}
  The set of accepting tree configurations is downwards closed with respect to $\lqdom$.
\end{lemma}
Hence, by Lemma~\ref{lem:downward-closed-decidable}, we conclude as before that the emptiness problem for the class $\atra(\opguess,\opspread)$ is decidable.
\end{proof}

\section{Forward-XPath}
\label{sec:forward-xpath}
\label{sec:definitionsxpath}

We consider a navigational fragment of \xpath $1.0$ with data equality and inequality. In particular this logic is here defined over \emph{data trees}. However, an \xml document may typically have not \emph{one} data value per node, but a set of \emph{attributes}, each carrying a data value. This is not a problem since every attribute of an \xml element can be encoded as a child node in a data tree labeled by the attribute's name (\cf\ Section~\ref{sec:dtrees-xml}). Thus, all the decidability results hold also for \xpath with attributes over \xml documents.

Let us define a simplified syntax for this logic. \xpath is a two-sorted language, with \emph{path} expressions ($\alpha, \beta, \dotsc$) and \emph{node} expressions ($\varphi, \psi, \dotsc$). We write $\xpath(\+{O},=)$ to denote the data-aware fragment with the set of axes $\+{O} \subseteq \set{\down,\td,\rightarrow,\tr,\leftarrow,\tl, \uparrow, \tu\,}$. It is defined by mutual recursion as follows,
\begin{align*}
\alpha, \beta \;&\Coloneqq\;  o \midd [\varphi] \midd \alpha\beta \midd \alpha \cup \beta && o \in \+{O}\cup\set{\varepsilon}\\
\varphi, \psi \;&\Coloneqq\; a \midd \lnot \varphi \midd \varphi \lor \psi \midd \varphi \land \psi \midd \tup{\alpha} \midd \tup{\alpha = \beta} \midd \tup{\alpha \not= \beta} && a \in \A
\end{align*}
where $\A$ is a finite alphabet. A \emph{formula} of $\xpath(\+{O},=)$ is either a node expression  or a path expression. We define the `forward' set of axes as  $\mathfrak F \coloneqq  \set{\down,\td,\rightarrow,\tr}$, and consequently the fragment `forward-\xpath' as $\xpath(\mathfrak F,=)$. We also refer by $\xpatheps(\mathfrak F,=)$ to the fragment considered in \cite{JL08} where data tests are of the restricted form $\tup{\eps = \alpha}$ or $\tup{\eps \not= \alpha}$.\footnote{\cite{JL08} refers to $\xpatheps(\mathfrak F,=)$ as `forward \xpath'. Here, `forward \xpath' is the unrestricted fragment $\xpath(\mathfrak F,=)$, as we believe is more appropriate.}

There have been efforts to extend \xpath to have the full expressivity of MSO, e.g. by adding a least fix-point operator (\cf\ \cite[Sect.~4.2]{tC06}), but these logics generally lack clarity and simplicity. However, a form of recursion can be added by means of the Kleene star, which allows us to take the transitive closure of any path expression. Although in general this is not enough to already have MSO \cite{tCS08}, it does give an intuitive language with a counting ability. By $\rxpath(\+{O},=)$ we refer to the enriched language where path expressions are extended by allowing the Kleene star on \emph{any} path expression.
\[\alpha, \beta \;\Coloneqq\; o \midd [\varphi] \midd \alpha\beta \midd \alpha \cup \beta \midd \alpha^* \qquad o \in \+{O}\cup\set{\varepsilon}\]

\newcommand{\model}{\ensuremath{\cl{T}}\xspace}
Let $\tT$ be a data tree. The semantics of \xpath is defined as the set of elements (in the case of node expressions) or pairs of elements (in the case of path expressions) selected by the expression. The data aware expressions are the cases $\tup{\alpha = \beta}$ and $\tup{\alpha \not=\beta}$. The formal  definition of its semantics is in Figure~\ref{fig:xpath-semantics}. We write $\tT \models \varphi$ to denote $\dbracket{\varphi}^{\tT}\not=\emptyset$, and in this case we say that $\tT$ \emph{satisfies} $\varphi$.

\begin{figure}
\newcommand{\indspace}{\mathrel{\phantom{=}}\mathop{\phantom{\{}}}
\small
  \centering
  \begin{align*}
   \abracket{\varepsilon}^\tT & = \set{(x,x) \mid x \in \tpos(\tT)} & \abracket{\downarrow}^\tT & = \set{(x,x\conc i) \mid x \conc i \in
    \tpos(\tT)} \\ 
\abracket{\alpha \cup \beta}^\tT & =
    \abracket{\alpha}^\tT \cup \abracket{\beta}^\tT&
\abracket{\rightarrow}^\tT & = \set{(x \conc i,x\conc (i+1)) \mid x \conc (i+1) \in
    \tpos(\tT)}\\ 
\abracket{\varphi \land \psi}^\tT & = \abracket{\varphi}^\tT \cap
    \abracket{\psi}^\tT&
\abracket{\alpha^*}^\tT & =
    \textrm{the reflexive transitive closure of
    }\abracket{\alpha}^\tT\\
    \abracket{a}^\tT & = \set{ x \in \tpos(\tT) \mid
    \aA(x) = a }    &
\abracket{\alpha \beta}^\tT & = \{(x,z) \mid
    \text{ there exists } y 
\text{ such that }\\
    \abracket{\lnot \varphi}^\tT & =  \tpos(\tT) \setminus \abracket{\varphi}^\tT&&\indspace (x,y) \in
    \abracket{\alpha}^\tT, (y,z) \in \abracket{\beta}^\tT\}    \\
     \abracket{\tup{\alpha}}^\tT & = \{ x \in \tpos(\tT) \mid \exists
    y. (x,y) \in \abracket{\alpha}^\tT \}
&\abracket{[\varphi]}^\tT & = \{(x,x)  \mid x \in \abracket{\varphi}^\tT\}
\\
 \abracket{\tup{\alpha {=} \beta}}^\tT 
& = \{ x \in \tpos(\tT) \mid \exists y, z. (x,y) \in
    \abracket{\alpha}^\tT,
&
\abracket{\tup{\alpha {\not=} \beta}}^\tT & =  \{ x \in \tpos(\tT) \mid \exists y, z. (x,y) \in \abracket{\alpha}^\tT,
    \\ 
&\indspace (x,z) \in
    \abracket{\beta}^\tT, \dD(y)=\dD(z)\}&
&\indspace (x,z) \in \abracket{\beta}^\tT, \dD(y)\not=\dD(z)\}
  \end{align*}
  \caption{Semantics of $\rxpath(\mathfrak F,=)$ for a data tree $\tT = \aA \otimes \dD$.}
\label{fig:xpath-semantics}
\end{figure}

\begin{example}
In the model of Figure~\ref{fig:data-tree} on page \pageref{fig:data-tree}, \[\dbracket{\;\tup{\;\td\![b\;\land \tup{\down\![b]\not=\down\![b]}]\;}\;}^\tT=\set{\epsilon,\, 1,\, 1{\conc} 2}.\]
\end{example}

We define $\subf(\varphi)$ to denote the set of all substrings of $\varphi$ which are formul{\ae},
 $\psubf(\varphi)\coloneqq \{ \alpha \mid \alpha \in \subf(\varphi), \alpha$ is a path expression$\}$, and
$\nsubf(\varphi)\coloneqq \{ \psi \mid \psi \in \subf(\varphi), \psi$ is a node expression$\}$.

\subsection{Key constraints}
It is worth noting that $\xpath(\mathfrak F,=)$ ---contrary to $\xpatheps(\mathfrak{F},=)$--- can express unary \emph{key constraints}. That is, whether for some symbol $a$, all the $a$-elements in the tree have different data values. 
\begin{lemma}\label{lem:primary-key}
For every $a \in \A$  let $key(a)$ be the property over a data tree $\tT=\aA \otimes \dD$: ``\textsl{For every two different positions $x,x' \in \tpos(\tT)$ of the tree, if $\aA(x)=\aA(x')=a$, then $\dD(x)\not=\dD(x')$}.'' Then, $key(a)$ is expressible in $\xpath(\mathfrak F, =)$ for any $a$.
\end{lemma}
\begin{proof}
  It is easy to see that the \emph{negation} of this property can be
  tested by first \emph{guessing} the closest common ancestor of two
  different $a$-elements with equal datum in the underlying
  first-child next-sibling binary tree. At this node, we
  verify the presence of two $a$-nodes with equal datum, one
  accessible with a ``$\td\,$'' relation and the other with a compound
  ``$\trs \, \td\,$'' relation (hence the nodes are different). The
  expressibility of the property then follows from the logic being
  closed under negation. The reader can check that the following formula expresses the property
  \begin{gather*}
    ~~~~~~key(a) \equiv \lnot \tup{\;\td [\, \tup{\eps[a] = \tds[a]} \lor \tup{\td[a] = \trs\td[a]} \, ]\;}~~~~~~ 
  \end{gather*}
where `$\tds$\,' $=$ `$\down\td$\,' and `$\trs$\,' $=$ `$\rightarrow \tr$\,'.
\end{proof}

Note that while \atra cannot express, for instance, that there are two different nodes with the same data value, $\atra(\opguess)$ can express it. But on the other hand, $\atra(\opguess,\opspread)$ cannot express the \emph{negation} of the property.
\begin{lemma}\label{lem:atrags-cannot-leaves}
 The class  $\atra(\opguess,\opspread)$ cannot express the property ``all the data values of the data tree are different''.
\end{lemma}
\begin{proof}
Ad absurdum, suppose that there exists an automaton $\anAut$ expressing the property, and let $Q$ be its set of states.
Notice that for any accepting run of minimal length there are no more that $f(|Q|)$ consecutive $\epsilon$-transitions (\ie, $\atraT$ transitions that are fired by an underlying $\rightarrow_\epsilon$ transition on node configurations), for some fixed function $f$.

Consider a data tree $\tT$ with 
\[\tpos(\tT) = \set{\epsilon} \cup \set{1,2, 3, \dotsc, N} \cup \set{11, 111, \dotsc, \underbrace{111 \dotsb 1}_{\text{$N$ times}}}\] 
for $N= 2\cdot f(|Q|) +4$. That is, the root of the tree has $N$ children, and the first child has a long branch with $N$ nodes. All positions of $\tT$ have different data values, and they all  carry the same label, say $a$. 

Consider a minimal accepting run of $\anAut$ on $\tT$, $\cl{S}_1 \atraT \dotsb \atraT \cl{S}_n$. Let $\cl{S}_i$ be the first tree configuration containing a node configuration with position $2$. That is, the first configuration after the second moving transition. Since there are at most $2\cdot f(|Q|)$ non-moving transitions before $\cl{S}_i$ in the run, we have that $i \leq 2 \cdot f(|Q|) +2$. In particular, this means that $\cl{S}_i$ cannot contain more than $2 \cdot f(|Q|) +2$ different data values. Therefore, there is a  position $x$ with $2 \leq x \leq N$ and a position $y$ with $y \descendant 1$ such that neither $\dD(x)$ or $\dD(y)$ are in $\data(\cl{S}_i)$. (Remember that by definition of $\tT$, $\dD(x) \neq \dD(y)$.) This is because there are $N-1 = 2\cdot f(|Q|) +3$ different possible data values of $x$ and for $y$.

Consider $\tT'$ as the result of replacing $\dD(x)$ by $\dD(y)$ in $\tT$. Clearly, $\tT'$ does not have the property ``all the data values of the data tree are different''. Now, consider the run obtained as the result of replacing $\dD(x)$ by $\dD(y)$ in $\cl{S}_1, \dotsc, \cl{S}_n$. Note that this is still a run, and it is still accepting. Therefore $\anAut$ accepts $\tT$ and thus $\anAut$ does not express the property.
\end{proof}

\subsection{Satisfiability problem}
\label{sec:reduction-atra}
This section is mainly dedicated to the decidability of the satisfiability problem for $\xpath(\mathfrak F, =)$, known as `forward-\xpath'. This is proved by a reduction to the emptiness problem of the automata model $\atra(\opguess,\opspread)$ introduced in Section~\ref{sec:atra-model}. 

\cite{JL08} shows that \atra captures the fragment \mbox{$\xpatheps(\mathfrak{F},=)$}.
 It is immediate to see that
\atra can also easily capture the Kleene star operator on any path
formula, obtaining decidability of
$\rxpatheps(\mathfrak{F},=)$. However, these decidability results
cannot be further generalized to the full unrestricted forward fragment
$\xpath(\mathfrak F,=)$ as \atra is not powerful enough to capture the full expressivity of the logic. Indeed, while $\xpath(\mathfrak F,=)$ can express that all the data values of the leaves are different, $\atra(\opguess,\opspread)$ cannot (Lemma~\ref{lem:atrags-cannot-leaves}).
Although $\atra(\opguess,\opspread)$ cannot capture $\xpath(\mathfrak F,=)$, in the sequel we show that there exists a
reduction from the satisfiability of
$\rxpath(\mathfrak{F},=)$ to the emptiness of $\atra(\opguess,\opspread)$, and hence that the former problem is decidable.
This result settles an open question 
regarding the decidability of the satisfiability problem for the 
forward-\xpath fragment $\xpath(\mathfrak{F},=)$.
The main results that will be shown in Section~\ref{sec:allow-arbitr-data} are the following.
\begin{theorem}\label{thm:fxpath-dec}
Satisfiability of $\rxpath(\mathfrak F,=)$ in the presence of DTDs (or any regular language) and unary key constraints is decidable, non-primitive-recursive. 
\end{theorem}
And hence the next corollary follows from the logic being closed under boolean operations.
\begin{corollary}
  The query containment and the query equivalence problems are decidable for $\xpath(\mathfrak F,=)$.
\end{corollary}

Moreover, these decidability results hold for $\rxpath(\mathfrak F,=)$ and even for two extensions:
\begin{iteMize}{$\bullet$}
\item a navigational extension with \emph{upward} axes (in Section~\ref{sec:allowing-upward-axes}), and
\item a generalization of the data tests that can be performed (in Section~\ref{sec:allow-strong-data}).
\end{iteMize}

\subsection{Data trees and {XML} documents}
\label{sec:dtrees-xml}



Although our main motivation for working with trees is related to static analysis of logics for \xml documents, we work with \emph{data trees}, being a simpler formalism to work with, from where results can be transferred to the class of \xml documents. We discuss briefly how all the results we give on \xpath over data trees, also hold for the class of \xml documents.

\begin{figure}
  \centering
    \includegraphics[scale=.5]{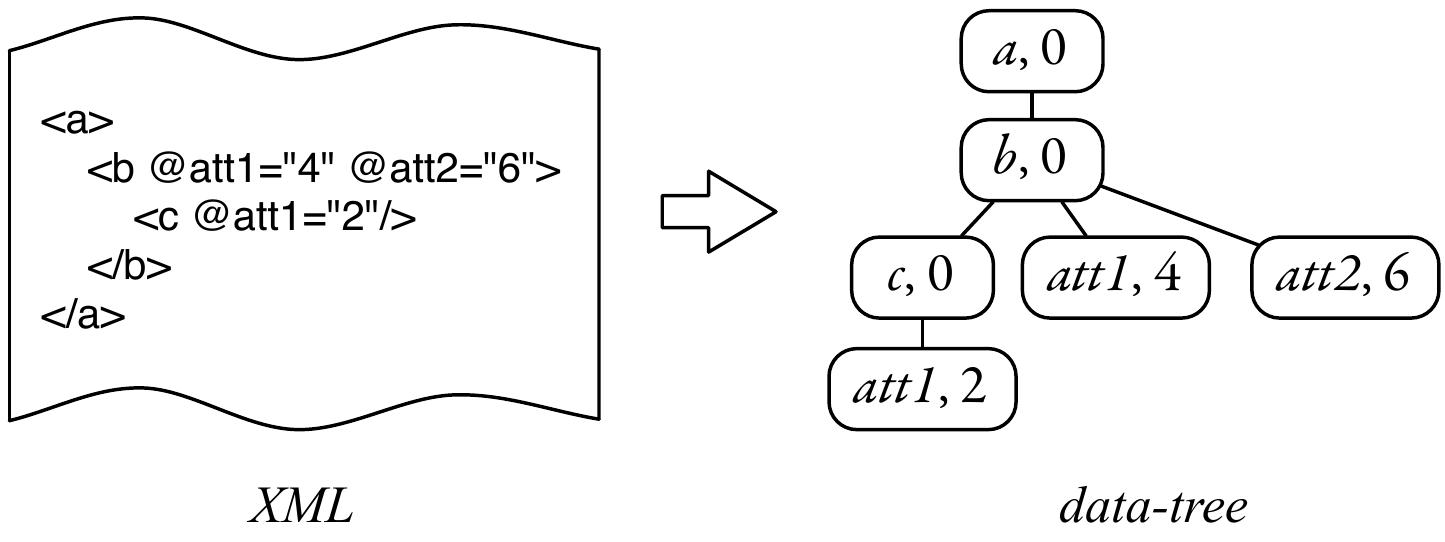}
  \vspace{-1mm}  
  \caption{From XML documents to data-trees.}
  \label{fig:xml2datatree}
  \vspace{-3mm}  
\end{figure}

While a data tree has \emph{one} data value for each node, an \xml document may have \emph{several} attributes at a node, each with a data value. However, every attribute of an \xml element can be encoded as a child node in a data tree labeled by the attribute's name, as in Figure~\ref{fig:xml2datatree}. This coding can be enforced by the formalisms we present below, and we can thus transfer all the decidability results to the class of \xml documents. In fact, it suffices to demand that all the attribute symbols can only occur at the leaves of the data tree and to interpret attribute expressions like `$\mathit{@attrib1}$' of \xpath formul{\ae} as child path expressions `$\down[\mathit{attrib1}]$'.

\subsection{Decidability of forward XPath}
\label{sec:allow-arbitr-data}
This section is devoted to the proof of the following statement.
\begin{proposition}\label{prop:rxpathRedAtra}
  For every $\eta \in \rxpath(\mathfrak{F},=)$ there is a computable $\atra(\opguess,\opspread)$ automaton \anAut such
  that \anAut is nonempty if{f} $\eta$ is satisfiable.
\end{proposition}
Markedly, the $\atra(\opguess,\opspread)$ class does not capture  $\rxpath(\mathfrak F, =)$. However, given a formula $\eta$, it is possible to construct an automaton that tests a property that guarantees the existence of a data tree  verifying $\eta$.

\subsection*{Disjoint values property}

To show the above proposition, we need to work with runs with the \emph{disjoint values property} as stated next.

\begin{definition}
  A run $\cl{S}_1 \atraT \dotsb \atraT \cl{S}_n$ on a data tree $\tT$ has the \textbf{disjoint values property} if for every $x \conc i \in \tpos(\tT)$ and $\rho$ a \emph{moving} node configuration of the run with position $x\conc i$, then
  \begin{gather*}
    \data(\tT|_{x \conc i }) ~~\cap~ \bigcup_{\substack{x \conc
        j \in \tpos(\tT)\\ j > i}}\data(\tT|_{x \conc j})
    ~~~\subseteq~~~ \data(\rho) \ .
  \end{gather*}
Figure~\ref{fig:atra-disjoint} illustrates this property.
\end{definition}

\begin{figure}
  \centering
  \includegraphics[width=3.2cm]{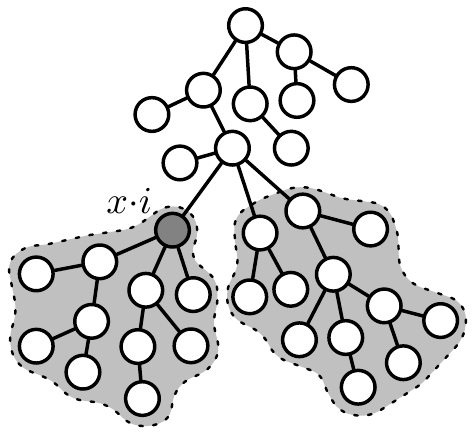}
  \caption{The disjoint values property states that for every position $x \conc i$, the intersection of the grey zones is present in the last configuration for $x \conc i$ appearing in the run.}
  \label{fig:atra-disjoint}
\end{figure}

The proof of Proposition~\ref{prop:rxpathRedAtra} can be sketched as follows:
\begin{enumerate}[(1)]
\item We show that for every nonempty automaton $\anAut \in \atra(\opguess,\opspread)$ there is an accepting run on a data tree with the disjoint values property.
\item We give an effective translation from an arbitrary forward \xpath formula $\eta$ to an $\atra(\opguess,\opspread)$ automaton $\anAut$ such that 
  \begin{iteMize}{$\bullet$}
  \item  any tree accepted by a run of the automaton $\anAut$ with the disjoint values property verifies the \xpath formula $\eta$, and
  \item  any tree verified by the formula $\eta$ is accepted by a run of the automaton \anAut with the disjoint values property.
  \end{iteMize}
\end{enumerate}

We start by proving the disjoint values property normal form.

\begin{proposition}\label{prop:normal-form-disjoint-values-properties}
  For any nonempty automaton $\anAut \in \atra(\opguess,\opspread)$  there exists an accepting run over some data tree with the disjoint values property.
\end{proposition}
\begin{proof}
\newcommand{\ancestorb}{\ancestor^{\textit{fcns}}}
\newcommand{\ancestorstrb}{\ancestorstr^{\textit{fcns}}}
\newcommand{\descendantb}{\descendant^{\textit{fcns}}}
\newcommand{\descendantstrb}{\descendantstr^{\textit{fcns}}}
Given any accepting run $\cl{S}_1 \atraT \dotsb \atraT \cl{S}_n$ on a data tree $\tT = \aA \otimes \dD$, we show how to modify the run and the tree in order to satisfy the disjoint values property. We only need to replace some of the data values, so that the resulting tree and accepting run will be essentially the same. 

The idea is as follows. For a given position $x \conc k$ of the tree, we consider the moving node configuration of the run at position $x \conc k$, and we replace all the data values from $\tT|_{x \conc k}$ by fresh ones except for those present in the node configuration, obtaining a new data tree $\tT'$. We also make the same replacement of data values for all node configurations in the run of nodes below $x \conc k$. Thus, we end up with a modified data tree $\tT'$ and accepting run  that satisfy the disjoint values property at $x \conc k$. That is, such that 
\[\data(\tT'|_{x \conc k }) ~~\cap~ \bigcup_{\substack{x \conc
        j \in \tpos(\tT)\\ j > k}}\data(\tT'|_{x \conc j})
    ~~~\subseteq~~~ \data(\rho_{x \conc k})\]
where $\rho_{x \conc k}$ is the moving node configuration at position $x \conc k$ of the run. If we repeat this procedure for all nodes of the tree, we obtain a run and tree with the disjoint values property. Next, we formalize this transformation.

\smallskip

Take any $x \conc k \in \tpos(\tT)$, and let $\rho_{x\conc k} \in \cl{S}_i$ for some $i$ be such that $\rho_{x \conc k}$ is a moving node configuration with position $x \conc k$. Consider any injective function 
\begin{gather*}
  f ~:~ \data(\tT|_{x \conc k}) \setminus \data(\rho_{x\conc k})
  \,\to\, \D \setminus \data(\tT)
\end{gather*}
and let $\hat f : \data(\tT|_{x \conc k}) \to \D$ be such that $\hat f(d) = d$ if $d \in \data(\rho_{x\conc k})$, or $\hat f (d) = f(d)$ otherwise. Note that $\hat f$ is injective.
Let us consider then $\cl{S}'_1, \dotsc, \cl{S}'_n$ where $\cl{S}'_j$ consists in replacing every node configuration $\rho \in \cl{S}_j$ with $h(\rho)$, where\footnote{By $\hat f (\rho)$ we denote the replacement of every data value $d$ by $\hat f (d)$ in $\rho$.}
\begin{align*}
  h(\rho) &=
  \begin{cases}
    \hat f(\rho) & \text{if $\rho$ has a position $y \descendantstr x \conc k$} \\
    \rho & \text{otherwise.}
  \end{cases}
\end{align*}
Take $\tT'$ to be  the data tree that results from the replacement in $\tT$ of every data value of a position $y \descendantstr x \conc k$  by $\hat f(\dD(y))$.

\begin{claim}
\[\data(\tT'|_{x \conc k }) ~~\cap~ \bigcup_{\substack{x \conc
        j \in \tpos(\tT)\\ j > k}}\data(\tT'|_{x \conc j})
    ~~~\subseteq~~~ \data(\rho_{x \conc k}) \ .\]
\end{claim}
\begin{proof}
The only difference between $\tT$ and $\tT'$ is in the data values below $x \conc k$. Any  node $y$ of $\tT'$ which is below $x\conc k$ contains a data value $\hat f (\dD(y))$ such that
\begin{iteMize}{$\bullet$}
\item $\hat f (\dD(y))$ is in $\data(\rho_{x \conc k})$, or
\item $\hat f (\dD(y))$ is not in $\data(\tT)$, and therefore it is not in $\bigcup_{\substack{x \conc j \in \tpos(\tT)\\ j > k}}\data(\tT'|_{x \conc j})$.
\end{iteMize}
\end{proof}
\begin{claim}
  $\cl{S}'_1 \atraT \dotsb \atraT \cl{S}'_n$ is an accepting run of
  $\anAut$ over $\tT'$.
\end{claim}
\begin{proof}
  Take any leaf $y$ which is rightmost (\ie, with no siblings to its
  right) and consider the sequence of node configurations $\rho_{1}
  \in \cl{S}_{1}, \dotsc, \rho_n \in \cl{S}_n$ that are ancestors of
  $y$ in the first-child next-sibling underlying tree structure. There is exactly one configuration in each $\cl S_i$ due to Remark~\ref{rem:no2nodeconfdescendant}. This is the `sub-run' that leads to $y$: for every $\rho_i, \rho_{i+1}$ either $\rho_i = \rho_{i+1}$ or $\rho_i \atraTn \rho_{i+1}$. Let    $\rho'_j =h(\rho_j)$ for all $j \in [n]$.
  \begin{iteMize}{$\bullet$}
  \item If $y \descendant x \conc k$, take $\ell$ to be the last    index such that $\rho_{\ell} = \rho_{x \conc k}$ (there must be one). Note that $\rho'_j = \hat f(\rho_j)$ for every $\ell < j \leq n$. The sequence $\rho'_{\ell+1} \in \cl{S}'_{\ell+1}, \dotsc,   \rho'_n \in \cl{S}'_n$ is isomorphic, modulo renaming of data
    values, to $\rho_{\ell+1}, \dotsc, \rho_{n}$ since $\hat f$ is
    injective.  We then have that $\rho'_1, \dotsc, \rho'_{\ell}, \dotsc
    \rho'_n$ is a correct run on node configurations, since
    \begin{iteMize}{$-$}
    \item $\rho'_1, \dotsc, \rho'_{\ell}$ is equal to $\rho_1,
      \dotsc, \rho_{\ell}$ (it is not modified by $h$),
    \item $\rho'_{\ell+1}, \dotsc, \rho'_n$ is isomorphic to
      $\rho_{\ell+1}, \dotsc, \rho_n$ (we apply an injection $\hat f$      to every data value), and
    \item the pair $(\rho'_{\ell}, \rho'_{\ell+1})$ is isomorphic to      $(\rho_{\ell}, \rho_{\ell+1})$, as 
      \begin{iteMize}{$*$}
      \item $\rho'_\ell = f(\rho_{x \conc k}) = \rho_{x \conc k} = \rho_\ell$,
      \item $\rho'_{\ell+1} = \hat f (\rho_{\ell+1})$ is isomorphic to $\rho_{\ell+1}$, and
      \item $\data(\rho'_\ell) \cap \data(\rho'_{\ell+1}) = \data(\rho_\ell) \cap \data(\rho_{\ell+1})$ since $\hat f$ is the identity on $\data(\rho_\ell)$, and does not send any data value from $\D \setminus \data(\rho_\ell)$ to $\data(\rho_\ell)$.
      \end{iteMize}
    \end{iteMize}
  \item If $y \not\descendant x \conc k$, then nothing
    was modified: $\rho'_1=\rho_1 \in \cl{S}_1, \dotsc, \rho'_n=\rho_n
    \in \cl{S}_n$.
  \end{iteMize}
In any case, we have that $\rho'_1, \dotsc,
    \rho'_n$ is a correct run on node configurations.
This means that the modified data values are innocuous for the run. As the structure of the run is not changed, this implies that
$\cl{S}'_1 \atraT \dotsb \atraT \cl{S}'_n$
is an accepting run (that verifies the disjoint values property for $x \conc k$ by the previous Claim).
\end{proof}
 If we perform the same procedure for every position of the tree, we end up with an accepting run and tree with the disjoint values property.
\end{proof}

We define a translation from \mbox{$\rxpath(\mathfrak{F},=)$} formul{\ae} to
$\atra(\opguess,\opspread)$.  Let $\eta$ be a
$\rxpath(\mathfrak{F},=)$ formula and let $\anAut$ be the corresponding $\atra(\opguess,\opspread)$ automaton defined by the translation. We show that (i) if a
data tree $\tT$ is accepted by $\anAut$ by a run verifying the disjoint values property, then $\tT\models\eta$, and in turn (ii) if $\tT \models \eta$, then $\tT$ is accepted by $\anAut$. Thus, by the disjoint values normal form (Proposition~\ref{prop:normal-form-disjoint-values-properties}) we obtain our main result of Proposition~\ref{prop:rxpathRedAtra}, which by decidability of $\atra(\opguess,\opspread)$ (Theorem~\ref{thm:atrags-decidable}) implies that the satisfiability problem for $\rxpath(\mathfrak F,=)$ is decidable.

\subsection*{Normal form}
For succinctness and simplicity of the translation, we assume that $\eta$ is in a normal form such that the $\downarrow$-axis is interpreted as the \emph{leftmost} child. To obtain this normal form, it suffices to
replace every appearance of `$\downarrow$' by `$\downarrow\rightarrow^*$'. Also, we assume that every data test subformula of the form $\tup{\alpha = \beta}$ or $\tup{\alpha \neq \beta}$ is such that $\alpha = \varepsilon$, $\alpha = \down \gamma$, or $\alpha = \rightarrow \gamma$ for some $\gamma$, and idem for $\beta$. This will simplify the translation and the proofs for correctness. It is easy to see that every expression $\tup{\alpha = \beta}$ can be effectively transformed into a disjunction of formul{\ae} of the aforementioned form.

\subsection*{The translation}
\newcommand{\tran}{\ensuremath{\msf{tr}}\xspace}
\newcommand{\llpar}{\ensuremath{(\!\hspace{-.4pt}|}\xspace}
\newcommand{\rrpar}{\ensuremath{|\!\hspace{-.4pt})}\xspace}
Let $\eta$ be a $\rxpath(\mathfrak{F},=)$ node expression of the above form in negation normal form (\msf{nnf} for short).  
Let $\tT$ be a data tree and let $z_0 \ancestor_{fcns} z_n$ be two positions of $\tT$  at distance $n \in \Nz$ such that, if $z_0 \neq z_n$,
\[z_0\ancestorstr_{fcns} z_1 \ancestorstr_{fcns} \dotsb \ancestorstr_{fcns} z_{n-1} \ancestorstr_{fcns} z_n\]
for some $z_1, \dotsc, z_{n-1}$, we define $str(z_0,z_{n}) \in \A_\eta^*$ to be the string $a_1 S_1 \dotsb a_{n-1} S_{n-1}$ where $a_i$ is either $\rig$ or $\dow$ depending on the relation between $z_{i-1}$ and $z_i$, and $S_i = \set{ \psi \in \nsubf(\eta) \mid z_i \in \dbracket{\psi}^\tT}$. If $n=0$, then $str(z_0,z_0)=\epsilon$. Note that $str(z_0,z_n)$ does not take into account $z_0$, since we are working under the normal form where no path subformul{\ae} makes tests on the node where they start. 
For every path expression $\alpha \in \psubf(\eta)$, consider a deterministic complete finite automaton $\cl{H}_\alpha$ over the alphabet $\A_\eta = 2^{\nsubf(\eta)} \cup \set{\downarrow, \rightarrow}$ which corresponds to that regular expression, in the sense that the following claim holds. 
\begin{claim}\label{cl:paths}
  For every $\alpha \in \psubf(\eta)$, and $x \ancestor_{fcns} y$, we have that $\cl H_\alpha$ accepts $str(x,y)$ if, and only if, $(x,y) \in \dbracket{\alpha}^\tT$.
\end{claim}
We assume the following names of its components: $\cl{H}_\alpha = \tup{\A_\eta, \delta_\alpha, Q_\alpha,0, F_\alpha}$, where $Q_\alpha \subseteq \Nz$ is the finite set of states and $0 \in Q_\alpha$ is the initial state. We assume that $Q_\alpha$ is partitioned into \emph{moving} and \emph{testing} states, such that in every accepting run the states of the run alternate between moving and testing states, starting in the moving state $0$.

We next show how to translate $\eta$ into an $\atra(\opguess,\opspread)$ automaton $\anAut$. For the sake of readability we define the transitions as positive boolean combinations of $\lor$ and $\land$ over the set of basic tests and  states. Any of these ---take for instance $\delta(q)=(\opset( q_1) \land \tridown q_2) \lor (q_3 \land \bar{\texttt{a}})$--- can be rewritten into an equivalent \atra with at most one boolean connector per transition (as in Definition~\ref{def:atra-model}) in polynomial time. The most important cases are those relative to the following data tests: 
\newcommand{\itemTeq}{(1)}
\newcommand{\itemTeqn}{(2)}
\newcommand{\itemTneq}{(3)}
\newcommand{\itemTneqn}{(4)}
\[
\text{\itemTeq\; }\tup{\alpha=\beta}
\qquad\quad
\text{\itemTeqn\; }\tup{\alpha\not=\beta}
\qquad\quad
\text{\itemTneq\; }\lnot \tup{\alpha = \beta}
\qquad\quad
\text{\itemTneqn\; }\lnot \tup{\alpha \not= \beta}
\]

We define the $\atra(\opguess,\opspread)$ automaton 
\[\anAut \coloneqq  \tup{\A,Q,\llpar \eta \rrpar,\delta}\]
with
\begin{align*}
  Q &\coloneqq \{\llpar \varphi \rrpar, 
  \llpar \alpha \rrpar^{\circledast}_{i},
  \llpar \alpha \rrpar^{\circledast}_{\text{test},i},
 \llpar \alpha  \rrpar^{\circledast}_{\msf{F}}, 
   \llpar \alpha, \beta \rrpar^{\circledast}_{i,j},
   \llpar \alpha, \beta \rrpar^{\circledast}_{\text{test},i,j}
\mid \varphi
  \in\nsubf^\lnot(\eta), 
\\&\fudgece \phantom{\{} \alpha,\beta \in \psubf^\lnot(\eta),
  \circledast \in \{=,\not=,{\lnot}{=},{\lnot}{\not=}\}, 
i \in  Q_\alpha, j \in Q_\beta \}
\end{align*}
where $\mathsf{op}^\lnot$ is the smallest superset of $\mathsf{op}$ closed under negation under \nnf, \ie, if $\varphi \in \mathsf{op}^\lnot(\eta)$ then $\nnf(\lnot\varphi)\in\mathsf{op}^\lnot(\eta)$ (where $\nnf$ is defined as shown in Table~\ref{tab:nnf-definition}).
\begin{table}
  \centering
  \begin{align*}
    \nnf(\varphi \land \psi) &\coloneqq  \nnf(\varphi) \land \nnf(\psi) &
    \nnf(\varphi \lor \psi) &\coloneqq  \nnf(\varphi) \lor \nnf(\psi)\\
    \nnf(\lnot(\varphi \land \psi)) &\coloneqq  \nnf(\lnot \varphi) \lor \nnf(\lnot \psi)&
    \nnf(\lnot(\varphi \lor \psi)) &\coloneqq  \nnf(\lnot \varphi) \land \nnf(\lnot \psi)\\
    \nnf(\alpha \,\beta) &\coloneqq  \nnf(\alpha) \, \nnf(\beta) &
    \nnf([\varphi]) &\coloneqq  [\nnf(\varphi)]\\
    \nnf( \alpha^* ) &\coloneqq (\nnf(\alpha))^* &
    \nnf( o ) &\coloneqq o \quad o \in \mathfrak F\\
    \nnf(\tup{\alpha \circledast \beta}) &\coloneqq  \tup{\nnf(\alpha) \circledast \nnf(\beta)} &
    \nnf(\lnot \tup{\alpha \circledast \beta}) &\coloneqq  \lnot \tup{\nnf(\alpha) \circledast \nnf(\beta)}\\
    \nnf(\texttt{a}) &\coloneqq  \texttt{a} &
    \nnf(\lnot \texttt{a}) &\coloneqq  \lnot\texttt{a} \\
    \nnf(\lnot\lnot \varphi) &\coloneqq  \nnf(\varphi) &
    \nnf(\tup{\alpha}) &\coloneqq \tup{\nnf(\alpha)}
  \end{align*}
  \caption{Definition of the Negation Normal Form for \xpath.}
\label{tab:nnf-definition}
\end{table}
The idea is that a state $\llpar \varphi \rrpar$ verifies  the formula $\varphi$. 
A state $\llpar\alpha\rrpar^=_{i}$ (resp.\ $\llpar\alpha\rrpar^{\neq}_{i}$) verifies that there is a forward path  in the tree ending at a node with the same (resp.\ different) data value as the one in the register, such that there exists a partial run  of $\cl{H}_\alpha$ over such path that starts in a moving state $i$ and ends in a final state. 
Similarly, a state $\llpar\alpha\rrpar^=_{\text{test},i}$ or $\llpar\alpha\rrpar^{\neq}_{\text{test},i}$ verifies the same when $i$ is a testing state.
A state $\llpar \alpha, \beta \rrpar^=_{i,j}$ (resp.\  $\llpar \alpha, \beta \rrpar^{\neq}_{i,j}$) verifies that there are \emph{two} paths ending in two nodes with the same (resp.\ one equal, the other different) data value as that of the register; such that one path has a partial accepting run of $\cl{H}_\alpha$ starting in a moving state $i$, and the other has a partial accepting run of $\cl{H}_\beta$ starting in a moving state $j$.
A state like $\llpar \alpha  \rrpar^{=}_{\msf{F}}$ is simply to mark that the run of $\cl{H}_\alpha$ on a path has ended, and the only remaining task is to test for equality of the data value with respect to the register. Finally, a state $\llpar \alpha \rrpar^{\lnot =}_{0}$ (resp. $\llpar \alpha \rrpar^{\lnot \neq}_{0}$) verifies that every node reachable by $\alpha$ has different (resp. equal) data than the register, and similar for the other universal states of the form $\llpar \cdots \rrpar^{\lnot  \cdots}_{\cdots}$.
We first take care of the boolean connectors and the simplest tests.
\begin{align*}
  \delta(\llpar\texttt{a} \rrpar) &\coloneqq  \texttt{a} &
  \delta(\llpar\varphi \lor \psi\rrpar) &\coloneqq  \llpar\varphi\rrpar \lor \llpar\psi\rrpar &
  \delta(\llpar \lnot\texttt{a} \rrpar) &\coloneqq  \bar{\texttt{a}}& 
  \delta(\llpar\varphi \land \psi \rrpar) &\coloneqq  \llpar\varphi\rrpar
  \land \llpar\psi\rrpar 
\end{align*}
First, we define the transitions associated to each $\cl{H}_\alpha$, for $i \in Q_\alpha, \+{C} \subseteq Q_\alpha, \circledast \in \set{=,\neq}$. Here, $\llpar \alpha \rrpar_{\msf{F}}$ holds at the endpoint of a path matching $\alpha$.
\begin{align*}
\delta(\llpar\alpha\rrpar^{\circledast}_{i})&\coloneqq \tridown
  \llpar\alpha\rrpar^{\circledast}_{\text{test}, \delta_\alpha(\downarrow,i)}
  \quad\lor\quad
 \triright
  \llpar\alpha\rrpar^{\circledast}_{\text{test}, \delta_\alpha(\rightarrow,i)}
  \quad\lor\quad
  \bigvee_{i \in
    F_\alpha}\hspace{-2pt}\llpar\alpha\rrpar^{\circledast}_{\msf{F}}
\\
  \delta(\llpar\alpha\rrpar^{\circledast}_{\text{test},i}) &\coloneqq 
  \bigvee_{\substack{S \subseteq \nsubf(\alpha)}}
  \big(
  \llpar\alpha\rrpar^{\circledast}_{\substack{ \delta_\alpha(S,i)}}
  \land \bigwedge_{\varphi \in S} \llpar\varphi\rrpar 
  \big)
\end{align*}

Next, we focus only on the data-aware test formul{\ae}, since the tests $\tup{\alpha}$ and $\lnot \tup{\alpha}$ are interpreted as the equivalent formul{\ae} $\tup{\alpha = \alpha}$ and $\lnot \tup{\alpha = \alpha}$ respectively.
Using the \opguess operator, we can easily define the cases corresponding to the data test cases {\itemTeq} and {\itemTeqn} as follows. 
\begin{align*}
\delta(\llpar\alpha = \beta\rrpar) &\coloneqq  \opguess(\llpar\alpha,\beta\rrpar^=) &
\delta(\llpar\alpha,\beta\rrpar^=) &\coloneqq  \llpar\alpha\rrpar^=_{0} \land \llpar\beta\rrpar^=_{0}&\delta(\llpar\alpha\rrpar^{=}_{\msf{F}}) &\coloneqq  \opeq\\
\delta(\llpar\alpha \not= \beta\rrpar) &\coloneqq  \opguess(\llpar\alpha,\beta\rrpar^{\not=}) &
\delta(\llpar\alpha,\beta\rrpar^{\not=}) &\coloneqq  \llpar\alpha\rrpar^{=}_{0} \land \llpar\beta\rrpar^{\not=}_{0} & 
\delta(\llpar\alpha\rrpar^{\not=}_{\msf{F}}) &\coloneqq  \opneq
\end{align*}

The test case {\itemTneqn} involves also an \emph{existential} quantification over data values. In fact, $\lnot \tup{\alpha \not = \beta}$ means that either
\begin{enumerate}[(i)]
\item \label{noreachbyalpha} there are no nodes reachable by $\alpha$,
\item \label{noreachbybeta} there are  no nodes reachable by $\beta$, or
\item \label{dvstbla} there \emph{exists} a data
  value $d$ such that both
  \begin{enumerate}[(a)]
  \item all elements reachable by $\alpha$ have datum $d$, and
  \item all elements reachable by $\beta$ have datum $d$.
  \end{enumerate}

\end{enumerate}

\begin{align*}
  \delta(\llpar\lnot \alpha \not= \beta\rrpar) &\coloneqq  \llpar \lnot \tup{\alpha}\rrpar \lor \llpar \lnot \tup{\beta} \rrpar \lor \opguess(\llpar\alpha,\beta\rrpar^{\lnot \not=}) \\ 
\delta(\llpar\alpha,\beta\rrpar^{\lnot \not=}) &\coloneqq  \llpar\alpha\rrpar^{\lnot \not=}_{0} \land \llpar\beta\rrpar^{\lnot\not=}_{0} \qquad
\delta(\llpar\alpha\rrpar^{\lnot \not=}_{\msf{F}}) \coloneqq  \opeq \qquad \delta(\llpar\alpha\rrpar^{\lnot =}_{\msf{F}}) \coloneqq  \opneq 
\end{align*}
\begin{align*}
  \delta(\llpar\alpha\rrpar^{\lnot \circledast}_{i}) &\coloneqq 
(\bar\tridown ? \lor
  \tridown \llpar\alpha\rrpar^{\lnot
    \circledast}_{\text{test},\delta_\alpha(\downarrow,i)}) 
\land (\bar\triright ? \lor
  \triright \llpar\alpha\rrpar^{\lnot
    \circledast}_{\text{test},\delta_\alpha(\rightarrow,i)}) 
\quad \land \quad
 \textit{final}
\\&\qquad 
\text{where } \bar\varphi\text{ stands for }\msf{nnf}(\lnot \varphi)
\text{ and } \textit{final} =
\begin{cases}
  \llpar \alpha \rrpar^{\lnot \circledast}_{\msf{F}} & \text{if $i \in F_\alpha$,}\\
  \textit{true} & \text{otherwise.}
\end{cases}
\\
\delta(\llpar\alpha\rrpar^{\lnot \circledast}_{\text{test},i})&\coloneqq
 \bigwedge_{\substack{S \subseteq \nsubf(\alpha)}}
  \big( 
\llpar\alpha\rrpar^{\lnot    \circledast}_{\delta_\alpha(S,i)} 
    \land \bigwedge_{\varphi \in S} \llpar\varphi\rrpar 
\big) 
\end{align*}
\smallskip

The difficult part is the translation of the data test case {\itemTneq}.  The main reason for this difficulty is the fact that
$\atra(\opguess,\opspread)$ automata do not have the expressivity to make these kinds of tests. An expression $\lnot \tup{\alpha = \beta}$ forces the set of data values reachable by
an $\alpha$-path  and the set of those reachable by a $\beta$-path to be disjoint. We show that nonetheless the automaton can test for a
condition that is equivalent to $\lnot \tup{\alpha = \beta}$ if we assume that the run and tree have the disjoint values property. 

\begin{example}
  As an example, suppose  that $\eta = \lnot \tup{\, {{\downarrow}
    \alpha} = {{\rightarrow} \beta} \,}$ is to be checked for
  satisfiability. One obvious answer would be to test separately
  $\alpha$ and $\beta$. If both tests succeed, one can then build a
  model satisfying $\eta$ out of the two witnessing trees by making
  sure they have disjoint sets of values. Otherwise, $\eta$ is clearly
  unsatisfiable. Suppose now that we have $\eta = \varphi \land \lnot \tup{\, {{\downarrow}  \alpha} = {{\rightarrow} \beta} \,}$, where $\varphi$ is any
  formula with no data tests of type {\itemTneq}. One could build the
  automaton for $\varphi$ and then ask for ``$\opspread(\,\llpar
  {\down}\alpha \rrpar^{\lnot=}_0 \lor \llpar {\rightarrow}\beta
  \rrpar^{\lnot=}_0\,)$'' in the automaton. This corresponds to the
  property ``\textsl{for every data value $d$ taken into account by
    the automaton (as a result of the translation of $\varphi$),
    either all elements reachable by $\alpha$ do not have datum $d$,
    or all elements reachable by $\beta$ do not have datum $d$}''. If
  $\varphi$ contains a $\tup{\alpha' =\beta'}$ formula, this
  translates to a \emph{guessing} of a witnessing data value
  $d$. Then, the use of $\opspread$ takes care of this particular data
  value, and indeed of all other data values that were guessed to
  satisfy similar demands. In other words, it is \emph{not} because of
  $d$ that $ \lnot \tup{\, {\downarrow} \alpha = {\rightarrow} \beta\,}$ will
  be falsified. But then, the disjoint values property ensures that no
  pair of nodes accessible by $\alpha$ and $\beta$ share the same
  datum. This is the main idea we encode next.
\end{example}
We define $\delta(\llpar\lnot \tup{\alpha = \beta} \rrpar) \coloneqq   \llpar \alpha, \beta \rrpar^{\lnot =}_{0,0} $. Given $\lnot \tup{\alpha = \beta}$, the automaton systematically looks for the closest common ancestor of every pair $(x,y)$ of nodes accessible by $\alpha$ and $\beta$ respectively, and tests, for every data value $d$ present in the node configuration, that either (1) all data values accessible by the remaining path of $\alpha$ are different from $d$, or (2) all data values accessible by the remaining path of $\beta$ are different from $d$.
\begin{align*}
  \delta(\llpar \alpha, \beta \rrpar^{\lnot =}_{i,j}) &\coloneqq   
\opspread\big(\llpar \alpha \rrpar^{\lnot =}_{i} \lor \llpar \beta \rrpar^{\lnot =}_{j}\big)
\\&
\fudgece\land \quad
(\bar\tridown ? \lor \tridown \llpar \alpha, \beta \rrpar^{\lnot =}_{\text{test},\delta_\alpha(\downarrow,i),\delta_\beta(\downarrow,j)})
\quad
\land
\quad
(\bar\triright ? \lor \triright \llpar \alpha, \beta \rrpar^{\lnot =}_{\text{test},\delta_\alpha(\rightarrow,i), \delta_\beta(\rightarrow,j)} )
\\&
\fudgece\land \quad
\textit{final}_\alpha
\quad\land\quad
\textit{final}_\beta
\\&
\fudgece\text{where } \textit{final}_\alpha =
\begin{cases}
  \opset( \llpar \beta \rrpar^{\lnot =}_{j} ) &\text{if $i \in F_\alpha$,}\\
  \textit{true} &\text{otherwise,}\\
\end{cases}\\&
\fudgece\text{and }
\textit{final}_\beta =
\begin{cases}
  \opset( \llpar \alpha \rrpar^{\lnot =}_{i} ) &\text{if $j \in F_\beta$,}\\
  \textit{true} &\text{otherwise.}\\
\end{cases}
\\
\delta(\llpar \alpha, \beta \rrpar^{\lnot =}_{\text{test},i,j})&\coloneqq
\bigwedge_{\substack{S \subseteq \nsubf(\alpha)}}
  \big( \llpar\alpha,\beta\rrpar^{\lnot=}_{\delta_\alpha(S,i),\delta_\beta(S,j)}
    \land \bigwedge_{\varphi \in S} \llpar\varphi\rrpar 
\big)
\end{align*}

The first line of $\delta(\llpar \alpha, \beta \rrpar^{\lnot =}_{i,j})$ corresponds to the tests (1) and (2) above, and the third line corresponds to the cases where $x \ancestor_{fcns} y$ and $y \ancestor_{fcns} x$.

Next we show the correctness of this translation. We say that $\+A$ has an accepting run from a thread $(q,d)$ on $\tT|_x$ if there is a sequence of tree configurations $\+S_1 \atraT \dotsb \atraT \+S_n$ such that $\+S_1$ is at position $x$ and contains $(q,d)$, and $\+S_n$ is accepting. We  say that there is a \emph{dvp-run} if $\+S_1 \atraT \dotsb \atraT \+S_n$ has the disjoint values property.
\begin{lemma}\label{lem:xp->atra}
  For any data tree $\tT$,
  \begin{enumerate}[\em($\Rightarrow$)]
  \item[\em($\Rightarrow$)] if $\tT\models \eta$ then \anAut accepts $\tT$ with a dvp-run, and
  \item[\em($\Leftarrow$)] if \anAut accepts $\tT$ with a dvp-run, then $\tT\models \eta$.
  \end{enumerate}
\end{lemma}
\begin{proof}
Let $\tT = \aA \otimes \dD$. We show that for every subformula $\varphi \in \subf(\eta)$ and position $x \in \tpos(\tT)$,
\begin{align}
  \text{if $\tT|_x \models \varphi$ then $\+A$ has an accepting dvp-run from
  $(\llpar \varphi \rrpar, d)$ on $\tT|_x$ for any $d \in \D$.}\tag{$\dag$}\label{indhyp:1}
\end{align}
And conversely,
\begin{align}
  \text{if $\+A$ has an accepting dvp-run from $(\llpar \varphi \rrpar,
  d)$ on $\tT|_x$ for some $d \in \D$, then $\tT|_x \models \varphi$.}\tag{$\ddag$}\label{indhyp:2}
\end{align}

We suppose that $\eta$ is in \nnf. We proceed by induction on the lexicographic order of $(f(x),|\varphi|)$, where $f(x)$ is the length of a maximal path from $x$ to a leaf in the fcns coding (roughly, the height of $\tT|_x$). The base cases when $\varphi = a$ or $\varphi = \lnot a$ are easy. Suppose then that the proposition is true for all position $y \descendantstr_{fcns} x$, and for all $\psi \in \subf(\eta)$, $|\psi| < |\varphi|$ at position $x$. The cases $\varphi = \varphi_1 \land \varphi_2$ or $\varphi = \varphi_1 \lor \varphi_2$ are straightforward from the inductive hypothesis. The following claim also follows from the inductive hypothesis.
\begin{claim} \label{cl:paths-correct}
For any path subformula $\alpha$ of $\psubf(\eta)$, and $y \descendant_{fcns} x'$,
\begin{enumerate}[\em(a)]
\item \label{it:paths-correct:1}
if  $(x,y) \in \dbracket{\alpha}^\tT$ then there is an accepting dvp-run from $( \llpar \alpha \rrpar^=_{ 0}, \dD(y))$  on $\tT|_x$,
\item \label{it:paths-correct:1b}
if there is an accepting dvp-run from $( \llpar \alpha \rrpar^=_{0}, d)$  on $\tT|_x$, then $(x,y) \in \dbracket{\alpha}^\tT$ for some $y$ such that $\dD(y) = d$,
\item \label{it:paths-correct:2}
if $(x,y) \in \dbracket{\alpha}^\tT$, there is an accepting dvp-run from $( \llpar \alpha \rrpar^{\neq}_{0}, d )$ with $d \neq \dD(y)$  on $\tT|_x$,
\item \label{it:paths-correct:2b}
if there is an accepting dvp-run from $( \llpar \alpha \rrpar^{\neq}_{0}, d )$  on $\tT|_x$, then $(x,y) \in \dbracket{\alpha}^\tT$ for some $y$ with $d \neq \dD(y)$.
\end{enumerate}
\end{claim}
\begin{proof}
By induction, suppose that Claim~\ref{cl:paths-correct} holds for all $\alpha ,x', y$ such that $(f(x'),|\alpha|) <_{lex} (f(x),|\varphi|)$. As a consequence of the inductive hypothesis on \eqref{indhyp:1}, for all positions $z$ with $x \ancestorstr_{fcns} z \ancestor_{fcns} y$, and for every subformula $\psi \in \nsubf(\eta)$, if $z \in \dbracket{\psi}^\tT$ then there is an accepting dvp-run of $\cl{A}$ from $(\llpar \psi \rrpar,d)$ on $\tT|_z$ for any $d\in\D$. 

\noindent\eqref{it:paths-correct:1} Suppose first that $(x,y) \in \dbracket{\alpha}^\tT$. Therefore, $\cl{H}_\alpha$ accepts $str(x,y)$ by Claim~\ref{cl:paths}. Further, suppose that $S \in \A_\eta$ is the  $(2 i)$-th element of $str(x,y)$, and that $z$ is the position of the $i$-th element in the path between $x$ and $y$ of the fcns coding.  For any $\psi \in S$ we have that  $z \in \dbracket{\psi}^\tT$, and applying the inductive hypothesis on \eqref{indhyp:1} we have that there is an accepting dvp-run from $(\llpar \psi \rrpar,d)$ on $\tT|_z$. Hence, by definition of the transition of $\llpar \alpha \rrpar^{=}_{0}$, we have that there is an accepting dvp-run from $(\llpar \alpha \rrpar^{=}_{0},d')$ on $\tT|_x$ if we take $d' = \dD(y)$. 

\noindent\eqref{it:paths-correct:1b} For any $z$ with $x \ancestorstr_{fcns} z \ancestor_{fcns} y$, if there is an accepting dvp-run from $(\llpar \psi \rrpar, d)$ on $\tT|_z$, by inductive hypothesis on \eqref{indhyp:2}, we have that $z \in \dbracket{\psi}^{\tT}$. Therefore, if there is an accepting dvp-run from $(\llpar \alpha \rrpar^=,d)$  on $\tT|_x$, by the definition of the transition relation there must be some $y \descendant_{fcns} x$ such that $str(x,y)$ is accepted by $\+H_\alpha$ and further $\dD(y)=d$. Hence, $(x,y) \in \dbracket{\alpha}^\tT$.

An identical reasoning applies to show \eqref{it:paths-correct:2} and \eqref{it:paths-correct:2b}.
%
\end{proof}
%
%

\noindent($\Rightarrow$) Suppose that $\tT|_x \models \varphi$.  We focus on the data test cases. If $\varphi = \tup{\alpha = \beta}$, then $\delta(\llpar \varphi \rrpar) = \opguess(\llpar \alpha, \beta \rrpar^=)$. Since $\delta(\llpar\alpha,\beta\rrpar^=) = \llpar\alpha\rrpar^=_{0} \land \llpar\beta\rrpar^=_{0}$, we then  have that for any data value $d$, there is an accepting run from $(\llpar \varphi \rrpar,d)$ on $\tT|_x$ if there exists a data value $d' \in \D$ such that $(\llpar\alpha\rrpar^=_{0},d')$ and $(\llpar\beta\rrpar^=_{0},d')$ have accepting dvp-runs. 
Since $\tT|_x \models \tup{\alpha = \beta}$, there must be two positions $y$ and $y'$ below $x$ such that $\dD(y)=\dD(y')$ where $(x,y) \in \dbracket{\alpha}^\tT$ and $(x,y') \in \dbracket{\beta}^\tT$. By Claim~\ref{cl:paths-correct}.\ref{it:paths-correct:1} there are accepting  dvp-runs from $(\llpar\alpha\rrpar^=_{0},d')$ and $(\llpar\beta\rrpar^=_{0},d')$ on $\tT|_x$. These two dvp-runs can be simply combined to build an accepting dvp-run from $(\llpar\varphi\rrpar,d)$. The case $\varphi = \tup{\alpha \neq \beta}$ is similar, using  Claims~\ref{cl:paths-correct}.\ref{it:paths-correct:1} and \ref{cl:paths-correct}.\ref{it:paths-correct:2}.

\smallskip

If $\varphi = \lnot \tup{\alpha = \beta}$, then $\tT|_x$ is such that for every pair of paths ending in some positions $y,y'$, if the paths satisfy $\alpha$ and $\beta$ respectively, then $\dD(y) \neq \dD(y')$. Since $\delta(\lnot \tup{\alpha = \beta}) = \llpar \alpha, \beta\rrpar^{\lnot =}$, we show that there is an accepting run from $(\llpar \alpha, \beta\rrpar^{\lnot =},d)$, for any $d \in \D$. 
Using the inductive hypothesis on all the nodes below $x$ (as we did in Claim~\ref{cl:paths-correct}), note that every time that
the automaton has an accepting  dvp-run  from $(\llpar \psi \rrpar, d)$ at $x$ for some $d \in \D$, then $x \in \dbracket{\psi}^\tT$. By definition of the transition relation on $\llpar \alpha, \beta\rrpar^{\lnot =}$, this means that a sufficient condition for the automaton to have an accepting dvp-run from $(\llpar \alpha, \beta\rrpar^{\lnot =},d)$ is: for every node $z \descendant_{fcns} x$, if the states of $\+H_\alpha, \+H_\beta$ are respectively $i,j$ after reading $str(x,z)$, then there is an accepting dvp-run from $(\opspread\big(\llpar \alpha \rrpar^{\lnot =}_{i} \lor \llpar \beta \rrpar^{\lnot =}_{j}\big),d)$ on $\tT|_z$. To simplify the argument, we can assume that the configuration contains \emph{all} the data values of the tree and therefore that $opspread$ takes into account all possible data values (this is not a necessary condition, but certainly a sufficient one to have an accepting dvp-run).
Also using the inductive hypothesis, there is an accepting dvp-run from $(\llpar \alpha \rrpar^{\lnot =}_{i},d')$ on $\tT|_z$ if for every node $z' \descendant_{fcns} z$ such that $\+H_\alpha$ takes the state $i$ to a final state after reading $str(z,z')$, we have $\dD(z') \neq d'$. Therefore, there is an accepting dvp-run from $\opspread\big(\llpar \alpha \rrpar^{\lnot =}_{i} \lor \llpar \beta \rrpar^{\lnot =}_{j}\big)$ if for every data value $d'$ and nodes $y,y'$ that complete the paths, it cannot be that $d'=\dD(y)=\dD(y')$. This is true, since $\tT|_x \models \lnot \tup{\alpha = \beta}$.  On the other hand, note that by the same reason, for every node $z'$ such that $(z,z') \in \dbracket{\beta}^{\tT}$,  there are accepting dvp-runs for $(\llpar \alpha \rrpar^{\lnot =}_{ i},\dD(z'))$ and for $(\opset( \llpar \alpha \rrpar^{\lnot =}_{ i} ),d')$ from $z'$ (and respectively swapping $\alpha$ and $\beta$).
Thus, by combining all the accepting dvp-runs from $(\opspread\big(\llpar \alpha \rrpar^{\lnot =}_{i} \lor \llpar \beta \rrpar^{\lnot =}_{j}\big),d)$ for every such $z$, we obtain an accepting dvp-run from $(\llpar \varphi \rrpar,d)$ on $\tT|_x$.
\todo{agregar una remark que diga que combinar dvp-runs resulta en un dvp-run claramente?}

\smallskip

Finally, if $\varphi = \lnot \tup{\alpha \neq \beta}$, then   $\delta(\llpar\varphi\rrpar) = \llpar \lnot \tup{\alpha}\rrpar \lor \llpar \lnot \tup{\beta} \rrpar \lor \opguess(\llpar\alpha,\beta\rrpar^{\lnot \not=}) $. Since $x \in \dbracket{\varphi}^\tT$, we have that either \eqref{noreachbyalpha}, \eqref{noreachbybeta}, or \eqref{dvstbla} (on page \pageref{dvstbla}) holds for $\tT|_x$. 
The first two conditions correspond to the properties for which $\lnot \tup{\alpha}$ and $\lnot \tup{\beta}$ are satisfied, which are equivalent to $\lnot \tup{\alpha = \alpha}$ and $\lnot \tup{\beta = \beta}$. By the case already shown, this means that $\llpar \lnot \tup{\alpha = \alpha} \rrpar$ or $\llpar \lnot \tup{\beta = \beta} \rrpar$ have accepting dvp-runs on $\tT|_x$ respectively. Since the definition of $\llpar \lnot \tup{\alpha = \alpha} \rrpar$ and $\llpar \lnot \tup{\alpha} \rrpar$ (idem with $\beta$) are treated identically by the automaton, $(\llpar \varphi \rrpar,d)$ has an accepting dvp-run on $\tT|_x$.

We now show that the third condition corresponds to having an accepting dvp-run from  $(\opguess(\llpar\alpha,\beta\rrpar^{\lnot \not\neq}),d)$  on $\tT|_x$ for any data value $d$. 
By definition of $\delta$, there is an accepting dvp-run if there is a data value $d'$ such that there are accepting dvp-runs from $(\llpar \alpha \rrpar^{\lnot \neq }_{0},d')$ and $(\llpar \beta \rrpar^{\lnot \neq }_{0},d')$. By applying the inductive hypothesis on all positions below $x$ and subformulas of $\nsubf(\eta)$, we have the following, for $\+D^\alpha_x = \set{\dD(y) \mid (x,y) \in \dbracket{\alpha}^\tT}$, $\+D^\beta_x = \set{\dD(y) \mid (x,y) \in \dbracket{\beta}^\tT}$.
\begin{claim} For any $e \not\in \+D^\alpha_x$, there is an accepting dvp-run from $(\llpar \alpha \rrpar^{\lnot =},e)$  on $\tT|_x$. If $\+D^\alpha_x= \set{e}$, there is an accepting dvp-run from $(\llpar \alpha \rrpar^{\lnot \neq},e)$  on $\tT|_x$.  Idem for $\beta$ and $\+D^\beta_x$.
\end{claim}
Note that since $\tT|_x \models \lnot \tup{\alpha \neq \beta}$, condition \eqref{noreachbyalpha} corresponds  to $\+D^\alpha_x = \emptyset$, condition  \eqref{noreachbybeta} to $\+D^\beta_x = \emptyset$, and condition \eqref{dvstbla} to $\+D^\alpha_x = \+D^\beta_x = \set e$. Therefore, if condition \eqref{dvstbla} holds, there are accepting  dvp-runs from $(\llpar \alpha \rrpar^{\lnot \neq},e)$ and $(\llpar \beta \rrpar^{\lnot \neq},e)$, and hence also from $(\opguess(\llpar\alpha,\beta\rrpar^{\lnot \not\neq}),d)$ for any $d$, by guessing the data value $e$. Thus, there is an accepting dvp-run from $(\llpar \varphi \rrpar,d)$ if conditions \eqref{noreachbyalpha}, \eqref{noreachbybeta} or \eqref{dvstbla} hold. Therefore, there is an accepting dvp-run from $(\llpar \varphi \rrpar,d)$ on $\tT|_x$.

\medskip

\noindent($\Leftarrow$) Consider an accepting dvp-run $\+S_1 \atraT \dotsb \atraT \+S_n$ from $(\llpar \varphi \rrpar,d)$  on $\tT|_x$. We show that $\tT|_x \models \varphi$.

If $\varphi = \tup{\alpha = \beta}$, by definition of $\delta$ the automaton guesses a data value $d$ and creates a thread $(\llpar  \alpha, \beta  \rrpar^=, d)$, which means that there is an accepting dvp-run from $(\llpar  \alpha, \beta  \rrpar^=, d)$. Then, there are accepting dvp-runs from $(\llpar  \alpha \rrpar^=, d)$ and $(\llpar  \beta \rrpar^=, d)$. By Claim~\ref{cl:paths-correct}.\ref{it:paths-correct:1}, there must be two nodes $y,z$ such that $(x,y) \in \dbracket{\alpha}$ and $(x,z) \in \dbracket{\beta}$ with $d=\dD(z) = \dD(y)$. Hence, $\tT|_x \models \varphi$.  The case $\varphi = \tup{\alpha \neq \beta}$ is treated in a similar way.

If $\varphi = \lnot \tup{\alpha = \beta}$, suppose ad absurdum that there is a data value $d$ and two positions $y,z$ such that $(x,y) \in \dbracket{\alpha}^\tT$ and $(x,z) \in \dbracket{\beta}^\tT$, and $\dD(y)=\dD(z)=d$. Let $z'$ be the closest common ancestor of $y,z$ in the fcns coding. 
Since there is an accepting dvp-run from $\llpar \alpha, \beta \rrpar^{\lnot =}$, there must be some states $i', j'$ assumed by $\+H_\alpha$ and $\+H_\beta$ after reading $str(x,z')$, such that the following holds.
\begin{claim}\label{cl:test-disjoint}
There is a moving configuration $\+S_r$ containing $\llpar \alpha, \beta \rrpar ^{\lnot = }_{i',j'}$ with position $z'$ such that
  \begin{iteMize}{$\bullet$}
  \item if $z'=y$, then $( \llpar \alpha \rrpar^{\lnot =}_{i'} ,\dD(y) )$ is in some $\+S_s$, $s>r$, with position $z'$,
  \item if $z'=z$, then $( \llpar \beta \rrpar^{\lnot =}_{j'} ,\dD(z) )$ is in some $\+S_s$, $s>r$, with position $z'$,
  \item for every data value $d$ of $\+S_r$, $(\llpar \alpha \rrpar^{\lnot =},d)_{i'}$ or $(\llpar \beta \rrpar^{\lnot =},d)_{j'}$ are in some $\+S_s$, $s>r$, with position $z'$.
  \end{iteMize}
\end{claim}

\noindent Hence, by Claim~\ref{cl:test-disjoint} there must be a configuration $\+S_r$ with position $z'$, where $z'$ is the common ancestor of $y$ and $z$ in the fcns coding, such that for all data values $d'$ of $\+S_r$ it is not true that $\dD(x)=\dD(y)=d'$. But by the fact that this is a dvp-run, this means that there cannot be \emph{any} data value $d''$ such that $\dD(y)=\dD(z)=d''$. In particular, $d$, and thus we have a contradiction. 
Then, $\tT|_x \models \varphi$.
The cases of $y \ancestor_{fcns} z$ or $z \ancestor_{fcns} y$ are only easier.

Finally, if $\varphi = \lnot\tup{\alpha \neq \beta}$, then $\cl{A}$ must have an accepting dvp-run on $\tT|_x$ from $(\llpar \lnot \tup{\alpha} \rrpar, d)$, $(\llpar \lnot \tup{\beta} \rrpar,d)$, or $(\llpar \alpha, \beta \rrpar^{\lnot \neq}_{0,0},d)$ for some $d \in \D$. In the first two cases this means that $\dbracket{\tup\alpha}^{\tT|_x} = \emptyset$ or $\dbracket{\tup\beta}^{\tT|_x} = \emptyset$ (by reduction to the case $\lnot \tup{\alpha = \alpha}$ already treated) and hence that $\tT \models \varphi$. In the latter case, by definition of $\delta$, all positions $y$ whose paths from $x$ satisfy $\alpha$ or $\beta$ are such that $\dD(y)=d$. This implies that $\tT \not\models \tup{\alpha \neq \beta}$, and hence $\tT \models \varphi$.
\end{proof}

Lemma \ref{lem:xp->atra} together with Proposition~\ref{prop:normal-form-disjoint-values-properties} concludes the proof of Proposition~\ref{prop:rxpathRedAtra}. 
We then have that Theorem~\ref{thm:fxpath-dec} holds.
\begin{proof}[Proof of Theorem~\ref{thm:fxpath-dec}]
By Proposition~\ref{prop:rxpathRedAtra}, satisfiability of $\rxpath(\mathfrak F,=)$ is reducible to the nonemptiness problem for $\atra(\opguess,\opspread)$. On the other hand, we remark that $\atra(\opguess,\opspread)$ automata can encode any regular tree language ---in particular a DTD, the core of XML Schema, or Relax NG--- and are closed under intersection by Proposition~\ref{prop:atra-closed}. Also, the logic can express any unary key constraint as stated in Lemma~\ref{lem:primary-key}. Hence, by Theorem~\ref{thm:atrags-decidable} the decidability follows.
\end{proof}

\section*{Extensions}

We consider some operators that can be added to $\rxpath(\mathfrak F, =)$ preserving the decidability of the satisfiability problem. For each of these operators, we will see that they can be coded as a $\atra(\opguess,\opspread)$ automaton, following the same lines of the translation in Section~\ref{sec:allow-arbitr-data}.

\subsection{Allowing upward axes}
\label{sec:allowing-upward-axes}
\newcommand{\rxpFB}{\mbox{$\rxpath^{\!\mathfrak B}(\mathfrak F,=)$}\xspace}
Here we explore one possible decidable extension to the logic $\rxpath(\mathfrak F,=)$, whose decidability can be reduced to that of $\atra(\opguess,\opspread)$. In this extension we can make use of \emph{upward} ($\uparrow$) and \emph{leftward} ($\leftarrow$) axes to navigate the tree in a restricted way.  We can test, for example, that all the ancestors of a given node labeled with $a$ have the same data value as all the descendants labeled with $b$, with the formula $\lnot\tup{{{\uparrow^*}[a]} \neq {\td [b]}}$, but we cannot test its negation.

Let \rxpFB be the fragment of $\rxpath(\mathfrak F \cup \mathfrak B,=)$ where $\mathfrak B \coloneqq  \set{\uparrow,\uparrow^*, \leftarrow,\tl}$ defined by the grammar
\begin{multline*}
  \varphi, \psi \Coloneqq \lnot
  \texttt{a} \,\mid\, \texttt{a} \,\mid\,  \varphi \land \psi \,\mid\, \varphi
  \lor \psi \,\mid\, \tup{\alpha_{\mathfrak f}} \,\mid\, \tup{\alpha_{\mathfrak b}} \,\mid\,
\\
\tup{\alpha_{\mathfrak f}\circledast\beta_{\mathfrak f}} \,\mid\, \lnot \tup{\alpha_{\mathfrak f} \circledast \beta_{\mathfrak f}}
  \,\mid\, \lnot \tup{\alpha_{\mathfrak b} = \beta_{\mathfrak f}} \,\mid\, \lnot
  \tup{\alpha_{\mathfrak b} \not= \beta_{\mathfrak f}}
\end{multline*}
with $\circledast \in \set{=,\neq}, \texttt{a} \in \A$, and
\begin{align*}
  \alpha_{\mathfrak f},\beta_{\mathfrak f} &\Coloneqq [\varphi] \,\mid\, \alpha_{\mathfrak f} \beta_{\mathfrak f} \,\mid\, \alpha_{\mathfrak f} \cup \beta_{\mathfrak f} \,\mid\, o\, \alpha_{\mathfrak f} \,\mid\, (\alpha_{\mathfrak f})^* && o \in \set{\down,\rightarrow,\eps},\\
  \alpha_{\mathfrak b},\beta_{\mathfrak b} &\Coloneqq [\varphi] \,\mid\, \alpha_{\mathfrak b} \beta_{\mathfrak b} \,\mid\, \alpha_{\mathfrak b} \cup \beta_{\mathfrak b} \,\mid\, o\, \alpha_{\mathfrak b} \,\mid\, (\alpha_{\mathfrak b})^*  && o \in \set{\uparrow,\leftarrow,\eps}.
\end{align*}
We must note that \rxpFB contains $\rxpath(\mathfrak F, \mathfrak B)$, that is, the full data-unaware fragment of \xpath. We also remark that it is \emph{not} closed under negation. Indeed, we cannot express the negation of ``there exists an $a$ such that all its ancestors labeled $b$ have different data value'' which is expressed by ${\downarrow^*}[a \land \lnot \tup{{\uparrow^*}[b]=\varepsilon}]$.
As shown in Proposition~\ref{prop:guess-spread-more-expressive}, if the negation of this property were expressible, then its satisfiability would be undecidable. It is not hard to see that we can decide the satisfiability problem for this fragment.

Consider  the data test expressions of the types
\[
  \lnot \tup{\alpha_{\mathfrak b} = \beta_{\mathfrak f}} \qquad \text{and} \qquad
  \lnot \tup{\alpha_{\mathfrak b} \not= \beta_{\mathfrak f}}
\]
where $\beta_{\mathfrak f} \in \rxpath(\mathfrak F,=)$ and $\alpha_{\mathfrak b} \in \rxpath(\mathfrak B)$. We can decide the satisfaction of these kinds of expressions by means of  $\opspread(~,~)$, using carefully its first parameter to select the desired threads from which to collect the data values we are interested in. Intuitively, along the dvp-run we throw threads that save the current data value and try out all possible ways to verify $\alpha_{\mathfrak b}^r \in \rxpath(\mathfrak F,=)$, where $(~)^r$ stands for the \emph{reverse} of the regular expression. Let the automaton arrive with a thread $(\llpar \alpha_{\mathfrak b} \rrpar, d)$ whenever  $\alpha_{\mathfrak b}^r$ is verified. This signals that there is a backwards path from the current node in the relation $\alpha_{\mathfrak b}$ that arrives at a node with data value $d$. Hence, at any given position, the instruction $\opspread(\llpar \alpha_{\mathfrak b} \rrpar, \llpar \alpha_{\mathfrak f}\rrpar^{\lnot \circledast})$ translates correctly the expression $\lnot \tup{\alpha_{\mathfrak b} \circledast \beta_{\mathfrak f}}$. 
Furthermore, $\alpha_{\mathfrak b}$ need not be necessarily in $\rxpath(\mathfrak B)$, as its intermediate node tests can be formul{\ae} from $\rxpath(\mathfrak F,=)$. We then obtain the following result.

\begin{remark}
  Satisfiability for \rxpFB under key constraints and DTDs is decidable.
\end{remark}

\subsection{Allowing stronger data tests}
\label{sec:allow-strong-data}

\newcommand{\dcbracket}[1]{\{\!\!\{#1\}\!\!\}}
Consider the property ``\textsl{there are three descendant nodes labeled $a$, $b$ and $c$ with the same data value}''. That is, there exists some data value $d$ such that there are three nodes accessible by $\td[a]$, $\td[b]$ and $\td[c]$ respectively, all carrying the datum $d$. Let us denote the fact that they have the same or different datum by introducing the symbols `$\sim$' and `$\not\sim$', and appending it at the end of the path. Then in this case we write that the elements must satisfy $\td[\texttt{a}] {\sim}$, $\td[\texttt{b}] {\sim}$, and $\td[\texttt{c}] {\sim}$. We then introduce the node expression
$\dcbracket{\alpha_1s_1, \ldots, \alpha_ns_n}$ where $\alpha_i$ is a path expression and $s_i \in \set{ \sim, \not\sim}$ for all $i \in [1..n]$. Semantically, it is a node expression that denotes all the tree positions $x$ from which we can access $n$ positions $x_1, \dotsc, x_n$ such that there exists $d \in \D$ where for all $i \in [n]$ the following holds:  $(x,x_i) \in \dbracket{\alpha_i}$; if $s_i = {\sim}$ then $\dD(x_i)=d$; and if $s_i = {\not\sim}$ then $\dD(x_i)\not=d$. Note that now we can express $\tup{\alpha = \beta}$ as $\dcbracket{\alpha {\sim}, \beta {\sim}}$ and $\tup{\alpha \not= \beta}$ as $\dcbracket{\alpha {\sim}, \beta {\not\sim}}$. Let us call $\rxpath^+\!(\mathfrak F,=)$ to $\rxpath(\mathfrak F,=)$ extended with the construction just explained. This is a more expressive formalism since the first mentioned property ---or, to give another example, $\dcbracket{\td[\texttt{a}]{\sim}, \td[\texttt{b}]{\sim}, \td[\texttt{a}]{\not\sim}, \td[\texttt{b}]{\not\sim}}$--- is not expressible in $\rxpath(\mathfrak F,=)$.

We argue that satisfiability for this extension can be decided in the same way as for $\rxpath(\mathfrak F,=)$. It is straightforward to see that \emph{positive} appearances can easily be translated with the help of the \opguess operator. On the other hand, for negative appearances, like $\lnot \dcbracket{\alpha_1s_1, \dotsc, \alpha_ns_n}$, we proceed in the same way as we did for $\rxpath(\mathfrak F,=)$. The only difference being that in this case the automaton will simulate the simultaneous evaluation of the $n$ expressions and calculate all possible configurations of the closest common ancestors of the endpoints, performing a \opspread at each of these intermediate points.
\begin{remark}
  Satisfiability of $\rxpath^+\!(\mathfrak F,=)$ under key constraints  and DTDs is decidable.
\end{remark}

\section{Concluding remarks}
\label{sec:concluding-remarks}

We presented a simplified framework to  work with one-way alternating register automata on data words and trees, enabling the possibility to easily show decidability of new operators by proving that they preserve the downward compatibility of a well-structured transition system. It would be interesting to hence investigate more decidable extensions, to study the expressiveness limits of decidable logics and automata for data trees.

Also, this work argues in favor of exploring computational models that although they might be not closed under all boolean operations, can serve to show decidability of logics closed under negation ---such as forward-\xpath--- or expressive natural extensions of existing logics 
 ---such as \ltlurxae.

We finally mention that even though $\xpath(\down,\td,\rightarrow,\tr,=)$ (\ie, forward-\xpath) has a non-primitive-recursive complexity, the results of \cite{Fig11} suggest that it seems plausible that $\xpath(\down,\td,\tr,=)$ or even $\xpath(\down,\td,\tr,\tl,=)$ are decidable in elementary time (see~\cite[Conjecture 1]{Fig11}).

\bibliographystyle{plain}

\begin{thebibliography}{}

\end{thebibliography}


\begin{thebibliography}{10}

\bibitem{ACJT-lics96}
Parosh~Aziz Abdulla, K{\=a}rlis {\v C}er{\=a}ns, Bengt Jonsson, and Yih-Kuen
  Tsay.
\newblock General decidability theorems for infinite-state systems.
\newblock In {\em Annual IEEE Symposium on Logic in Computer Science
  (LICS'96)}, pages 313--321, 1996.

\bibitem{ADOW05}
Parosh~Aziz Abdulla, Johann Deneux, Jo{\"e}l Ouaknine, and James Worrell.
\newblock Decidability and complexity results for timed automata via channel
  machines.
\newblock In {\em International Colloquium on Automata, Languages and
  Programming (ICALP'05)}, pages 1089--1101, 2005.

\bibitem{AD94}
Rajeev Alur and David~L. Dill.
\newblock A theory of timed automata.
\newblock {\em Theoretical Computer Science}, 126:183--235, 1994.

\bibitem{BFG08}
Michael Benedikt, Wenfei Fan, and Floris Geerts.
\newblock {XP}ath satisfiability in the presence of {DTD}s.
\newblock {\em Journal of the ACM}, 55(2):1--79, 2008.

\bibitem{BMSS09:xml:jacm}
Miko{\l}aj Boja{\'n}czyk, Anca Muscholl, Thomas Schwentick, and Luc Segoufin.
\newblock Two-variable logic on data trees and {XML} reasoning.
\newblock {\em Journal of the ACM}, 56(3):1--48, 2009.

\bibitem{CS-lics08}
Pierre Chambart and {\relax Ph}ilippe Schnoebelen.
\newblock The ordinal recursive complexity of lossy channel systems.
\newblock In {\em Annual IEEE Symposium on Logic in Computer Science
  (LICS'08)}, pages 205--216. {IEEE} Computer Society Press, 2008.

\bibitem{xpath:w3c}
James Clark and Steve DeRose.
\newblock {XML} path language ({XPath}).
\newblock Website, 1999.
\newblock {W3C Recommendation}. \url{http://www.w3.org/TR/xpath}.

\bibitem{DL-tocl08}
St{\'e}phane Demri and Ranko Lazi{\'c}.
\newblock {LTL} with the freeze quantifier and register automata.
\newblock {\em ACM Transactions on Computational Logic}, 10(3), 2009.

\bibitem{DLN05}
St{\'e}phane Demri, Ranko Lazi{\'c}, and David Nowak.
\newblock On the freeze quantifier in constraint {LTL}: Decidability and
  complexity.
\newblock In {\em International Symposium on Temporal Representation and
  Reasoning (TIME'05)}, pages 113--121. IEEE Computer Society Press, 2005.

\bibitem{dicksonslem}
Leonard~E. Dickson.
\newblock Finiteness of the odd perfect and primitive abundant numbers with n
  distinct prime factors.
\newblock {\em The American Journal of Mathematics}, 35(4):413--422, 1913.

\bibitem{F09}
Diego Figueira.
\newblock Satisfiability of downward {XP}ath with data equality tests.
\newblock In {\em ACM Symposium on Principles of Database Systems (PODS'09)},
  pages 197--206. ACM Press, 2009.

\bibitem{Fig10}
Diego Figueira.
\newblock Forward-{XP}ath and extended register automata on data-trees.
\newblock In {\em International Conference on Database Theory (ICDT'10)}. ACM
  Press, 2010.

\bibitem{FigPhD}
Diego Figueira.
\newblock {\em Reasoning on Words and Trees with Data}.
\newblock {P}h.{D}. thesis, Laboratoire Sp{\'e}cification et V{\'e}rification,
  ENS Cachan, France, December 2010.

\bibitem{Fig11}
Diego Figueira.
\newblock A decidable two-way logic on data words.
\newblock In {\em Annual IEEE Symposium on Logic in Computer Science
  (LICS'11)}. {IEEE} Computer Society Press, 2011.

\bibitem{FFSS10}
Diego Figueira, Santiago Figueira, Sylvain Schmitz, and {\relax Ph}ilippe
  Schnoebelen.
\newblock {A}ckermannian and primitive-recursive bounds with {D}ickson's lemma.
\newblock In {\em Annual IEEE Symposium on Logic in Computer Science
  (LICS'11)}. IEEE Computer Society Press, 2011.

\bibitem{FHL10}
Diego Figueira, Piotr Hofman, and S{\l}awomir Lasota.
\newblock Relating timed and register automata.
\newblock In {\em International Workshop on Expressiveness in Concurrency
  (EXPRESS'10)}, 2010.

\bibitem{FS09}
Diego Figueira and Luc Segoufin.
\newblock Future-looking logics on data words and trees.
\newblock In {\em International Symposium on Mathematical Foundations of
  Computer Science (MFCS'09)}, volume 5734 of {\em LNCS}, pages 331--343.
  Springer, 2009.

\bibitem{FS10}
Diego Figueira and Luc Segoufin.
\newblock Bottom-up automata on data trees and vertical {XP}ath.
\newblock In {\em International Symposium on Theoretical Aspects of Computer
  Science (STACS'11)}. Springer, 2011.

\bibitem{FS01}
Alain Finkel and {\relax Ph}ilippe Schnoebelen.
\newblock Well-structured transition systems everywhere!
\newblock {\em Theoretical Computer Science}, 256(1-2):63--92, 2001.

\bibitem{GF05}
Floris Geerts and Wenfei Fan.
\newblock Satisfiability of {XPath} queries with sibling axes.
\newblock In {\em International Symposium on Database Programming Languages
  (DBPL'05)}, volume 3774 of {\em Lecture Notes in Computer Science}, pages
  122--137. Springer, 2005.

\bibitem{GKP05}
Georg Gottlob, Christoph Koch, and Reinhard Pichler.
\newblock Efficient algorithms for processing {XPath} queries.
\newblock {\em ACM Transactions on Database Systems}, 30(2):444--491, 2005.

\bibitem{higmanslem}
Graham Higman.
\newblock Ordering by divisibility in abstract algebras.
\newblock {\em Proceedings of the London Mathematical Society (3)},
  2(7):326--336, 1952.

\bibitem{JL07}
Marcin Jurdzi{\'n}ski and Ranko Lazi{\'c}.
\newblock Alternation-free modal mu-calculus for data trees.
\newblock In {\em Annual IEEE Symposium on Logic in Computer Science
  (LICS'07)}, pages 131--140. IEEE Computer Society Press, 2007.

\bibitem{JL08}
Marcin Jurdzi{\'n}ski and Ranko Lazi{\'c}.
\newblock Alternating automata on data trees and {XP}ath satisfiability.
\newblock {\em ACM Transactions on Computational Logic}, 12(3):19, 2011.

\bibitem{Kr60}
Joseph~B. Kruskal.
\newblock Well-quasi-ordering, the tree theorem, and {V}azsonyi's conjecture.
\newblock {\em Transactions of the American Mathematical Society},
  95(2):210--225, 1960.

\bibitem{fast}
M.H. L\"ob and S.S. Wainer.
\newblock Hierarchies of number theoretic functions, {I}.
\newblock {\em Archiv f\"ur Mathematische Logik und Grundlagenforschung},
  13:39--51, 1970.

\bibitem{M04}
Maarten Marx.
\newblock {XP}ath with conditional axis relations.
\newblock In {\em International Conference on Extending Database Technology
  (EDBT'04)}, volume 2992 of {\em Lecture Notes in Computer Science}, pages
  477--494. Springer, 2004.

\bibitem{phs-mfcs2010}
{\relax Ph}ilippe Schnoebelen.
\newblock Revisiting {A}ckermann-hardness for lossy counter machines and reset
  {P}etri nets.
\newblock In {\em International Symposium on Mathematical Foundations of
  Computer Science ({MFCS}'10)}, volume 6281 of {\em Lecture Notes in Computer
  Science}, pages 616--628. Springer, 2010.

\bibitem{tC06}
Balder ten Cate.
\newblock The expressivity of {XP}ath with transitive closure.
\newblock In {\em ACM Symposium on Principles of Database Systems (PODS'06)},
  pages 328--337. ACM Press, 2006.

\bibitem{tCS08}
Balder ten Cate and Luc Segoufin.
\newblock {XP}ath, transitive closure logic, and nested tree walking automata.
\newblock In {\em ACM Symposium on Principles of Database Systems (PODS'08)},
  pages 251--260. ACM Press, 2008.

\end{thebibliography}

\end{document}